\documentclass[twocolumn,10pt]{IEEEtran}

\usepackage{amsmath,amssymb,amsthm}
\usepackage{graphicx}
\usepackage[acronym]{glossaries}
\usepackage{xcolor}
\usepackage{cases}
\usepackage{booktabs}
\usepackage{multirow}
\usepackage{algorithm}
\usepackage{algorithmic}

\newacronym{adc}{ADC}{analog-to-digital converter}
\newacronym{dac}{DAC}{digital-to-analog converter}
\newacronym{rf}{RF}{radio frequency}
\newacronym{milac}{MiLAC}{microwave linear analog computer}
\newacronym{svd}{SVD}{singular value decomposition}
\newacronym{em}{EM}{electromagnetic}
\newacronym{dft}{DFT}{discrete Fourier transform}
\newacronym{fft}{FFT}{fast Fourier transform}
\newacronym{sim}{SIM}{stacked intelligent metasurface}
\newacronym{espar}{ESPAR}{electronically steerable parasitic array radiator}
\newacronym{dsa}{DSA}{dynamic scattering array}
\newacronym{lmmse}{LMMSE}{linear minimum mean square error}
\newacronym{nn}{NN}{neural network}
\newacronym{ai}{AI}{artificial intelligence}
\newacronym{rtft}{RTFT}{real-time Fourier transform}
\newacronym{ris}{RIS}{reconfigurable intelligent surface}
\newacronym{mimo}{MIMO}{multiple-input multiple-output}
\newacronym{nmse}{NMSE}{normalized mean squared error}
\newacronym{vna}{VNA}{vector network analyzer}
\newacronym{cmos}{CMOS}{complementary metal-oxide-semiconductor}
\newacronym{snr}{SNR}{signal-to-noise ratio}
\newacronym{doa}{DoA}{direction of arrival}

\newtheorem{definition}{Definition}
\newtheorem{lemma}{Lemma}
\newtheorem{proposition}{Proposition}
\newtheorem{corollary}{Corollary}
\newtheorem{theorem}{Theorem}

\begin{document}
\bstctlcite{BSTcontrol}

\title{Analog Computing with\\Hybrid Couplers and Phase Shifters}

\author{Matteo~Nerini*,~\IEEEmembership{Senior~Member,~IEEE},
        Xuekang~Liu*,~\IEEEmembership{Member,~IEEE},
        Bruno~Clerckx,~\IEEEmembership{Fellow,~IEEE}

\thanks{* Contributed equally (Corresponding author: Xuekang Liu).}
\thanks{This work has been supported in part by UKRI under Grant EP/Y004086/1, EP/X040569/1, EP/Y037197/1, EP/X04047X/1, EP/Y037243/1.}
\thanks{The authors are with the Department of Electrical and Electronic Engineering, Imperial College London, SW7 2AZ London, U.K. (e-mail: \{m.nerini20, x.liu1, b.clerckx\}@imperial.ac.uk).}}

\maketitle

\begin{abstract}
Analog computing with microwave signals can enable exceptionally fast computations, potentially surpassing the limits of conventional digital computing.
For example, by letting some input signals propagate through a linear microwave network and reading the corresponding output signals, we can instantly compute a matrix-vector product without any digital operations.
In this paper, we investigate the computational capabilities of linear microwave networks made exclusively of two low-cost and fundamental components: hybrid couplers and phase shifters, which are both implementable in microstrip.
We derive a sufficient and necessary condition characterizing the class of linear transformations that can be computed in the analog domain using these two components.
Within this class, we identify three transformations of particular relevance to signal processing, namely the \gls{dft}, the Hadamard transform, and the Haar transform.
For each of these, we provide a systematic design method to construct networks of hybrid couplers and phase shifters capable of computing the transformation for any size power of two.
To validate our theoretical results, a hardware prototype was designed and fabricated, integrating hybrid couplers and phase shifters to implement the $4\times4$ \gls{dft}.
A systematic calibration procedure was subsequently developed to characterize the prototype and compensate for fabrication errors.
Measured results from the prototype demonstrate successful \gls{dft} computation in the analog domain, showing high correlation with theoretical expectations.
By realizing an analog computer through standard microwave components, this work demonstrates a practical pathway toward low-latency, real-time analog signal processing.
\end{abstract}

\glsresetall


\section{Introduction}

The growth of data-intensive applications is imposing increasingly stringent requirements on computing systems.
Conventional digital computers are fundamentally constrained by factors such as the clock frequency, costly analog-to-digital conversion, and high power consumption, and therefore are pushed to their limits \cite{rab02}.
As a promising alternative paradigm, analog computing has re-emerged as a hardware accelerator for computation-intensive linear algebraic operations.
By performing computation directly on the physical carrier of information, such as voltages and currents, analog computing can bypass the sampling-rate bottleneck imposed by \glspl{adc} and offer faster (with near-zero-latency) and more energy-efficient computations.
As a result, it is an attractive complement to digital processors for future communication and sensing systems \cite{cal13,sil14,zan21,iel18}.

An important form of analog computing leverages \gls{em} signals in the microwave regime, which are extensively used in communication systems and benefit from mature, well-characterized hardware.
Interestingly, microwave signals naturally interfere with each other during propagation, allowing the computation of specific operations ``for free'' in the analog domain.
Since microwave signals commonly propagate in linear media (e.g., air) or through linear components (e.g., transmission lines, resistors, capacitors, and inductors), in this work, we focus on exploiting linear microwave networks for analog computation.
Such networks designed for the purpose of computation have been denoted as \glspl{milac} \cite{ner25-1,ner25-2}.
A \gls{milac} can be abstracted as a black box, representing the linear microwave network, that receives input signals on some ports and returns output signals on other ports, as illustrated in Fig.~\ref{fig:milac-intro}.

\begin{figure}[t]
\centering
\includegraphics[width=0.36\textwidth]{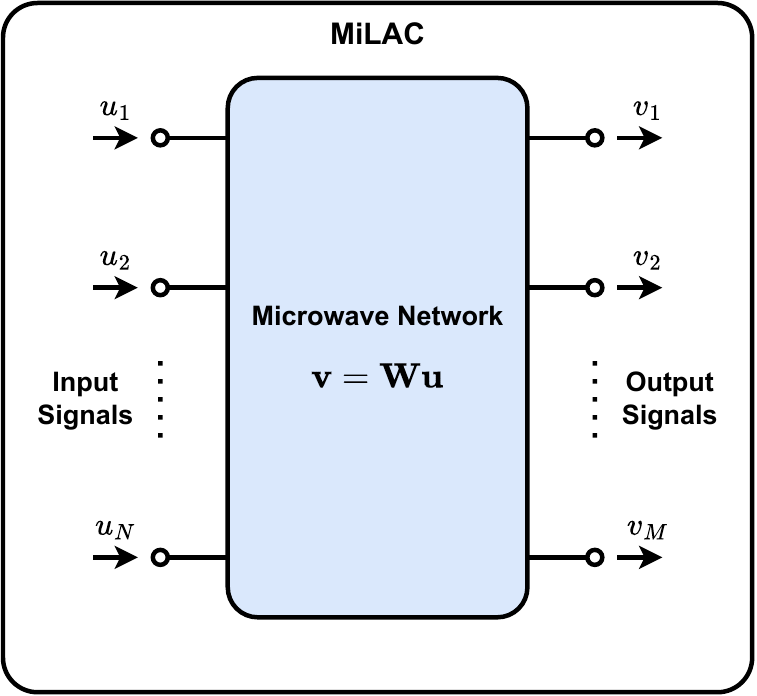}
\caption{High-level representation of a microwave linear analog computer.}
\label{fig:milac-intro}
\end{figure}

In general, say there are $N$ input and $M$ output signals, which are narrowband \gls{rf} signals.
We can write the $n$th input signal as $A_n\cos(\omega t+\varphi_n)$, with complex-valued baseband equivalent $u_n=A_ne^{j\varphi_n}\in\mathbb{C}$, and the $m$th output signal as $B_m\cos(\omega t+\psi_m)$, with complex-valued baseband equivalent $v_m=B_me^{j\psi_m}\in\mathbb{C}$.
Since the microwave network of a \gls{milac} is linear, the output vector $\mathbf{v}=[v_1,\ldots,v_M]^T$ will always be a linear function of the input vector $\mathbf{u}=[u_1,\ldots,u_N]^T$, i.e., $\mathbf{v}=\mathbf{W}\mathbf{u}$, where $\mathbf{W}\in\mathbb{C}^{M\times N}$ is a matrix which depends on the properties of the microwave network.
Therefore, the advantage of such analog computers is to compute the matrix-vector product $\mathbf{W}\mathbf{u}$ in the analog domain, as the signals propagate through the network at light speed.
Note that computing this matrix-vector product on a digital computer would require a number of operations growing with $NM$, i.e., a computational complexity $\mathcal{O}(NM)$.
Besides, no digital operation is required by a \gls{milac}, leading to a computational complexity of $\mathcal{O}(1)$.
This can significantly accelerate matrix-vector products in wireless multi-antenna systems (e.g., communications or radar sensing), where microwave signals can be linearly precoded or combined directly in the analog domain.

While the input vector $\mathbf{u}$ can be arbitrarily set by shaping the input microwave signals, the matrix $\mathbf{W}$ is constrained depending on the specific implementation of the microwave network.
In this work, we investigate the computational capabilities of microwave networks realized exclusively with two low-cost and fundamental linear components: hybrid couplers and phase shifters.
Our focus on these components is motivated by two considerations.
First, hybrid couplers are arguably the simplest components capable of splitting or combining two signals in a controlled manner, as they can be implemented in microstrip technology.
Second, phase shifters, realized as microstrip transmission lines, are the most basic components that enable controlled signal propagation.
These two components are well-known for their use in realizing the Butler matrix, a beamforming network that drives phased arrays and controls beam direction based on the selected input port \cite{moo64}.
Beyond the Butler matrix, our goal here is to explore what other analog computations can be performed using the same two building blocks.
Specifically, we want to answer the two questions:
\textit{For which matrices $\mathbf{W}$ can the product $\mathbf{W}\mathbf{u}$ be computed in the analog domain using only hybrid couplers and phase shifters?
And among these matrices, are there any of practical usefulness?}
This work provides exhaustive answers to both questions, as detailed in the following contributions.

\textit{First}, we characterize all the matrices $\mathbf{W}$ for which the product $\mathbf{W}\mathbf{u}$ can be computed in the analog domain using exclusively hybrid couplers and phase shifters.
This is achieved by providing a sufficient and necessary condition for such matrices.
Such a condition on $\mathbf{W}$ is that it can be decomposed as a product of $2L+1$ matrices, being alternatively permutation matrices and block diagonal matrices with block sizes $2\times2$ or $1\times1$, with special requirements on those blocks.
This main result is presented in Section~\ref{sec:main} and proved in Section~\ref{sec:proof}.

\textit{Second}, we analytically show that exists a microwave network made of hybrid couplers and phase shifters can compute the \gls{dft} of its input signal $\mathbf{u}$.\footnote{Note that the analog-domain \gls{dft} mentioned here is $\mathbf{F}\mathbf{u}$, where $\mathbf{F}$ is the \gls{dft} matrix, which is different from the related concept of \gls{rtft} \cite{cal13,yan25}.
While the \gls{dft} $\mathbf{F}\mathbf{u}$ operates on the baseband equivalent representations of multiple microwave signals $u_1,\ldots,u_N$, the \gls{rtft} performs on-the-fly the Fourier transform of a microwave signal $u(t)$, returning an output signal $U(t)$ whose shape in time corresponds to the Fourier transform of $u(t)$.}
This is established by showing that the \gls{dft} matrix can be decomposed as required by the sufficient and necessary condition on $\mathbf{W}$ of the first contribution (Section~\ref{sec:dft}).
Furthermore, we derive a systematic procedure to design the network of hybrid couplers and phase shifters that computes the $N\times N$ \gls{dft}, for any $N$ power of two.
We also provide four examples of such networks that compute the the $N\times N$ \gls{dft} with $N\in\{2,4,8,16\}$ (Section~\ref{sec:examples}).

\textit{Third}, we analytically show that there exists a microwave network made of hybrid couplers and phase shifters that can compute the $N\times N$ Hadamard transform of its input signal $\mathbf{u}$, for any $N$ power of two (Section~\ref{sec:hada}).
Similar to what is done for the \gls{dft}, we also provide a procedure to systematically design such networks for any valid transform size $N$.

\textit{Fourth}, we analytically demonstrate that there exists a microwave network composed of hybrid couplers and phase shifters that can realize the $N\times N$ Haar transform of an input signal $\mathbf{u}$, for any $N$ power of two (Section~\ref{sec:haar}).
We also present a systematic design procedure enabling the construction of such networks for any valid size $N$.
Applications of the proposed analog \gls{dft}, Hadamard transform, and Haar transform are provided in Section~\ref{sec:app}.

\textit{Fifth}, we fabricate a hardware prototype to validate the presented analytical results (Section~\ref{sec:experiment}).
Specifically, we implement in microstrip technology a network of hybrid couplers and phase shifters that computes the $4\times4$ \gls{dft} of its input signals.
We also develop a calibration procedure to mitigate fabrication imperfections.
Experimental measurements confirm successful analog-domain \gls{dft} computation, exhibiting strong agreement with theoretical predictions.

We review related literature in Section~\ref{sec:work}, and conclude this work and present directions for future research in Section~\ref{sec:conclusion}.

\textit{Notation}: Vectors and matrices are denoted with bold lower and bold upper letters, respectively.
Scalars are represented with letters not in bold font.
$\mathbf{a}^T$, $[\mathbf{a}]_{i}$, and $\Vert\mathbf{a}\Vert$ refer to the transpose, the $i$th element, and the $\ell_2$-norm of a vector $\mathbf{a}$, respectively.
$\mathbf{A}^T$, $[\mathbf{A}]_{i,k}$, and $\Vert\mathbf{A}\Vert_F$ refer to the transpose, the $(i,k)$th element, and the Frobenius norm of a matrix $\mathbf{A}$, respectively.
$\mathbb{N}$ and $\mathbb{C}$ denote the natural and complex number sets, respectively.
$j=\sqrt{-1}$ denotes the imaginary unit.
$\mathbf{0}_{N\times M}$, $\mathbf{0}_{N}$, and $\mathbf{I}_{N}$ denote the all-zero $N\times M$ matrix, the all-zero $N\times N$ matrix, and the $N\times N$ identity matrix, respectively.
diag$(a_1,\ldots,a_N)$ refers to a diagonal matrix with diagonal entries being $a_1,\ldots,a_N$, while diag$(\mathbf{A}_1,\ldots,\mathbf{A}_N)$ refers to a block diagonal matrix with blocks being $\mathbf{A}_1,\ldots,\mathbf{A}_N$.
$\mathbf{A}\otimes\mathbf{B}$ denotes the Kronecker product of matrices $\mathbf{A}$ and $\mathbf{B}$.

\section{System Model}
\label{sec:model}

The goal of this paper is to determine what computations can be performed in the analog domain by a network of hybrid couplers and phase shifters.
Therefore, we begin by explaining how a generic linear microwave network performs computations by processing signals at its input ports and producing corresponding signals at its output ports.
In this section, we review the fundamentals of \glspl{milac} \cite{ner25-1,ner25-2}.

Consider a multiport microwave network having $2N$ ports, where the first $N$ are designed as input ports and the second $N$ as output ports, as depicted in Fig.~\ref{fig:milac}.
While the number of input and output ports can be different in general, we only consider microwave networks having an equal number of input and output ports.
This is a property of all networks made exclusively of hybrid couplers and phase shifters, as it will be clear in Section~\ref{sec:main}.
The input is applied by $N$ voltage sources with their series impedance $Z_0$, e.g., $Z_0=50~\Omega$, connected to the input ports of the microwave network.
Denoting as $u_{n}\in\mathbb{C}$ the complex-valued baseband equivalent of the voltage of the $n$th source, the input vector is given by $\mathbf{u}=[u_{1},\ldots,u_{N}]^T\in\mathbb{C}^{N\times1}$.
The output is given by the voltages at the other $N$ ports, which are read with probes with input impedance $Z_0$.
For a linear microwave network, the output will always be a linear function of the input $\mathbf{u}$, depending on the properties of the network.
Interestingly, such a linear function can be computed in the analog domain with no digital operations required by applying the desired input voltages and reading the corresponding output voltages.

\begin{figure}[t]
\centering
\includegraphics[width=0.44\textwidth]{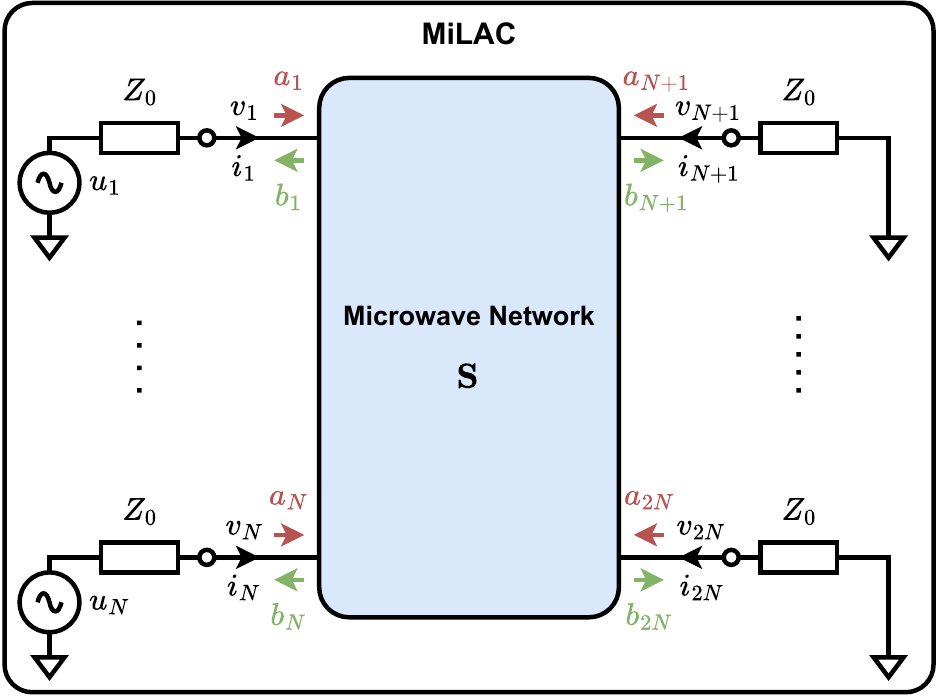}
\caption{Model of a microwave linear analog computer.}
\label{fig:milac}
\end{figure}

In the following, we introduce the quantities required to determine the expression of the linear mapping between the input and output of a \gls{milac}.
At the $n$th \gls{milac} port, we denote the voltage as $v_n\in\mathbb{C}$, the current as $i_n\in\mathbb{C}$ (with direction entering the network), the incident wave as $a_n\in\mathbb{C}$, and the reflected wave as $b_n\in\mathbb{C}$.
All these signals are complex-valued baseband equivalents of the corresponding narrowband \gls{rf} signals.
We collect these quantities at the $2N$ ports into four vectors
$\mathbf{v}=[v_1,\ldots,v_{2N}]^T\in\mathbb{C}^{2N\times 1}$, $\mathbf{i}=[i_1,\ldots,i_{2N}]^T\in\mathbb{C}^{2N\times 1}$, $\mathbf{a}=[a_1,\ldots,a_{2N}]^T\in\mathbb{C}^{2N\times 1}$, $\mathbf{b}=[b_1,\ldots,b_{2N}]^T\in\mathbb{C}^{2N\times 1}$, which are partitioned as
\begin{equation}
\mathbf{v}=\begin{bmatrix}\mathbf{v}_{1}\\\mathbf{v}_{2}\end{bmatrix},\;
\mathbf{i}=\begin{bmatrix}\mathbf{i}_{1}\\\mathbf{i}_{2}\end{bmatrix},\;
\mathbf{a}=\begin{bmatrix}\mathbf{a}_{1}\\\mathbf{a}_{2}\end{bmatrix},\;
\mathbf{b}=\begin{bmatrix}\mathbf{b}_{1}\\\mathbf{b}_{2}\end{bmatrix},
\end{equation}
where $\mathbf{v}_{1},\mathbf{i}_{1},\mathbf{a}_{1},\mathbf{b}_{1}\in\mathbb{C}^{N\times 1}$ are the electrical quantities at the input ports and $\mathbf{v}_{2},\mathbf{i}_{2},\mathbf{a}_{2},\mathbf{b}_{2}\in\mathbb{C}^{N\times 1}$ at the output ports.
The relationships between them are given by
\begin{equation}
\mathbf{v}=\mathbf{a}+\mathbf{b},\;
\mathbf{i}=\frac{\mathbf{a}-\mathbf{b}}{Z_0},
\end{equation}
according to multiport network theory \cite[Chapter~4]{poz12}.
Furthermore, the $2N$-port \gls{milac} network can be characterized by its scattering matrix $\mathbf{S}\in\mathbb{C}^{2N\times 2N}$, partitioned as
\begin{equation}
\mathbf{S}=
\begin{bmatrix}
\mathbf{S}_{11} & \mathbf{S}_{12}\\
\mathbf{S}_{21} & \mathbf{S}_{22}
\end{bmatrix},
\end{equation}
where $\mathbf{S}_{XY}\in\mathbb{C}^{N\times N}$, for $X,Y\in\{1,2\}$, which relates the incident and reflected waves $\mathbf{a}$ and $\mathbf{b}$ through $\mathbf{b}=\mathbf{S}\mathbf{a}$,
by definition of scattering matrix \cite[Chapter~4]{poz12}.

We now have all the tools to derive the expression of the output $\mathbf{v}_{2}$ as a function of the input $\mathbf{u}$ and the scattering matrix $\mathbf{S}$ of the \gls{milac} network.
At the input ports, $\mathbf{v}_1$ and $\mathbf{i}_1$ are related by $\mathbf{v}_1=\mathbf{u}-Z_0\mathbf{i}_1$, according to Ohm's law.
In addition, by substituting $\mathbf{v}_1=\mathbf{a}_1+\mathbf{b}_1$ and $\mathbf{i}_1=(\mathbf{a}_1-\mathbf{b}_1)/Z_0$ into $\mathbf{v}_1=\mathbf{u}-Z_0\mathbf{i}_1$, we obtain
$\mathbf{a}_1=\mathbf{u}/2$.
At the output ports, $\mathbf{v}_2$ and $\mathbf{i}_2$ are related by $\mathbf{v}_2=-Z_0\mathbf{i}_2$, by Ohm's law.
Hence, by substituting $\mathbf{v}_2=\mathbf{a}_2+\mathbf{b}_2$ and $\mathbf{i}_2=(\mathbf{a}_2-\mathbf{b}_2)/Z_0$ into $\mathbf{v}_2=-Z_0\mathbf{i}_2$, we obtain
$\mathbf{a}_2=\mathbf{0}_{N\times1}$, and $\mathbf{b}_2=\mathbf{v}_2$.
The physical meaning of $\mathbf{a}_2=\mathbf{0}_{N\times1}$ is that there are no reflected waves from the probes on the output ports, which is expected since we assume their input impedance to be matched to $Z_0$.
By using $\mathbf{a}_2=\mathbf{0}_{N\times1}$, the relationship $\mathbf{b}=\mathbf{S}\mathbf{a}$ gives
\begin{equation}
\mathbf{b}_2=\mathbf{S}_{21}\mathbf{a}_1.
\end{equation}
By further applying $\mathbf{b}_2=\mathbf{v}_2$ and $\mathbf{a}_1=\mathbf{u}/2$, we finally obtain
\begin{equation}
\mathbf{v}_2=\frac{1}{2}\mathbf{S}_{21}\mathbf{u},
\end{equation}
indicating that the output $\mathbf{v}_2$ is a linear function of the input $\mathbf{u}$ depending on the block $\mathbf{S}_{21}$ of the scattering matrix of the \gls{milac} network.
Note that this has been obtained with no assumptions on the \gls{milac} network, except linearity, which is an assumption required by the scattering matrix representation.

\section{Main Result}
\label{sec:main}

We have shown that a \gls{milac} can process linearly an input $\mathbf{u}$ and return $\mathbf{v}_2=\mathbf{S}_{21}\mathbf{u}/2$, depending of the block $\mathbf{S}_{21}$ of its scattering matrix.
Since our goal is to characterize what functions can be computed with a \gls{milac} made of hybrid couplers and phase shifters, this means that we want to characterize all the possible blocks $\mathbf{S}_{21}$ of such a \gls{milac}.
To characterize all such blocks $\mathbf{S}_{21}$, we adopt the same approach that Claude Shannon used to characterize all systems of differential equations solvable by a differential analyzer, a famous mechanical analog computer conceived by Lord Kelvin and first constructed by Vannevar Bush \cite{sha41}.
We first provide a precise mathematical definition of all available components.
Next, we define how these components can be interconnected.
Finally, we derive a necessary and sufficient condition under which a function can be computed with a \gls{milac} composed of hybrid couplers and phase shifters.
We begin by formally introducing hybrid couplers and phase shifters, along with two additional components required to interconnect them to each other.
These four components are introduced by presenting their scattering matrix, as this is sufficient to fully characterize their behavior.

\begin{figure}[t]
\centering
\includegraphics[width=0.36\textwidth]{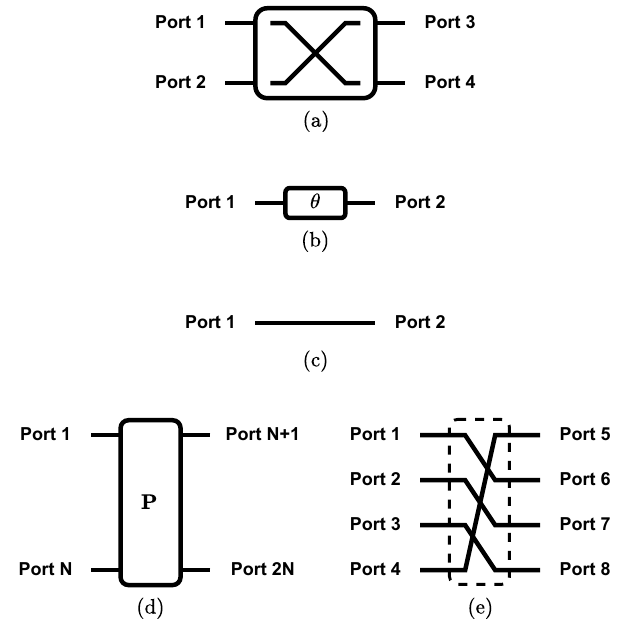}
\caption{Symbols of (a) a hybrid coupler, (b) a phase shifter, (c) an interconnection, (d) a permutation network, and (e) implementation of the permutation network having permutation matrix $\mathbf{P}$ given by \eqref{eq:Pexample}.}
\label{fig:comp}
\end{figure}

Hybrid couplers and phase shifters are defined as follows.
\begin{definition}
(Hybrid coupler)
We define a hybrid coupler as a $4$-port network with scattering matrix $\mathbf{S}\in\mathbb{C}^{4\times 4}$ given by
\begin{equation}
\mathbf{S}=\frac{1}{\sqrt{2}}
\begin{bmatrix}
0 & 0 & j & 1\\
0 & 0 & 1 & j\\
j & 1 & 0 & 0\\
1 & j & 0 & 0
\end{bmatrix}.
\end{equation}
\end{definition}
According to this definition, ports 1 and 2 are the input ports of the hybrid coupler while ports 3 and 4 are the output ports, as represented in Fig.~\ref{fig:comp}(a).
The signal reflected at port 3 $b_3\in\mathbb{C}$ is a combination of the incident signals at port 1 $a_1\in\mathbb{C}$ and port 2 $a_2\in\mathbb{C}$, given by $b_3=(ja_1+a_2)/\sqrt{2}$.
Similarly, the signal reflected at port 4 $b_4\in\mathbb{C}$ is a combination of the incident signals at ports 1 and 2, as $b_4=(a_1+ja_2)/\sqrt{2}$.
Hence, the role of a hybrid coupler is to split the signals at the input ports into two components (the in-phase and quadrature components) and return them on the output ports.
\begin{definition}
(Phase shifter)
We define a phase shifter as a $2$-port network with scattering matrix $\mathbf{S}\in\mathbb{C}^{2\times 2}$ given by
\begin{equation}
\mathbf{S}=
\begin{bmatrix}
0 & e^{j\theta}\\
e^{j\theta} & 0
\end{bmatrix},
\end{equation}
where $\theta\in[0,2\pi)$ is the phase shift of the phase shifter.
\end{definition}
Following this definition, port 1 is the input port of the phase shifter while port 2 is the output port, as represented in Fig.~\ref{fig:comp}(b), and the role of a phase shifter is to shift the input signal by a phase $\theta$ and return it on the output port.\footnote{Recall that negative values of the phase shift $\theta$ are mathematically equivalent to their corresponding positive values modulo $2\pi$, which are more commonly used in transmission line implementations.}
Both hybrid couplers and phase shifters have scattering matrices that are unitary and symmetric, which means that they are lossless and reciprocal networks \cite[Chapter~4]{poz12}.

In addition to hybrid couplers and phase shifters, we also formally introduce the concepts of interconnection and permutation network, needed to interconnect the ports of hybrid couplers and phase shifters to each other. 
\begin{definition}
(Interconnection)
We define an interconnection as a $2$-port network with scattering matrix $\mathbf{S}$ given by
\begin{equation}
\mathbf{S}=
\begin{bmatrix}
0 & 1\\
1 & 0
\end{bmatrix}.
\end{equation}
\end{definition}
An interconnection, represented in Fig.~\ref{fig:comp}(c), can be used to interconnect the output of a component (hybrid coupler or phase shifter) to the input of another component.
In the following analysis, we will regard the interconnection as a third component that can be used to implement a \gls{milac}, since this simplifies the theoretical analysis.
Nevertheless, an interconnection can merely be seen as a special case of a phase shifter, where the imposed phase shift is $\theta=0$.
A set of $N$ interconnections between $N$ ports and other $N$ ports can be regarded as a permutation network, defined as follows.
\begin{definition}
(Permutation network)
We define a permutation network with permutation matrix $\mathbf{P}\in\{0,1\}^{N\times N}$ as a $2N$-port network with scattering matrix $\mathbf{S}\in\mathbb{C}^{2N\times 2N}$ given by
\begin{equation}
\mathbf{S}=
\begin{bmatrix}
\mathbf{0}_N & \mathbf{P}^T\\
\mathbf{P} & \mathbf{0}_N
\end{bmatrix}.
\end{equation}
\end{definition}
In a permutation network, port $i$ and port $N+k$ are interconnected if and only if $[\mathbf{P}]_{i,k}=1$, for $i,k=1,\ldots,N$, otherwise they are disconnected.
Hence, a permutation network only reorders its input signals without performing any operation on them.
The symbol of a generic permutation network is represented in Fig.~\ref{fig:comp}(d), while an example of permutation network is illustrated in Fig.~\ref{fig:comp}(e), with permutation matrix
\begin{equation}
\mathbf{P}=
\begin{bmatrix}
0 & 0 & 0 & 1\\
1 & 0 & 0 & 0\\
0 & 1 & 0 & 0\\
0 & 0 & 1 & 0
\end{bmatrix}.\label{eq:Pexample}
\end{equation}
It follows from the definitions of interconnection and permutation network that they are lossless and reciprocal components since their scattering matrices are unitary and symmetric.

Given these four components, more complicated networks can be constructed by interconnecting them in series or in parallel.
We regard two components as connected in series when the output ports of the first are connected to the input ports of the second.
The series between two networks having scattering matrices $\mathbf{Q}\in\mathbb{C}^{2N\times 2N}$ and $\mathbf{R}\in\mathbb{C}^{2N\times 2N}$ is visually shown in Fig.~\ref{fig:series}.
Besides, two components are connected in parallel when their input and output ports become the input and output ports of a larger network, respectively.
The parallel between two networks having scattering matrices $\mathbf{Q}\in\mathbb{C}^{2N\times 2N}$ and $\mathbf{R}\in\mathbb{C}^{2M\times 2M}$ is visually shown in Fig.~\ref{fig:parallel}.
Our goal is hence to characterize all the possible networks that can be constructed by connecting an arbitrarily large, but finite, number of components (hybrid couplers, phase shifters, interconnections, and permutation networks).
These components can be either connected in series or in parallel with each other.
We refer to these networks as \textit{network implementable with hybrid couplers and phase shifters}, formally defined through the following inductive definition.
\begin{definition}
A network is implementable with hybrid couplers and phase shifters if any of the following statements are true
\begin{enumerate}
\item The network is a hybrid coupler,
\item The network is a phase shifter,
\item The network is an interconnection,
\item The network is a permutation network,
\item The network is a series of two networks, both implementable with hybrid couplers and phase shifters,
\item The network is a parallel of two networks, both implementable with hybrid couplers and phase shifters.
\end{enumerate}
\end{definition}

\begin{figure}[t]
\centering
\includegraphics[width=0.36\textwidth]{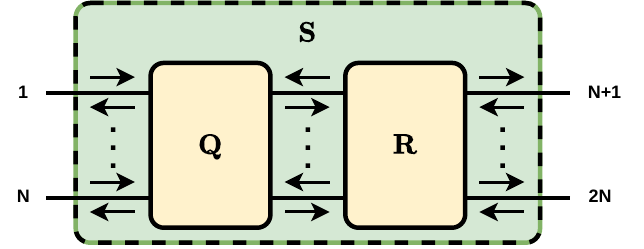}
\caption{Series of two networks having scattering matrices $\mathbf{Q}$ and $\mathbf{R}$.}
\label{fig:series}
\end{figure}

\begin{figure}[t]
\centering
\includegraphics[width=0.30\textwidth]{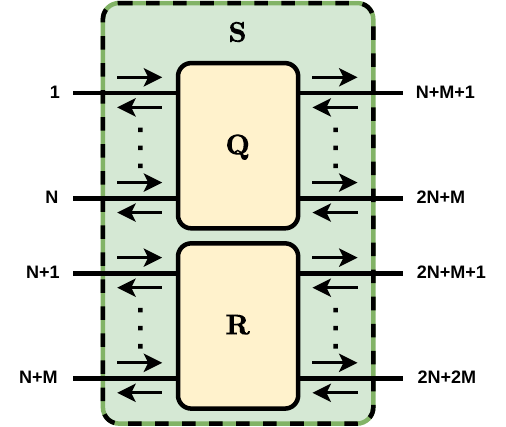}
\caption{Parallel of two networks having scattering matrices $\mathbf{Q}$ and $\mathbf{R}$.}
\label{fig:parallel}
\end{figure}

Characterizing all possible networks implementable with hybrid couplers and phase shifters means deriving a general expression for their scattering matrices.
Hence, we introduce the following two propositions that provide the expressions of the scattering matrices of the series and parallel of two networks with given scattering matrices.
\begin{proposition}
(Series of networks)
Consider a series system consisting of a first $2N$-port network with scattering matrix $\mathbf{Q}\in\mathbb{C}^{2N\times 2N}$ and a second $2N$-port network with scattering matrix $\mathbf{R}\in\mathbb{C}^{2N\times 2N}$, partitioned as
\begin{equation}
\mathbf{Q}=
\begin{bmatrix}
\mathbf{Q}_{11} & \mathbf{Q}_{12}\\
\mathbf{Q}_{21} & \mathbf{Q}_{22}
\end{bmatrix},\;
\mathbf{R}=
\begin{bmatrix}
\mathbf{R}_{11} & \mathbf{R}_{12}\\
\mathbf{R}_{21} & \mathbf{R}_{22}
\end{bmatrix},\label{eq:QR}
\end{equation}
where $\mathbf{Q}_{XY},\mathbf{R}_{XY}\in\mathbb{C}^{N\times N}$, for $X,Y\in\{1,2\}$.
In this series, the last $N$ ports of the first network are individually connected to the first $N$ ports of the second network.
Thus, the whole series network can be regarded as a $2N$-port network, with scattering matrix $\mathbf{S}\in\mathbb{C}^{2N\times 2N}$ partitioned as
\begin{equation}
\mathbf{S}=
\begin{bmatrix}
\mathbf{S}_{11} & \mathbf{S}_{12}\\
\mathbf{S}_{21} & \mathbf{S}_{22}
\end{bmatrix},\label{eq:S}
\end{equation}
where
\begin{gather}
\mathbf{S}_{11}=\mathbf{Q}_{11}+\mathbf{Q}_{12}\left(\mathbf{I}_N-\mathbf{R}_{11}\mathbf{Q}_{22}\right)^{-1}\mathbf{R}_{11}\mathbf{Q}_{21},\\
\mathbf{S}_{12}=\mathbf{Q}_{12}\left(\mathbf{I}_N-\mathbf{R}_{11}\mathbf{Q}_{22}\right)^{-1}\mathbf{R}_{12},\\
\mathbf{S}_{21}=\mathbf{R}_{21}\left(\mathbf{I}_N-\mathbf{Q}_{22}\mathbf{R}_{11}\right)^{-1}\mathbf{Q}_{21},\\
\mathbf{S}_{22}=\mathbf{R}_{22}+\mathbf{R}_{21}\left(\mathbf{I}_N-\mathbf{Q}_{22}\mathbf{R}_{11}\right)^{-1}\mathbf{Q}_{22}\mathbf{R}_{12}.
\end{gather}
\label{pro:ser}
\end{proposition}
\begin{proof}
Please, refer to Appendix~A.
\end{proof}
\begin{proposition}
(Parallel of networks)
Consider a parallel system consisting of a first $2N$-port network with scattering matrix $\mathbf{Q}\in\mathbb{C}^{2N\times 2N}$ partitioned as in \eqref{eq:QR}, where $\mathbf{Q}_{XY}\in\mathbb{C}^{N\times N}$, for $X,Y\in\{1,2\}$, and a second $2M$-port network with scattering matrix $\mathbf{R}\in\mathbb{C}^{2M\times 2M}$ partitioned as in \eqref{eq:QR}, where $\mathbf{R}_{XY}\in\mathbb{C}^{M\times M}$, for $X,Y\in\{1,2\}$.
Thus, the whole parallel network can be regarded as a $(2N+2M)$-port network, with scattering matrix $\mathbf{S}\in\mathbb{C}^{(2N+2M)\times (2N+2M)}$ partitioned as in \eqref{eq:S}, where
\begin{gather}
\mathbf{S}_{11}=
\begin{bmatrix}
\mathbf{Q}_{11} & \mathbf{0}_{N\times M}\\
\mathbf{0}_{M\times N} & \mathbf{R}_{11}
\end{bmatrix},\;
\mathbf{S}_{12}=
\begin{bmatrix}
\mathbf{Q}_{12} & \mathbf{0}_{N\times M}\\
\mathbf{0}_{M\times N} & \mathbf{R}_{12}
\end{bmatrix},\\
\mathbf{S}_{21}=
\begin{bmatrix}
\mathbf{Q}_{21} & \mathbf{0}_{N\times M}\\
\mathbf{0}_{M\times N} & \mathbf{R}_{21}
\end{bmatrix},\;
\mathbf{S}_{22}=
\begin{bmatrix}
\mathbf{Q}_{22} & \mathbf{0}_{N\times M}\\
\mathbf{0}_{M\times N} & \mathbf{R}_{22}
\end{bmatrix}.
\end{gather}
\label{pro:par}
\end{proposition}
\begin{proof}
Please, refer to Appendix~B.
\end{proof}

Proposition~\ref{pro:ser} and \ref{pro:par} show that the scattering matrix of the series or parallel of two networks can be a quite complicated function of the scattering matrices of the two networks.
This makes our analysis difficult.
Nevertheless, we notice that all four considered components fulfill a special property, which is preserved by their series or parallel combinations.
To formalize such a property, we introduce the concepts of \textit{matched network} and its \textit{transmission scattering matrix}.
\begin{definition}
(Matched network and transmission scattering matrix)
A $2N$-port network is referred to as matched if its scattering matrix $\mathbf{S}\in\mathbb{C}^{2N\times 2N}$ can be partitioned as
\begin{equation}
\mathbf{S}=
\begin{bmatrix}
\mathbf{0}_N & \mathbf{S}_{12}\\
\mathbf{S}_{21} & \mathbf{0}_N
\end{bmatrix},\label{eq:matched}
\end{equation}
where $\mathbf{S}_{XY}\in\mathbb{C}^{N\times N}$, for $X,Y\in\{1,2\}$.
For a matched network with scattering matrix partitioned as in \eqref{eq:matched}, we define its transmission scattering matrix $\bar{\mathbf{S}}\in\mathbb{C}^{N\times N}$ as $\bar{\mathbf{S}}=\mathbf{S}_{21}$.
\end{definition}
The concept of a matched network is particularly relevant to us since all our components, i.e., hybrid couplers, phase shifters, interconnections, and permutation networks, are matched networks by definition.
First, a hybrid coupler is a matched network with transmission scattering matrix $\bar{\mathbf{S}}\in\mathbb{C}^{2\times 2}$ given by $\bar{\mathbf{S}}=[[j,1]^T,[1,j]^T]/\sqrt{2}$.
Second, a phase shifter is a matched network with transmission scattering matrix $\bar{s}\in\mathbb{C}$ given by $\bar{s}=e^{j\theta}$.
Third, an interconnection is a matched network with transmission scattering matrix $\bar{s}=1$.
Fourth, a permutation network is a matched network with transmission scattering matrix $\bar{\mathbf{S}}=\mathbf{P}$, where $\mathbf{P}$ is its permutation matrix.
Recalling that any reciprocal network has a symmetric scattering matrix, a reciprocal matched network is fully characterized by its transmission scattering matrix since $\mathbf{S}_{12}=\mathbf{S}_{21}^T$.
Hence, the rationale of introducing the concept of transmission scattering matrix is to greatly simplify the computations without losing any information.

To clarify how the computations of scattering matrices simplify for matched networks, we introduce the following two corollaries.
\begin{corollary}
(Series of matched networks)
Consider the same series system as in Proposition~\ref{pro:ser}, and assume that the two networks are matched, i.e., $\mathbf{Q}_{11}=\mathbf{Q}_{22}=\mathbf{R}_{11}=\mathbf{R}_{22}=\mathbf{0}_N$.
In this case, we have
\begin{gather}
\mathbf{S}_{11}=\mathbf{0}_N,\;\mathbf{S}_{12}=\mathbf{Q}_{12}\mathbf{R}_{12},\\
\mathbf{S}_{21}=\mathbf{R}_{21}\mathbf{Q}_{21},\;\mathbf{S}_{22}=\mathbf{0}_N.
\end{gather}
Thus, the resulting series network is matched, and has a transmission scattering matrix
\begin{equation}
\bar{\mathbf{S}}=\bar{\mathbf{R}}\bar{\mathbf{Q}},
\end{equation}
where $\bar{\mathbf{Q}}=\mathbf{Q}_{21}$ and $\bar{\mathbf{R}}=\mathbf{R}_{21}$ are the transmission scattering matrices of the two networks in series.
\label{cor:ser}
\end{corollary}
\begin{proof}
This result follows directly by substituting $\mathbf{Q}_{11}=\mathbf{Q}_{22}=\mathbf{R}_{11}=\mathbf{R}_{22}=\mathbf{0}_N$ in Proposition~\ref{pro:ser}, and by applying the definitions of matched network and transmission scattering matrix.
\end{proof}
\begin{corollary}
(Parallel of matched networks)
Consider the same parallel system as in Proposition~\ref{pro:par}, and assume that the two networks are matched, i.e., $\mathbf{Q}_{11}=\mathbf{Q}_{22}=\mathbf{0}_N$ and $\mathbf{R}_{11}=\mathbf{R}_{22}=\mathbf{0}_M$.
In this case, we have
\begin{gather}
\mathbf{S}_{11}=\mathbf{0}_{N+M},\;
\mathbf{S}_{12}=
\begin{bmatrix}
\mathbf{Q}_{12} & \mathbf{0}_{N\times M}\\
\mathbf{0}_{M\times N} & \mathbf{R}_{12}
\end{bmatrix},\\
\mathbf{S}_{21}=
\begin{bmatrix}
\mathbf{Q}_{21} & \mathbf{0}_{N\times M}\\
\mathbf{0}_{M\times N} & \mathbf{R}_{21}
\end{bmatrix},\;
\mathbf{S}_{22}=\mathbf{0}_{N+M}.
\end{gather}
Thus, the resulting parallel network is matched, and has a transmission scattering matrix
\begin{equation}
\bar{\mathbf{S}}=
\begin{bmatrix}
\bar{\mathbf{Q}} & \mathbf{0}_{N\times M}\\
\mathbf{0}_{M\times N} & \bar{\mathbf{R}}
\end{bmatrix},
\end{equation}
where $\bar{\mathbf{Q}}=\mathbf{Q}_{21}$ and $\bar{\mathbf{R}}=\mathbf{R}_{21}$ are the transmission scattering matrices of the two networks in parallel.
\label{cor:par}
\end{corollary}
\begin{proof}
This result directly follows by substituting $\mathbf{Q}_{11}=\mathbf{Q}_{22}=\mathbf{0}_N$ and $\mathbf{R}_{11}=\mathbf{R}_{22}=\mathbf{0}_M$ in Proposition~\ref{pro:par}, and by applying the definitions of matched network and transmission scattering matrix.
\end{proof}

\begin{figure*}[t]
\centering
\includegraphics[width=0.96\textwidth]{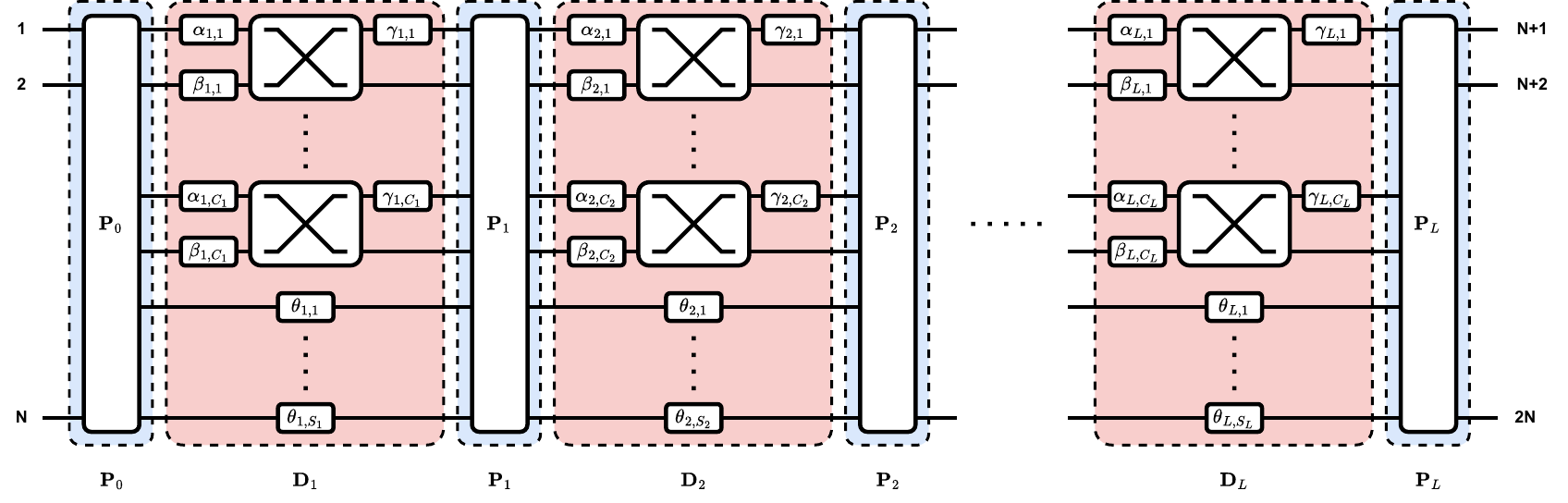}
\caption{Network of hybrid couplers and phase shifters used to prove the sufficient condition of Theorem~1.}
\label{fig:sufficient}
\end{figure*}

By exploiting the introduced definitions and Corollaries~\ref{cor:ser} and \ref{cor:par}, we are now ready to state the main result of this work, characterizing the scattering matrices of all the networks that are implementable with hybrid couplers and phase shifters.
\begin{theorem}
A network is implementable with hybrid couplers and phase shifters if and only if it has a scattering matrix $\mathbf{S}\in\mathbb{C}^{2N\times 2N}$ given by
\begin{equation}
\mathbf{S}=
\begin{bmatrix}
\mathbf{0}_N & \bar{\mathbf{S}}^T\\
\bar{\mathbf{S}} & \mathbf{0}_N
\end{bmatrix},
\end{equation}
where $\bar{\mathbf{S}}\in\mathbb{C}^{N\times N}$ can be decomposed as
\begin{equation}
\bar{\mathbf{S}}=\mathbf{P}_{L}\mathbf{D}_{L}\cdots\mathbf{P}_{2}\mathbf{D}_{2}\mathbf{P}_{1}\mathbf{D}_{1}\mathbf{P}_{0},\label{eq:theor}
\end{equation}
where $\mathbf{P}_{\ell}\in\mathbb{C}^{N\times N}$ is a permutation matrix, for $\ell=0,\ldots,L$, and $\mathbf{D}_{\ell}\in\mathbb{C}^{N\times N}$ is a block diagonal matrix, for $\ell=1,\ldots,L$, partitioned as
\begin{equation}
\mathbf{D}_{\ell}=\text{diag}\left(\mathbf{D}_{\ell,1},\ldots,\mathbf{D}_{\ell,C_\ell},e^{j\theta_{\ell,1}},\ldots,e^{j\theta_{\ell,S_\ell}}\right),\label{eq:Dell}
\end{equation}
where $\theta_{\ell,s}\in[0,2\pi)$, for $s=1,\ldots,S_\ell$, and $\mathbf{D}_{\ell,c}\in\mathbb{C}^{2\times2}$ is given by
\begin{equation}
\mathbf{D}_{\ell,c}=\frac{1}{\sqrt{2}}
\begin{bmatrix}
e^{j\theta_{\ell,c,11}} & e^{j\theta_{\ell,c,12}}\\
e^{j\theta_{\ell,c,21}} & e^{j\theta_{\ell,c,22}}
\end{bmatrix},\label{eq:Dellc}
\end{equation}
with $\theta_{\ell,c,11},\theta_{\ell,c,12},\theta_{\ell,c,21}\in[0,2\pi)$ and $\theta_{\ell,c,22}=\pi-\theta_{\ell,c,11}+\theta_{\ell,c,12}+\theta_{\ell,c,21}$, for $c=1,\ldots,C_\ell$, where $2C_\ell+S_\ell=N$.
\label{the:1}
\end{theorem}
\begin{proof}
Please, refer to Section~\ref{sec:proof}
\end{proof}

This theorem includes an ``if and only if'' statement that provides two-fold information.
First, it states that any matched network whose transmission scattering matrix can be decomposed as in \eqref{eq:theor} can be implemented with hybrid couplers and phase shifters.
Second, any network of hybrid couplers and phase shifters is a matched network whose transmission scattering matrix can be decomposed as in \eqref{eq:theor}.
Therefore, microwave networks of hybrid couplers and phase shifts allow to compute $\mathbf{v}_2=\mathbf{S}_{21}\mathbf{u}/2$, as derived in Section~\ref{sec:model}, where $\mathbf{u}$ is an arbitrary input signal and $\mathbf{S}_{21}$ is any matrix that can be expressed according to \eqref{eq:theor}.
Theorem~\ref{the:1} holds with both fixed and reconfigurable phase shifters.
If they are reconfigurable, multiple configurations will be possible with the same network, all satisfying the condition of Theorem~\ref{the:1}.

\begin{figure*}[t]
\centering
\includegraphics[width=0.80\textwidth]{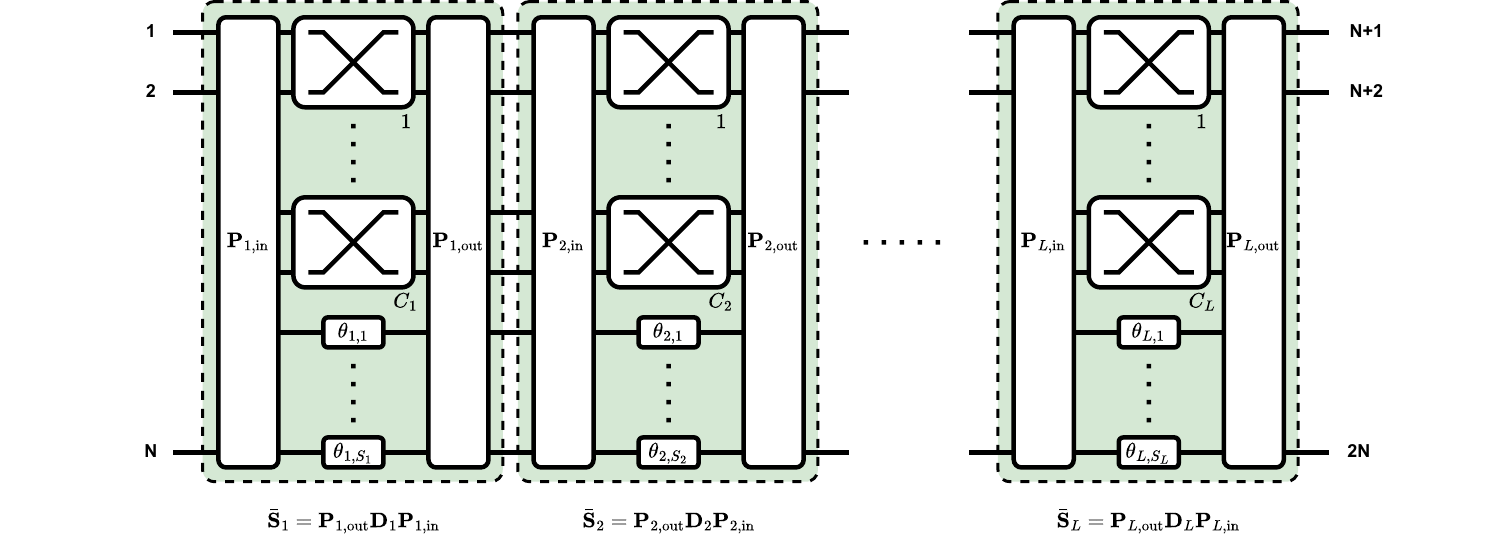}
\caption{Network of hybrid couplers and phase shifters used to prove the necessary condition of Theorem~1.}
\label{fig:necessary}
\end{figure*}

\section{Proof of the Main Result}
\label{sec:proof}

We have stated Theorem~\ref{the:1}, which provides an expression for the scattering matrices of all the possible networks of hybrid couplers and phase shifters.
In this section, we prove Theorem~\ref{the:1} by first proving its sufficient condition and then its necessary condition.

\subsection{Proof of Sufficiency}

The sufficient condition of Theorem~\ref{the:1} states that a matched network with transmission scattering matrix written as in \eqref{eq:theor} can be implemented with hybrid couplers and phase shifters.
To prove this sufficient condition, we provide the design of a network including series and parallel combinations of hybrid couplers, phase shifters, interconnections, and permutation networks, whose scattering matrix is in the form provided by Theorem~\ref{the:1}.
Such a network is the series of $2L+1$ subnetworks, all being reciprocal matched networks with $N$ input ports and $N$ output ports, as represented in Fig.~\ref{fig:sufficient}.
From left to right, these subnetworks have transmission scattering matrices $\mathbf{P}_0,\mathbf{D}_1,\mathbf{P}_1,\mathbf{D}_2,\mathbf{P}_2,\ldots,\mathbf{D}_L,\mathbf{P}_L$.
Since these subnetworks are matched networks, the resulting series network is a matched network with transmission scattering matrix given by $\bar{\mathbf{S}}=\mathbf{P}_{L}\mathbf{D}_{L}\cdots\mathbf{P}_{2}\mathbf{D}_{2}\mathbf{P}_{1}\mathbf{D}_{1}\mathbf{P}_{0}$, following Corollary~\ref{cor:ser}.

In the following, we provide the detailed design for all the $2L+1$ subnetworks.
First, the odd subnetworks (having transmission scattering matrix $\mathbf{P}_\ell$) are permutation networks, where the permutation matrix associated with the $\ell$th odd network is $\mathbf{P}_\ell$, for $\ell=0,\ldots,L$.
Second, the even subnetworks (having transmission scattering matrix $\mathbf{D}_\ell$) are formed by multiple matched networks in parallel.
The $\ell$th even network is the parallel of $C_\ell$ $4$-port networks and $S_\ell$ phase shifters with phase shifts $\theta_{\ell,1},\theta_{\ell,2},\ldots,\theta_{\ell,S_\ell}$, such that its transmission scattering matrix $\mathbf{D}_\ell$ fulfills \eqref{eq:Dell} following Corollary~\ref{cor:par}, for $\ell=1,\ldots,L$.
The $C_\ell$ $4$-port networks must be matched networks with transmission scattering matrix $\mathbf{D}_{\ell,1},\mathbf{D}_{\ell,2},\ldots,\mathbf{D}_{\ell,C_\ell}$ as given by \eqref{eq:Dellc}, hence each of them is the series of three networks with transmission scattering matrices
\begin{equation}
\begin{bmatrix}e^{j\alpha_{\ell,c}} & 0\\0 & e^{j\beta_{\ell,c}}\end{bmatrix},\;
\frac{1}{\sqrt{2}}
\begin{bmatrix}j & 1\\1 & j\end{bmatrix},\;
\begin{bmatrix}e^{j\gamma_{\ell,c}} & 0\\0 & 1\end{bmatrix},
\end{equation}
where the first network is the parallel of two phase shifters with phase shifts $\alpha_{\ell,c}$ and $\beta_{\ell,c}$, the second network is a hybrid coupler, and the third network is the parallel of a phase shifter with phase shift $\gamma_{\ell,c}$ and an interconnection, for $c=1,\ldots,C_\ell$ and $\ell=1,\ldots,L$.
To ensure that $\mathbf{D}_{\ell,c}$ fulfills \eqref{eq:Dellc}, we set
\begin{gather}
\alpha_{\ell,c}=\theta_{\ell,c,21},\label{eq:alpha}\\
\beta_{\ell,c}=\pi/2-\theta_{\ell,c,11}+\theta_{\ell,c,12}+\theta_{\ell,c,21},\label{eq:beta}\\
\gamma_{\ell,c}=\theta_{\ell,c,11}-\theta_{\ell,c,21}-\pi/2,\label{eq:gamma}
\end{gather}
for $c=1,\ldots,C_\ell$ and $\ell=1,\ldots,L$.
In this way, following Corollary~\ref{cor:ser}, we have
\begin{align}
\mathbf{D}_{\ell,c}
&=\frac{1}{\sqrt{2}}
\begin{bmatrix}e^{j\gamma_{\ell,c}} & 0\\0 & 1\end{bmatrix}
\begin{bmatrix}j & 1\\1 & j\end{bmatrix}
\begin{bmatrix}e^{j\alpha_{\ell,c}} & 0\\0 & e^{j\beta_{\ell,c}}\end{bmatrix}\\
&=\frac{1}{\sqrt{2}}
\begin{bmatrix}
e^{j\theta_{\ell,c,11}} & e^{j\theta_{\ell,c,12}}\\
e^{j\theta_{\ell,c,21}} & e^{j(\pi-\theta_{\ell,c,11}+\theta_{\ell,c,12}+\theta_{\ell,c,21})}
\end{bmatrix},
\end{align}
and the constructed network has a scattering matrix in the form given by Theorem~\ref{the:1}.

\subsection{Proof of Necessity}

The necessary condition of Theorem~\ref{the:1} states that any network of hybrid couplers and phase shifters is a matched network with transmission scattering matrix written as in \eqref{eq:theor}.
To prove this necessary condition, we show that any network of hybrid couplers and phase shifters has a scattering matrix that can be written in the form provided by Theorem~\ref{the:1}.

We begin by noting that any such network can be decomposed into $L$ subnetworks connected in series, each allowed to contain hybrid couplers and phase shifters only in parallel to each other.
For the $\ell$th subnetwork, we denote as $C_\ell$ the number of hybrid couplers, and as $S_\ell$ the number of phase shifters, whose phase shifts are $\theta_{\ell,1},\theta_{\ell,2},\ldots,\theta_{\ell,S_\ell}$, as illustrated in Fig.~\ref{fig:necessary}.
This parallel of hybrid couplers and phase shifters can be, in general, connected in series to two permutation networks: one preceding it, with permutation matrix $\mathbf{P}_{\ell,\text{in}}\in\mathbb{R}^{N\times N}$, and one following it, with permutation matrix $\mathbf{P}_{\ell,\text{out}}\in\mathbb{R}^{N\times N}$.
In principle, the parallel of hybrid couplers and phase shifters in the $\ell$th subnetwork could also include interconnections and permutation networks as further components in parallel.
However, this possibility can be excluded without loss of generality for two reasons.
First, recalling that interconnections are phase shifters with zero phase shift, their effect is already captured within the $S_\ell$ phase shifters.
Second, permutation networks would be redundant as their effect is already incorporated in $\mathbf{P}_{\ell,\text{in}}$ and $\mathbf{P}_{\ell,\text{out}}$.
Therefore, the numbers of hybrid couplers and phase shifters in the $\ell$th subnetwork are constrained by $2C_\ell+S_\ell=N$.

Following the description of the $\ell$th subnetwork, we obtain from Corollaries~\ref{cor:ser} and \ref{cor:par} that it is a matched network with transmission scattering matrix $\bar{\mathbf{S}}_{\ell}\in\mathbb{C}^{N\times N}$ given by
\begin{equation}
\bar{\mathbf{S}}_{\ell}=\mathbf{P}_{\ell,\text{out}}\mathbf{D}_{\ell}\mathbf{P}_{\ell,\text{in}},\label{eq:Sell}
\end{equation}
where $\mathbf{D}_{\ell}\in\mathbb{C}^{N\times N}$ is a block diagonal matrix partitioned as
\begin{equation}
\mathbf{D}_{\ell}=\text{diag}\left(\mathbf{D},\ldots,\mathbf{D},e^{j\theta_{\ell,1}},\ldots,e^{j\theta_{\ell,S_\ell}}\right),
\end{equation}
where $\mathbf{D}\in\mathbb{C}^{2\times2}$ is the transmission scattering matrix of a hybrid coupler given by $\mathbf{D}=[[j,1]^T,[1,j]^T]/\sqrt{2}$.

Since a generic network of hybrid couplers and phase shifters is a series of $L$ such subnetworks, its transmission scattering matrix $\bar{\mathbf{S}}\in\mathbb{C}^{N\times N}$ writes as
\begin{equation}
\bar{\mathbf{S}}=\bar{\mathbf{S}}_{L}\cdots\bar{\mathbf{S}}_{2}\bar{\mathbf{S}}_{1},
\end{equation}
following Corollary~\ref{cor:ser}, which can be rewritten by using \eqref{eq:Sell} as
\begin{equation}
\bar{\mathbf{S}}
=\mathbf{P}_{L,\text{out}}\mathbf{D}_{L}\mathbf{P}_{L,\text{in}}
\cdots
\mathbf{P}_{2,\text{out}}\mathbf{D}_{2}\mathbf{P}_{2,\text{in}}
\mathbf{P}_{1,\text{out}}\mathbf{D}_{1}\mathbf{P}_{1,\text{in}}.
\end{equation}
Therefore, we have that any network of hybrid couplers and phase shifters has a transmission scattering matrix as required by Theorem~\ref{the:1}, where $\mathbf{P}_{0}=\mathbf{P}_{1,\text{in}}$, $\mathbf{P}_{\ell}=\mathbf{P}_{\ell+1,\text{in}}\mathbf{P}_{\ell,\text{out}}$, for $\ell=1,\ldots,L-1$, and $\mathbf{P}_{L}=\mathbf{P}_{L,\text{out}}$ are permutation matrices, and $\theta_{\ell,c,11}=\theta_{\ell,c,22}=0$, $\theta_{\ell,c,12}=\theta_{\ell,c,21}=\pi/2$, for $c=1,\ldots,C_\ell$ and $\ell=1,\ldots,L$.

\begin{figure*}[t]
\centering
\includegraphics[width=0.96\textwidth]{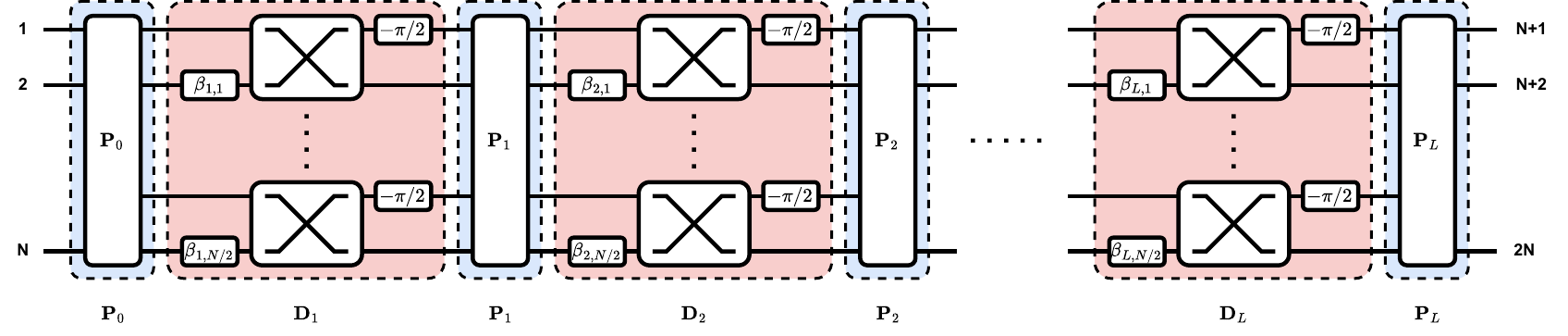}
\caption{Network of hybrid couplers and phase shifters performing the $N\times N$ DFT of the input, with $N=2^L$.
The permutation matrices $\mathbf{P}_\ell$ are given by \eqref{eq:P0-dft}-\eqref{eq:Pell-dft}, and the phase shifts $\beta_{\ell,c}$, for $c=1,\ldots,N/2$, are $\beta_{\ell,c}=\pi/2-2\pi((c-1)\mod 2^{\ell-1})/2^{\ell}$.}
\label{fig:dft}
\end{figure*}

\section{Analog Computation of the DFT}
\label{sec:dft}

We have shown that specific matrix-vector products, also known as linear transformations, can be performed in the analog domain through networks implemented with hybrid couplers and phase shifters.
From our main result in Theorem~\ref{the:1}, we notice that all these linear transformations are unitary, since the transmission scattering matrix of any network of hybrid couplers and phase shifters is unitary.
Therefore, we are interested in understanding whether the most popular unitary transformation, namely the \gls{dft}, can be performed in the analog domain through hybrid couplers and phase shifters.

We recall that the $N\times N$ \gls{dft} matrix $\mathbf{F}_N\in\mathbb{C}^{N\times N}$ is defined as
\begin{equation}
\left[\mathbf{F}_N\right]_{i,k}=\frac{1}{\sqrt{N}}\omega^{(i-1)(k-1)},\label{eq:dft}
\end{equation}
for $i,k=1,\ldots,N$, where $\omega=e^{-j2\pi/N}$.
We show in the following that the \gls{dft} can be computed with networks of hybrid couplers and phase shifters by proving that it can be decomposed as required by Theorem~\ref{the:1}.
To this end, we first introduce a known decomposition of the \gls{dft} that we later exploit and the related concept of odd-even permutation matrix.
\begin{definition}
(Odd-even permutation matrix)
An $N\times N$ permutation matrix, with $N$ even, is an odd-even permutation matrix, denoted as $\tilde{\mathbf{P}}_N\in\{0,1\}^{N\times N}$, if it is given by
\begin{equation}
\left[\tilde{\mathbf{P}}_N\right]_{i,k} = 
\begin{cases}
1 & \text{if } i\leq N/2 \text{ and } k = 2i-1\\
1 & \text{if } i>N/2 \text{ and } k = 2i-N\\
0 & \text{otherwise}
\end{cases},
\end{equation}
for $i,k=1,\ldots,N$.
\end{definition}
Note that this permutation matrix is named as ``odd-even'' since when multiplied by a vector $\mathbf{v}\in\mathbb{C}^{N\times1}$ as in $\tilde{\mathbf{P}}_N\mathbf{v}$, it returns first the odd entries and then the even entries of $\mathbf{v}$.
This matrix can also be referred to as the ``even-odd'' permutation matrix when the entries of a vector are indexed starting from zero, e.g., in \cite[Chapter~10.3]{str09}.
As three illustrative examples of the odd-even permutation matrix $\tilde{\mathbf{P}}_N$, the matrices $\tilde{\mathbf{P}}_2$, $\tilde{\mathbf{P}}_4$, and $\tilde{\mathbf{P}}_8$ are $\tilde{\mathbf{P}}_2=\mathbf{I}_2$,
\begin{equation}
\tilde{\mathbf{P}}_4=
\begin{bmatrix}
1 & 0 & 0 & 0\\
0 & 0 & 1 & 0\\
0 & 1 & 0 & 0\\
0 & 0 & 0 & 1
\end{bmatrix},
\tilde{\mathbf{P}}_8=
\begin{bmatrix}
1 & 0 & 0 & 0 & 0 & 0 & 0 & 0\\
0 & 0 & 1 & 0 & 0 & 0 & 0 & 0\\
0 & 0 & 0 & 0 & 1 & 0 & 0 & 0\\
0 & 0 & 0 & 0 & 0 & 0 & 1 & 0\\
0 & 1 & 0 & 0 & 0 & 0 & 0 & 0\\
0 & 0 & 0 & 1 & 0 & 0 & 0 & 0\\
0 & 0 & 0 & 0 & 0 & 1 & 0 & 0\\
0 & 0 & 0 & 0 & 0 & 0 & 0 & 1
\end{bmatrix}.
\end{equation}

The concept of odd-even permutation matrix appears in a decomposition of the \gls{dft} matrix used in the \gls{fft} \cite{coo65}, arguably ``the most valuable numerical algorithm in the last century'' \cite[Chapter~10.3]{str09}.
Such a decomposition is recalled in the following lemma.
\begin{lemma}
The $N\times N$ \gls{dft} matrix $\mathbf{F}_N$, with $N=2^L$, $L\in\mathbb{N}$, can be decomposed as
\begin{equation}
\mathbf{F}_N=\frac{1}{\sqrt{2}}
\begin{bmatrix}
\mathbf{I}_{N/2} & \mathbf{\Omega}_{N/2}\\
\mathbf{I}_{N/2} & -\mathbf{\Omega}_{N/2}
\end{bmatrix}
\begin{bmatrix}
\mathbf{F}_{N/2} & \\
 & \mathbf{F}_{N/2}
\end{bmatrix}
\tilde{\mathbf{P}}_N,\label{eq:fft}
\end{equation}
where $\mathbf{\Omega}_{N/2}\in\mathbb{C}^{N/2\times N/2}$ is a diagonal matrix defined as
\begin{equation}
\mathbf{\Omega}_{N/2}=\emph{diag}\left(1,\omega,\omega^2,\ldots,\omega^{N/2-2},\omega^{N/2-1}\right),
\end{equation}
with $\omega=e^{-j2\pi/N}$, i.e., $[\mathbf{\Omega}_{N/2}]_{c,c}=e^{-j2\pi(c-1)/N}$, for $c=1,\ldots,N/2$, and $\tilde{\mathbf{P}}_N$ is the $N\times N$ odd-even permutation matrix.
\label{lem:fft}
\end{lemma}
\begin{proof}
Please, refer to \cite[Chapter~10.3]{str09}.
\end{proof}

By exploiting the decomposition in Lemma~\ref{lem:fft}, we derive another decomposition of the \gls{dft} matrix in the following proposition, which is in the form required by Theorem~\ref{the:1}.
\begin{proposition}
The $N\times N$ \gls{dft} matrix $\mathbf{F}_N$, with $N=2^L$, $L\in\mathbb{N}$, can be decomposed as
\begin{equation}
\mathbf{F}_{N}=\mathbf{P}_{L}\mathbf{D}_{L}\cdots\mathbf{P}_{2}\mathbf{D}_{2}\mathbf{P}_{1}\mathbf{D}_{1}\mathbf{P}_{0}.\label{eq:theor-dft}
\end{equation}
The matrix $\mathbf{P}_{\ell}\in\{0,1\}^{N\times N}$ is a permutation matrix given by
\begin{equation}
\mathbf{P}_{0}=\prod_{\ell=1}^{L}\left(\mathbf{I}_{2^{L-\ell}}\otimes\tilde{\mathbf{P}}_{2^\ell}\right),\label{eq:P0-dft}
\end{equation}
\begin{equation}
\mathbf{P}_{\ell}=\left(\mathbf{I}_{2^{L-\ell-1}}\otimes\tilde{\mathbf{P}}_{2^{\ell+1}}^T\right)\left(\mathbf{I}_{2^{L-\ell}}\otimes\tilde{\mathbf{P}}_{2^{\ell}}\right),\label{eq:Pell-dft}
\end{equation}
for $\ell=1,\ldots,L-1$, and $\mathbf{P}_{L}=\tilde{\mathbf{P}}_{2^L}$.
Besides, $\mathbf{D}_{\ell}\in\mathbb{C}^{N\times N}$ is a block diagonal matrix, for $\ell=1,\ldots,L$, given by
\begin{align}
\mathbf{D}_{\ell}=\mathbf{I}_{2^{L-\ell}}\otimes\tilde{\mathbf{D}}_{2^\ell},\label{eq:Dell-dft-1}
\end{align}
i.e., having $2^{L-\ell}$ times the block $\tilde{\mathbf{D}}_{2^\ell}\in\mathbb{C}^{2^\ell\times 2^\ell}$ on the diagonal, where $\tilde{\mathbf{D}}_{2^\ell}$ is also a block diagonal matrix
\begin{equation}
\tilde{\mathbf{D}}_{2^\ell}=\text{diag}\left(\tilde{\mathbf{D}}_{2^\ell,1},\tilde{\mathbf{D}}_{2^\ell,2},\ldots,\tilde{\mathbf{D}}_{2^\ell,2^{\ell-1}}\right),
\end{equation}
with $\tilde{\mathbf{D}}_{2^\ell,c}\in\mathbb{C}^{2\times 2}$ given by
\begin{equation}
\tilde{\mathbf{D}}_{2^\ell,c}=\frac{1}{\sqrt{2}}
\begin{bmatrix}
1 &  e^{-j\frac{2\pi(c-1)}{2^{\ell}}}\\
1 & -e^{-j\frac{2\pi(c-1)}{2^{\ell}}}
\end{bmatrix},\label{eq:Dell-dft-3}
\end{equation}
for $c=1,\ldots,2^{\ell-1}$.
\label{pro:dft}
\end{proposition}
\begin{proof}
Please, refer to Appendix~C.
\end{proof}

As a direct consequence of the \gls{dft} matrix decomposition in Proposition~\ref{pro:dft}, we provide the following result.
\begin{corollary}
A microwave network that performs the \gls{dft} of its input signal $\mathbf{u}$, i.e., computes $\mathbf{F}_N\mathbf{u}/2$, with $N=2^L$, $L\in\mathbb{N}$, is implementable with hybrid couplers and phase shifters.
\label{cor:dft}
\end{corollary}
\begin{proof}
Recalling the model of a \gls{milac} in Section~\ref{sec:model}, a microwave network performs the \gls{dft} of its input signal when it is a matched network with transmission scattering matrix being the $N\times N$ \gls{dft} matrix $\mathbf{F}_N$.
Therefore, such a network can be implemented using hybrid couplers and phase shifters if and only if $\mathbf{F}_N$ can be decomposed as stated in Theorem~\ref{the:1}.
To show this, we prove that the decomposition of $\mathbf{F}_N$ provided by Proposition~\ref{pro:dft} is in the form required by Theorem~\ref{the:1}.
Since the decomposition is already in the form $\mathbf{F}_{N}=\mathbf{P}_{L}\mathbf{D}_{L}\cdots\mathbf{P}_{2}\mathbf{D}_{2}\mathbf{P}_{1}\mathbf{D}_{1}\mathbf{P}_{0}$, we need to verify that the matrices $\mathbf{P}_{\ell}$ and $\mathbf{D}_{\ell}$ are in the correct form.
First, all matrices $\mathbf{P}_{\ell}$ are permutation matrices since they are products of permutation matrices (see \eqref{eq:P0-dft}-\eqref{eq:Pell-dft}), and therefore fulfill the condition of Theorem~\ref{the:1}.
Second, also all the matrices $\mathbf{D}_{\ell}=\mathbf{I}_{2^{L-\ell}}\otimes\tilde{\mathbf{D}}_{2^\ell}$ fulfill Theorem~\ref{the:1}, as they are block diagonal matrices as requested by Theorem~\ref{the:1} where $C_\ell=N/2$ and $S_\ell=0$, and phase shifts $\theta_{\ell,c,11},\theta_{\ell,c,12},\theta_{\ell,c,21}$ are $\theta_{\ell,c,11}=\theta_{\ell,c,21}=0$, $\theta_{\ell,c,12}=-2\pi((c-1)\mod 2^{\ell-1})/2^{\ell}$, for $c=1,\ldots,N/2$ (see \eqref{eq:Dell-dft-1}-\eqref{eq:Dell-dft-3}).
\end{proof}

Corollary~\ref{cor:dft} tells us that we can compute the \gls{dft} of an input vector in the analog domain with a network made exclusively of hybrid couplers and phase shifters.
Interestingly, we also have a systematic design method to construct such networks for any \gls{dft} with size $N$ power of two, which can be obtained from the Proof of sufficiency of Theorem~\ref{the:1}.
In that proof, we have shown how to construct a network of hybrid couplers and phase shifters that has a transmission scattering matrix fulfilling Theorem~\ref{the:1}, depending on the permutation matrices $\mathbf{P}_{\ell}$ and the phase shifts $\theta_{\ell,s}$, for $s=1,\ldots,S_\ell$, and $\theta_{\ell,c,11},\theta_{\ell,c,12},\theta_{\ell,c,21}$, for $c=1,\ldots,C_\ell$.
Such a network is represented in Fig.~\ref{fig:sufficient}.
Since the \gls{dft} matrix has a decomposition that is a special case of the decomposition in Theorem~\ref{the:1}, we also know how to construct a network of hybrid couplers and phase shifters whose transmission scattering matrix is the \gls{dft} matrix.
Specifically, the \gls{dft} matrix decomposition in Proposition~\ref{pro:dft} requires $C_\ell=N/2$ and $S_\ell=0$, indicating that in the subnetworks implementing the matrices $\mathbf{D}_{\ell}$ there are $N/2$ hybrid couplers in parallel.
Furthermore, $\theta_{\ell,c,11}=\theta_{\ell,c,21}=0$ and $\theta_{\ell,c,12}=-2\pi((c-1)\mod 2^{\ell-1})/2^{\ell}$ give that the phase shifts in \eqref{eq:alpha}-\eqref{eq:gamma} boil down to $\alpha_{\ell,c}=0$, $\beta_{\ell,c}=\pi/2-2\pi((c-1)\mod 2^{\ell-1})/2^{\ell}$, and $\gamma_{\ell,c}=-\pi/2$, for $c=1,\ldots,N/2$ and $\ell=1,\ldots,L$.
We have therefore all the information needed to construct a network in the form of the one in Fig.~\ref{fig:sufficient} that computes the \gls{dft} of the input vector, as represented in Fig.~\ref{fig:dft}.

It is worth noting that the \gls{dft} matrix is different from the Butler matrix used for beam steering, as clarified in Appendix~D.
For implementing the $N\times N$ Butler matrix using hybrid couplers and phase shifters, systematic design procedures have been available since the 1960s \cite{moo64}.
Around the same period, shortly after the introduction of the \gls{fft} algorithm \cite{coo65}, researchers identified a striking conceptual similarity between the Butler matrix and the \gls{fft} \cite{nes68,she68,uen81}.
Specifically, the signal flow graph of a Butler matrix resembles the information processing structure of the \gls{fft}, and the operations performed by the hybrid couplers in a Butler matrix correspond to those carried out in the \gls{fft} algorithm.
Very recently, an inspiring work proposed computing the \gls{dft} in the analog domain with \gls{em} signals by stacking multiple transmissive metasurfaces with suitably designed transmissive properties \cite{an24}, thereby differing from our approach.

\section{Examples of Networks Computing the DFT}
\label{sec:examples}

We have shown that the \gls{dft} can be computed in the analog domain by hybrid couplers and phase shifters.
In this section, we clarify how this can be achieved by showing networks of hybrid couplers and phase shifters that compute four examples of $N\times N$ \glspl{dft}, with $N=2^L$ where $L\in\{1,2,3,4\}$.
For each of these four \gls{dft} matrices, we show their specific decomposition $\mathbf{F}_{N}=\mathbf{P}_{L}\mathbf{D}_{L}\cdots\mathbf{P}_{2}\mathbf{D}_{2}\mathbf{P}_{1}\mathbf{D}_{1}\mathbf{P}_{0}$ and provide the schematic of the network of hybrid couplers and phase shifters whose transmission scattering matrix is $\mathbf{F}_{N}$.

For $L=1$, the $2\times2$ \gls{dft} matrix $\mathbf{F}_{2}$ can be trivially decomposed according to Proposition~\ref{pro:dft} as
\begin{equation}
\mathbf{F}_2=
\underbrace{\mathbf{I}_2}_{\mathbf{P}_{1}}
\underbrace{\tilde{\mathbf{D}}_2}_{\mathbf{D}_{1}}
\underbrace{\mathbf{I}_2}_{\mathbf{P}_{0}},\label{eq:F2}
\end{equation}
where we highlight that $\mathbf{P}_0=\mathbf{I}_2$, $\mathbf{D}_1=\tilde{\mathbf{D}}_2$, and $\mathbf{P}_1=\mathbf{I}_2$.
Following the systematic design procedure summarized in Fig.~\ref{fig:dft} for the case $L=1$, the resulting network of hybrid couplers and phase shifters whose transmission scattering matrix is $\mathbf{F}_2$ is represented in Fig.~\ref{fig:dft2}.
Since $\mathbf{P}_0$ and $\mathbf{P}_1$ are the identity matrix, there are no permutation networks.
Note that this network is a series of three networks with transmission scattering matrices
\begin{equation}
\begin{bmatrix}
1 & 0\\
0 & e^{j\frac{\pi}{2}}
\end{bmatrix},\;
\frac{1}{\sqrt{2}}
\begin{bmatrix}
j & 1\\
1 & j
\end{bmatrix},\;
\begin{bmatrix}
e^{-j\frac{\pi}{2}} & 0\\
0 & 1
\end{bmatrix},
\end{equation}
where the first network is the parallel of an interconnection and a phase shifter with phase $\pi/2$, the second network is a hybrid coupler, and the third network is the parallel of a phase shifter with phase $-\pi/2$ and an interconnection.
Hence, following Corollary~\ref{cor:ser}, the network in Fig.~\ref{fig:dft2} has transmission scattering matrix
\begin{align}
\bar{\mathbf{S}}&=
\begin{bmatrix}e^{-j\frac{\pi}{2}} & 0\\0 & 1\end{bmatrix}
\left(\frac{1}{\sqrt{2}}
\begin{bmatrix}j & 1\\1 & j\end{bmatrix}\right)
\begin{bmatrix}1 & 0\\0 & e^{j\frac{\pi}{2}}\end{bmatrix}\\
&=\frac{1}{\sqrt{2}}
\begin{bmatrix}
1 & 1\\
1 & -1
\end{bmatrix},
\end{align}
which is $\bar{\mathbf{S}}=\mathbf{F}_2$, as expected.

\begin{figure}[t]
\centering
\includegraphics[width=0.24\textwidth]{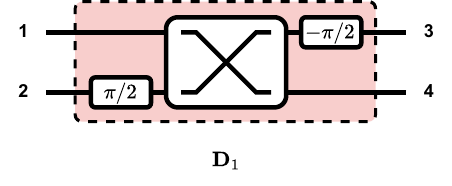}
\caption{A network computing the $2\times2$ DFT.}
\label{fig:dft2}
\end{figure}

For $L=2$, the $4\times4$ \gls{dft} matrix $\mathbf{F}_{4}$ can be decomposed following Proposition~\ref{pro:dft} as
\begin{equation}
\mathbf{F}_4=
\underbrace{\tilde{\mathbf{P}}_4}_{\mathbf{P}_2}
\underbrace{\tilde{\mathbf{D}}_4}_{\mathbf{D}_2}
\underbrace{\tilde{\mathbf{P}}_4^T}_{\mathbf{P}_1}
\underbrace{\left(\mathbf{I}_{2}\otimes\tilde{\mathbf{D}}_2\right)}_{\mathbf{D}_1}
\underbrace{\tilde{\mathbf{P}}_4}_{\mathbf{P}_0}.\label{eq:F4}
\end{equation}
Note that this decomposition can also be obtained by recalling that $\mathbf{F}_4=\tilde{\mathbf{P}}_4\tilde{\mathbf{D}}_4\tilde{\mathbf{P}}_4^T\left(\mathbf{I}_{2}\otimes\mathbf{F}_{2}\right)\tilde{\mathbf{P}}_4$, and expressing $\mathbf{F}_{2}$ as in \eqref{eq:F2}.
Specializing Fig.~\ref{fig:dft} for the case $L=2$, the resulting network of hybrid couplers and phase shifters whose transmission scattering matrix is $\mathbf{F}_{4}$ is represented in Fig.~\ref{fig:dft4}.

For $L=3$, the $8\times8$ \gls{dft} matrix $\mathbf{F}_{8}$ can be decomposed with Proposition~\ref{pro:dft} as
\begin{multline}
\mathbf{F}_8=
\underbrace{\tilde{\mathbf{P}}_8}_{\mathbf{P}_3}
\underbrace{\tilde{\mathbf{D}}_8}_{\mathbf{D}_3}
\underbrace{\tilde{\mathbf{P}}_8^T\left(\mathbf{I}_2\otimes\tilde{\mathbf{P}}_{4}\right)}_{\mathbf{P}_2}
\underbrace{\left(\mathbf{I}_2\otimes\tilde{\mathbf{D}}_{4}\right)}_{\mathbf{D}_2}\\
\underbrace{\left(\mathbf{I}_2\otimes\tilde{\mathbf{P}}_{4}^T\right)}_{\mathbf{P}_1}
\underbrace{\left(\mathbf{I}_4\otimes\tilde{\mathbf{D}}_{2}\right)}_{\mathbf{D}_1}
\underbrace{\left(\mathbf{I}_2\otimes\tilde{\mathbf{P}}_{4}\right)\tilde{\mathbf{P}}_8}_{\mathbf{P}_0}.\label{eq:F8}
\end{multline}
The same decomposition can also be obtained $\mathbf{F}_8=\tilde{\mathbf{P}}_8\tilde{\mathbf{D}}_8\tilde{\mathbf{P}}_8^T\left(\mathbf{I}_2\otimes\mathbf{F}_{4}\right)\tilde{\mathbf{P}}_8$, expressing $\mathbf{F}_{4}$ by \eqref{eq:F4}, and applying the mixed-product property of the Kronecker product.
Based on the design in Fig.~\ref{fig:dft} specialized for $L=3$, a network of hybrid couplers and phase shifters whose transmission scattering matrix is $\mathbf{F}_{8}$ is shown in Fig.~\ref{fig:dft8}.

For $L=4$, the $16\times16$ \gls{dft} matrix $\mathbf{F}_{16}$ can be decomposed with Proposition~\ref{pro:dft} as
\begin{multline}
\mathbf{F}_{16}=
\underbrace{\tilde{\mathbf{P}}_{16}}_{\mathbf{P}_{4}}
\underbrace{\tilde{\mathbf{D}}_{16}}_{\mathbf{D}_{4}}
\underbrace{\tilde{\mathbf{P}}_{16}^T\left(\mathbf{I}_2\otimes\tilde{\mathbf{P}}_{8}\right)}_{\mathbf{P}_{3}}
\underbrace{\left(\mathbf{I}_2\otimes\tilde{\mathbf{D}}_{8}\right)}_{\mathbf{D}_{3}}\\
\underbrace{\left(\mathbf{I}_2\otimes\tilde{\mathbf{P}}_{8}^T\right)\left(\mathbf{I}_4\otimes\tilde{\mathbf{P}}_{4}\right)}_{\mathbf{P}_{2}}
\underbrace{\left(\mathbf{I}_4\otimes\tilde{\mathbf{D}}_{4}\right)}_{\mathbf{D}_{2}}
\underbrace{\left(\mathbf{I}_4\otimes\tilde{\mathbf{P}}_{4}^T\right)}_{\mathbf{P}_{1}}\\
\underbrace{\left(\mathbf{I}_8\otimes\tilde{\mathbf{D}}_{2}\right)}_{\mathbf{D}_{1}}
\underbrace{\left(\mathbf{I}_4\otimes\tilde{\mathbf{P}}_{4}\right)\left(\mathbf{I}_2\otimes\tilde{\mathbf{P}}_{8}\right)\tilde{\mathbf{P}}_{16}}_{\mathbf{P}_{0}}.\label{eq:F16}
\end{multline}
Such a decomposition can also be derived recalling that $\mathbf{F}_{16}=\tilde{\mathbf{P}}_{16}\tilde{\mathbf{D}}_{16}\tilde{\mathbf{P}}_{16}^T\left(\mathbf{I}_2\otimes\mathbf{F}_{8}\right)\tilde{\mathbf{P}}_{16}$, where $\mathbf{F}_{8}$ is given by \eqref{eq:F8}, and applying the mixed-product property of the Kronecker product.
Considering the design in Fig.~\ref{fig:dft} for the case $L=4$, a network of hybrid couplers and phase shifters whose transmission scattering matrix is $\mathbf{F}_{16}$ is given in Fig.~\ref{fig:dft16}.

\begin{figure}[t]
\centering
\includegraphics[width=0.48\textwidth]{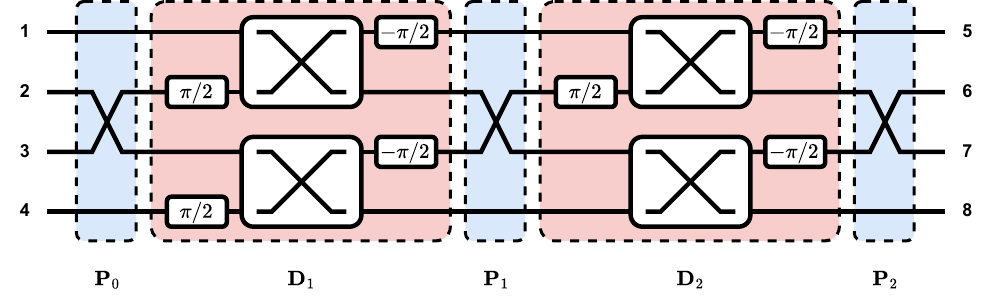}
\caption{A network computing the $4\times4$ DFT.}
\label{fig:dft4}
\end{figure}

\begin{figure*}[t]
\centering
\includegraphics[width=0.72\textwidth]{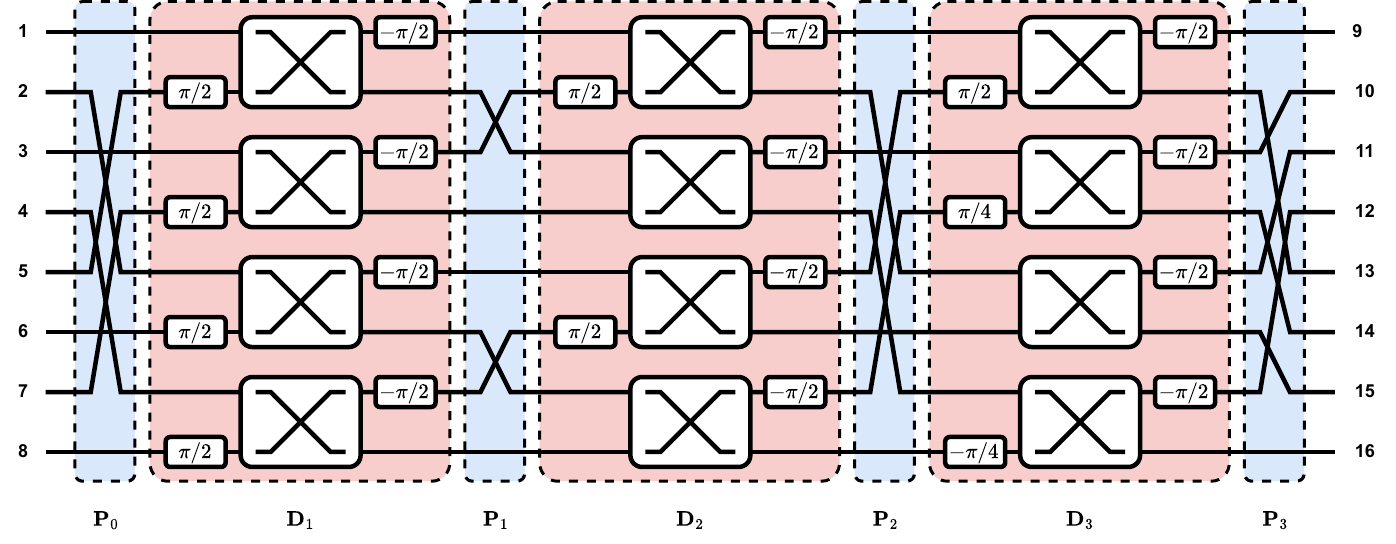}
\caption{A network computing the $8\times8$ DFT.}
\label{fig:dft8}
\end{figure*}

\begin{figure*}[t]
\centering
\includegraphics[width=0.96\textwidth]{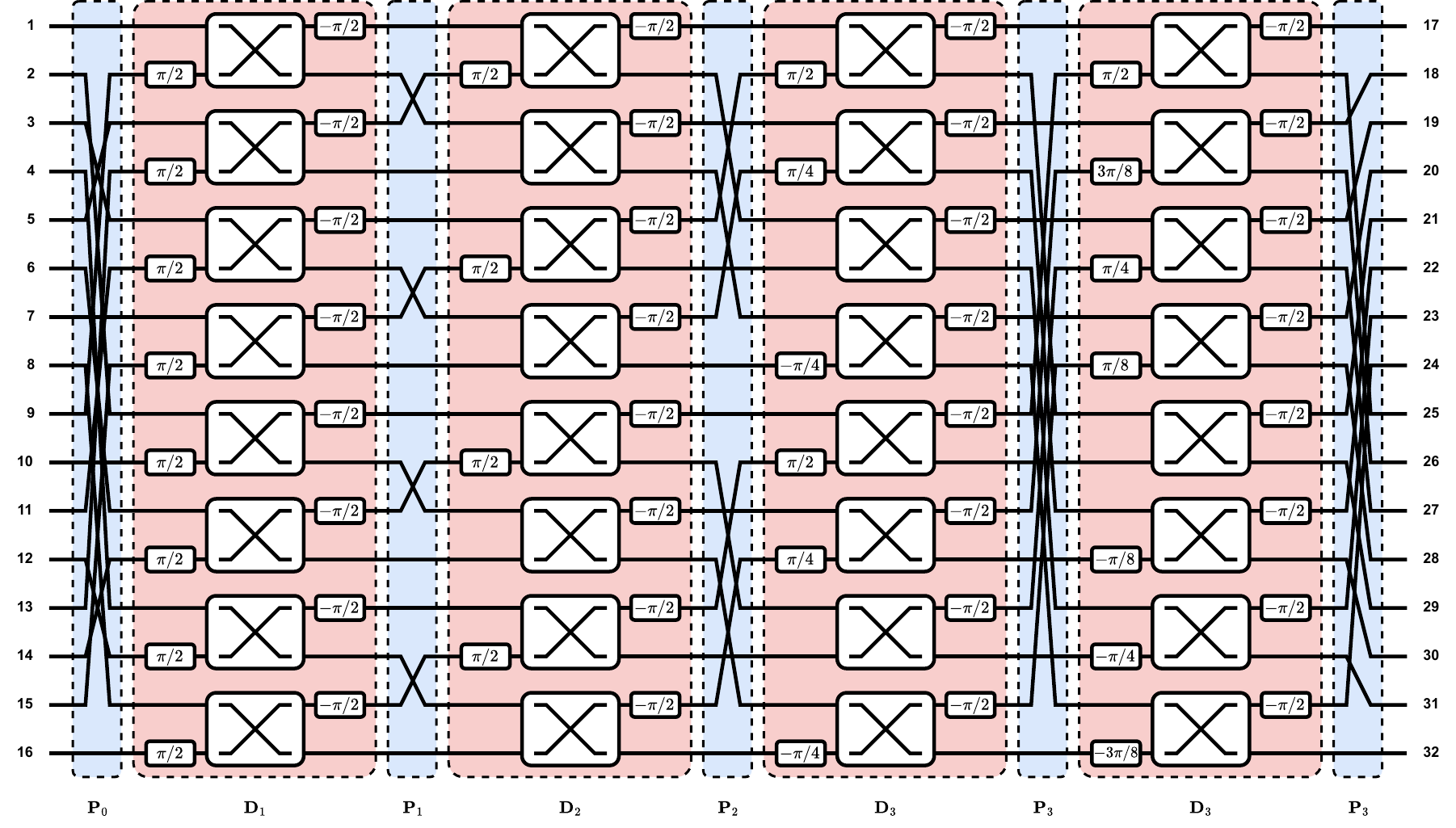}
\caption{A network computing the $16\times16$ DFT.}
\label{fig:dft16}
\end{figure*}

We have explicitly seen how to implement networks of hybrid couplers and phase shifters to compute $N\times N$ \glspl{dft} of the input signal, where $N=2^L$ with $L\in\{1,2,3,4\}$.
Larger \glspl{dft} with $L>4$ can also be computed in the analog domain with hybrid couplers and phase shifters, and the corresponding networks can readily be designed following Fig.~\ref{fig:dft}.
We have observed that for $L\in\{1,2,3,4\}$ the required number of hybrid couplers is 1, 4, 12, and 32, respectively.
In general, $2^{L-1}$ hybrid couplers are needed to implement each matrix $\mathbf{D}_{\ell}$, for $\ell=1,\ldots,L$, resulting in a total of
\begin{equation}
C^{\text{DFT}}=2^{L-1}L=\frac{N}{2}\log_2\left(N\right).
\end{equation}
This value therefore characterizes how the circuit size scales with the \gls{dft} size $N$.

\section{Analog Computation of the\\Hadamard Transform}
\label{sec:hada}

\begin{figure*}[t]
\centering
\includegraphics[width=0.96\textwidth]{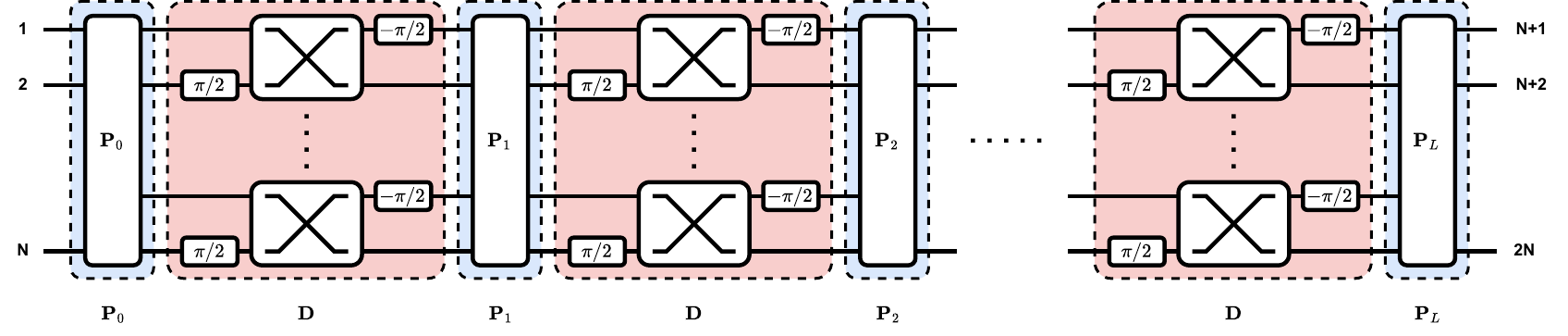}
\caption{Network of hybrid couplers and phase shifters performing the $N\times N$ Hadamard transform of the input, with $N=2^L$.
The permutation matrices $\mathbf{P}_\ell$ are given by \eqref{eq:Pell-hada}.}
\label{fig:hada}
\end{figure*}

In this section, we demonstrate how the Hadamard transform can be computed using hybrid couplers and phase shifters, following an approach similar to that used for the \gls{dft}.
We begin introducing the Hadamard matrix with a recursive definition \cite{fin77}, and then show that it has a decomposition as required by Theorem~\ref{the:1}.
\begin{definition}
(Hadamard matrix)
The $N\times N$ Hadamard matrix $\mathbf{H}_N\in\mathbb{R}^{N\times N}$, with $N=2^L$, $L\in\mathbb{N}$, is recursively defined as
\begin{equation}
\mathbf{H}_N=\mathbf{H}_2\otimes\mathbf{H}_{N/2},
\end{equation}
when $L>1$, and $\mathbf{H}_2=[[1,1]^T,[1,-1]^T]/\sqrt{2}$.
\end{definition}
As three illustrative examples, the Hadamard matrices $\mathbf{H}_2$, $\mathbf{H}_4$, and $\mathbf{H}_8$ are
\begin{equation}
\mathbf{H}_2=\frac{1}{\sqrt{2}}
\begin{bmatrix}
1 & 1\\
1 & -1
\end{bmatrix},\;
\mathbf{H}_4=\frac{1}{\sqrt{4}}
\begin{bmatrix}
1 &  1 &  1 &  1\\
1 & -1 &  1 & -1\\
1 &  1 & -1 & -1\\
1 & -1 & -1 &  1
\end{bmatrix},
\end{equation}
\begin{equation}
\mathbf{H}_8=\frac{1}{\sqrt{8}}
\begin{bmatrix}
1 &  1 &  1 &  1 &  1 &  1 &  1 &  1\\
1 & -1 &  1 & -1 &  1 & -1 &  1 & -1\\
1 &  1 & -1 & -1 &  1 &  1 & -1 & -1\\
1 & -1 & -1 &  1 &  1 & -1 & -1 &  1\\
1 &  1 &  1 &  1 & -1 & -1 & -1 & -1\\
1 & -1 &  1 & -1 & -1 &  1 & -1 &  1\\
1 &  1 & -1 & -1 & -1 & -1 &  1 &  1\\
1 & -1 & -1 &  1 & -1 &  1 &  1 & -1
\end{bmatrix}.
\end{equation}
Note that Hadamard matrices are unitary and have entries being $+1$ or $-1$ (up to the scaling factor $1/\sqrt{N}$).

By exploiting the recursive definition of the Hadamard matrix, we can decompose it in the form required by Theorem~\ref{the:1}, as stated in the following proposition.
\begin{proposition}
The $N\times N$ Hadamard matrix $\mathbf{H}_N$, with $N=2^L$, $L\in\mathbb{N}$, can be decomposed as
\begin{equation}
\mathbf{H}_{N}=\mathbf{P}_{L}\mathbf{D}\cdots\mathbf{P}_{2}\mathbf{D}\mathbf{P}_{1}\mathbf{D}\mathbf{P}_{0}.\label{eq:theor-hada}
\end{equation}
The matrix $\mathbf{P}_{\ell}\in\{0,1\}^{N\times N}$ is a permutation matrix given by
\begin{equation}
\mathbf{P}_{\ell}=\left(\mathbf{I}_{2^{L-\ell-1}}\otimes\tilde{\mathbf{P}}_{2^{\ell+1}}^T\right)\left(\mathbf{I}_{2^{L-\ell}}\otimes\tilde{\mathbf{P}}_{2^{\ell}}\right),\label{eq:Pell-hada}
\end{equation}
for $\ell=0,\ldots,L-1$, and $\mathbf{P}_{L}=\tilde{\mathbf{P}}_{2^L}$.
Besides, $\mathbf{D}\in\mathbb{C}^{N\times N}$ is a block diagonal matrix given by
\begin{equation}
\mathbf{D}=\mathbf{I}_{2^{L-1}}\otimes\mathbf{H}_{2},
\end{equation}
i.e., having $2^{L-1}$ times the $2\times2$ block $\mathbf{H}_{2}$ on the diagonal.
\label{pro:hada}
\end{proposition}
\begin{proof}
Please, refer to Appendix~E.
\end{proof}

As a direct consequence of the Hadamard matrix decomposition in Proposition~\ref{pro:hada}, we obtain the following corollary.
\begin{corollary}
A microwave network that performs the Hadamard transform of its input signal $\mathbf{u}$, i.e., computes $\mathbf{H}_N\mathbf{u}/2$, with $N=2^L$, $L\in\mathbb{N}$, is implementable with hybrid couplers and phase shifters.
\label{cor:hada}
\end{corollary}
\begin{proof}
A microwave network performs the Hadamard transform of its input when it is a matched network with transmission scattering matrix being the $N\times N$ Hadamard matrix $\mathbf{H}_N$.
Such a network can therefore be implemented with hybrid couplers and phase shifters if and only if $\mathbf{H}_N$ can be decomposed as in Theorem~\ref{the:1}.
Interestingly, this is the case, since the decomposition of $\mathbf{H}_N$ provided by Proposition~\ref{pro:hada} is in the form required by Theorem~\ref{the:1}.
Note that the decomposition is already in the form $\mathbf{H}_{N}=\mathbf{P}_{L}\mathbf{D}_{L}\cdots\mathbf{P}_{2}\mathbf{D}_{2}\mathbf{P}_{1}\mathbf{D}_{1}\mathbf{P}_{0}$, hence we need to verify that the matrices $\mathbf{P}_{\ell}$ and $\mathbf{D}_{\ell}$ are in the correct form.
First, all matrices $\mathbf{P}_{\ell}$ given by \eqref{eq:Pell-hada} are permutation matrices since they are products of permutation matrices, fulfilling Theorem~\ref{the:1}.
Second, all matrices $\mathbf{D}_{\ell}$ are given by $\mathbf{D}_{\ell}=\mathbf{I}_{2^{L-1}}\otimes\mathbf{H}_{2}$, satisfying Theorem~\ref{the:1} since they are block diagonal matrices where $C_\ell=N/2$ and $S_\ell=0$, and $\theta_{\ell,c,11}=\theta_{\ell,c,12}=\theta_{\ell,c,21}=0$, for $c=1,\ldots,N/2$.
\end{proof}

Corollary~\ref{cor:hada} states that the Hadamard transform of a given vector can be computed in the analog domain with a network of hybrid couplers and phase shifters.
Such a network can be systematically constructed for any Hadamard transform with size $N$ power of two, following the Proof of sufficiency of Theorem~\ref{the:1}, similar to what was discussed for the \gls{dft}.
Specifically, the Hadamard matrix decomposition in Proposition~\ref{pro:hada} requires $C_\ell=N/2$ and $S_\ell=0$, indicating that in the subnetworks implementing the matrices $\mathbf{D}_{\ell}$ there are $N/2$ hybrid couplers in parallel.
Furthermore, $\theta_{\ell,c,11}=\theta_{\ell,c,12}=\theta_{\ell,c,21}=0$ give that the phase shifts in \eqref{eq:alpha}-\eqref{eq:gamma} become $\alpha_{\ell,c}=0$, $\beta_{\ell,c}=\pi/2$, and $\gamma_{\ell,c}=-\pi/2$, for $c=1,\ldots,N/2$ and $\ell=1,\ldots,L$.
We can therefore construct a network in the form of the one in Fig.~\ref{fig:sufficient} that computes the Hadamard transform of its input, as represented in Fig.~\ref{fig:hada}.
As observed for the \gls{dft}, for computing the $N\times N$ Hadamard transform we need $2^{L-1}$ hybrid couplers to implement each matrix $\mathbf{D}$, and there are $L$ of them, resulting in a total of
\begin{equation}
C^{\text{Hadamard}}=2^{L-1}L=\frac{N}{2}\log_2\left(N\right).
\end{equation}
This indicates that this circuit exhibits the same scaling behavior as the circuit for the \gls{dft}.

\section{Analog Computation of the Haar Transform}
\label{sec:haar}

\begin{figure*}[t]
\centering
\includegraphics[width=0.96\textwidth]{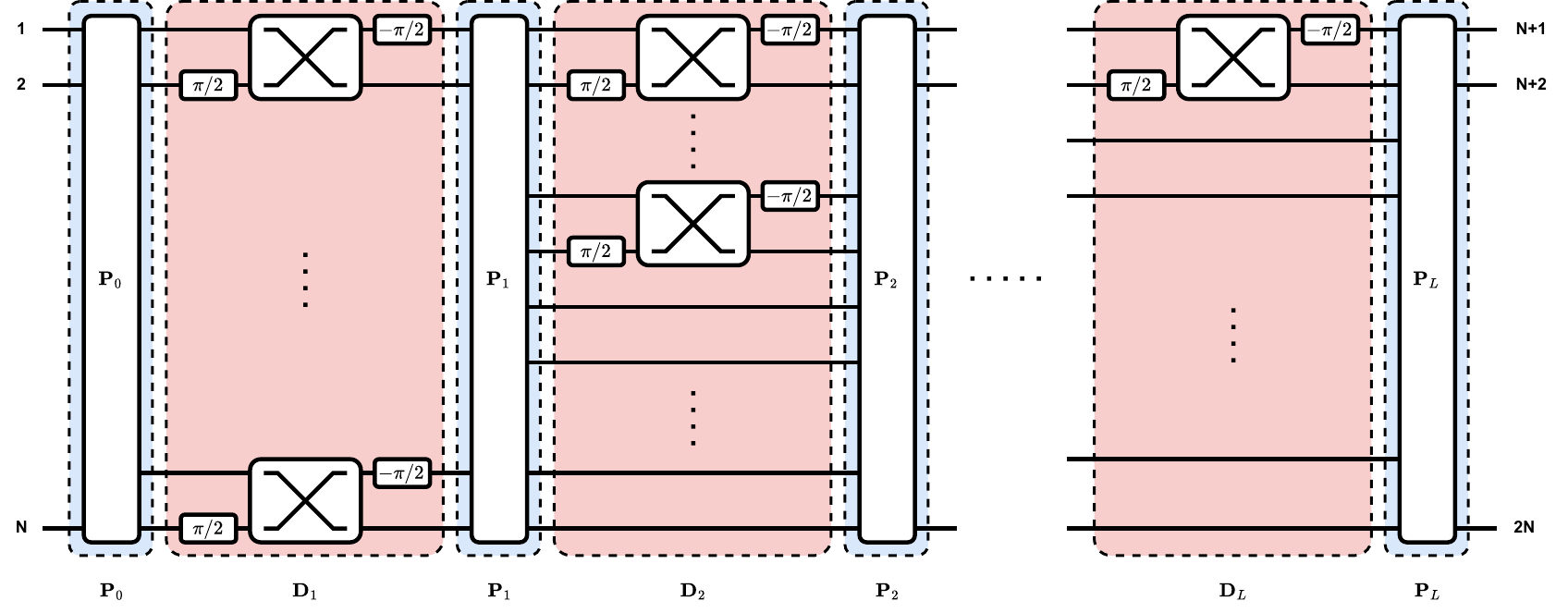}
\caption{Network of hybrid couplers and phase shifters performing the $N\times N$ Haar transform of the input, with $N=2^L$.
The permutation matrices $\mathbf{P}_\ell$ are given by \eqref{eq:Pell-haar}.}
\label{fig:haar}
\end{figure*}

In this section, we demonstrate how the Haar transform, which is a popular Wavelet transform, can be implemented using hybrid couplers and phase shifters, following an approach analogous to that employed for the \gls{dft} and Hadamard transform.
We begin by introducing the Haar matrix through the following recursive definition and then show that it can be decomposed as required by Theorem~\ref{the:1}.
\begin{definition}
(Haar matrix)
The $N\times N$ Haar matrix $\mathbf{W}_N\in\mathbb{R}^{N\times N}$, with $N=2^L$, $L\in\mathbb{N}$, is recursively defined as
\begin{equation}
\mathbf{W}_N=\frac{1}{\sqrt{2}}
\begin{bmatrix}
\mathbf{W}_{N/2} & \mathbf{W}_{N/2}\\
\mathbf{I}_{N/2} & -\mathbf{I}_{N/2}
\end{bmatrix},
\end{equation}
when $L>1$, and $\mathbf{W}_2=[[1,1]^T,[1,-1]^T]/\sqrt{2}$.\footnote{Often, the Haar matrix is also defined as
$\mathbf{W}_N=[[\mathbf{W}_{N/2}\otimes[1,1]]^T,[\mathbf{I}_{N/2}\otimes[1,-1]]^T]^T/\sqrt{2}$,
which is equivalent to our definition up to a permutation of the rows and columns \cite{fin77}.
Our definition is adopted for analytical convenience in the following discussion.}
\end{definition}
Note that Haar matrices are unitary, and the matrices $\mathbf{W}_2$, $\mathbf{W}_4$, and $\mathbf{W}_8$ are reported as three examples:
{\scriptsize
\begin{equation}
\mathbf{W}_2=\frac{1}{\sqrt{2}}
\begin{bmatrix}
1 & 1\\
1 & -1
\end{bmatrix},\;
\mathbf{W}_4=\frac{1}{\sqrt{4}}
\begin{bmatrix}
1 &  1 &  1 &  1\\
1 & -1 &  1 & -1\\
\sqrt{2} & 0 & -\sqrt{2} & 0\\
0 & \sqrt{2} & 0 & -\sqrt{2}
\end{bmatrix},
\end{equation}
\begin{equation}
\mathbf{W}_8=\frac{1}{\sqrt{8}}
\begin{bmatrix}
1 &  1 &  1 &  1 &  1 &  1 &  1 &  1\\
1 & -1 &  1 & -1 &  1 & -1 &  1 & -1\\
\sqrt{2} & 0 & -\sqrt{2} & 0 & \sqrt{2} & 0 & -\sqrt{2} & 0\\
0 & \sqrt{2} & 0 & -\sqrt{2} & 0 & \sqrt{2} & 0 & -\sqrt{2}\\
2 &  0 &  0 &  0 & -2 &  0 &  0 &  0\\
0 &  2 &  0 &  0 &  0 & -2 &  0 &  0\\
0 &  0 &  2 &  0 &  0 &  0 & -2 &  0\\
0 &  0 &  0 &  2 &  0 &  0 &  0 & -2
\end{bmatrix}.
\end{equation}}

By using the recursive definition of Haar matrix, we show in the following proposition that it can be decomposed as required by Theorem~\ref{the:1}.
\begin{proposition}
The $N\times N$ Haar matrix $\mathbf{W}_N$, with $N=2^L$, $L\in\mathbb{N}$, can be decomposed as
\begin{equation}
\mathbf{W}_{N}=\mathbf{P}_{L}\mathbf{D}_{L}\cdots\mathbf{P}_{2}\mathbf{D}_{2}\mathbf{P}_{1}\mathbf{D}_{1}\mathbf{P}_{0}.\label{eq:theor-haar}
\end{equation}
The matrix $\mathbf{P}_{\ell}\in\{0,1\}^{N\times N}$ is a permutation matrix given by $\mathbf{P}_{0}=\tilde{\mathbf{P}}_{2^L}^T$,
\begin{equation}
\mathbf{P}_{\ell}=\text{diag}\left(\text{diag}\left(\tilde{\mathbf{P}}_{2^{L-\ell}}^T,\mathbf{I}_{2^{L-\ell}}\right)\tilde{\mathbf{P}}_{2^{L-\ell+1}},\mathbf{I}_{2^{L}-2^{L-\ell+1}}\right),\label{eq:Pell-haar}
\end{equation}
for $\ell=1,\ldots,L-1$, and $\mathbf{P}_{L}=\mathbf{I}_{2^{L}}$.
Besides, $\mathbf{D}_{\ell}\in\mathbb{C}^{N\times N}$ is a block diagonal matrix given by
\begin{equation}
\mathbf{D}_{\ell}=\text{diag}\left(\mathbf{I}_{2^{L-\ell}}\otimes\mathbf{W}_{2},\mathbf{I}_{2^{L}-2^{L-\ell+1}}\right),\label{eq:Dell-haar}
\end{equation}
for $\ell=1,\ldots,L$, having $2^{L-\ell}$ times the $2\times2$ block $\mathbf{W}_{2}$ on the diagonal, and $2^{L}-2^{L-\ell+1}$ ones.
\label{pro:haar}
\end{proposition}
\begin{proof}
Please, refer to Appendix~F.
\end{proof}

Using the Haar matrix decomposition in Proposition~\ref{pro:haar}, we derive the following corollary.
\begin{corollary}
A microwave network that performs the Haar transform of its input signal $\mathbf{u}$, i.e., computes $\mathbf{W}_N\mathbf{u}/2$, with $N=2^L$, $L\in\mathbb{N}$, is implementable with hybrid couplers and phase shifters.
\label{cor:haar}
\end{corollary}
\begin{proof}
A microwave network performs the Haar transform of its input when it is a matched network with transmission scattering matrix being the $N\times N$ Haar matrix $\mathbf{W}_N$.
A network can therefore be implemented with hybrid couplers and phase shifters if and only if $\mathbf{W}_N$ can be decomposed as in Theorem~\ref{the:1}.
To show this, we prove that the decomposition of $\mathbf{W}_N$ provided by Proposition~\ref{pro:haar} is in the form required by Theorem~\ref{the:1}.
Since the decomposition of Proposition~\ref{pro:haar} is already in the form $\mathbf{W}_{N}=\mathbf{P}_{L}\mathbf{D}_{L}\cdots\mathbf{P}_{2}\mathbf{D}_{2}\mathbf{P}_{1}\mathbf{D}_{1}\mathbf{P}_{0}$, we need to verify that the matrices $\mathbf{P}_{\ell}$ and $\mathbf{D}_{\ell}$ are in the correct form.
First, all matrices $\mathbf{P}_{\ell}$ given by \eqref{eq:Pell-haar} are permutation matrices, satisfying Theorem~\ref{the:1}.
Second, all matrices $\mathbf{D}_{\ell}$ given by \eqref{eq:Dell-haar} fulfill Theorem~\ref{the:1} since they are block diagonal matrices where $C_\ell=2^{L-\ell}$ and $S_\ell=2^L-2^{L-\ell+1}$, and the phase shifts are $\theta_{\ell,c,11}=\theta_{\ell,c,12}=\theta_{\ell,c,21}=0$, for $c=1,\ldots,2^{L-\ell}$, and $\theta_{\ell,s}=0$, for $s=1,\ldots,2^L-2^{L-\ell+1}$.
\end{proof}

The Haar transform of a given vector can therefore be computed in the analog domain with a network of hybrid couplers and phase shifters.
Such a network can be systematically constructed for any Haar transform with size $N$ power of two, following the Proof of sufficiency of Theorem~\ref{the:1}, similar to what was discussed for the \gls{dft} and the Hadamard transform.
Specifically, the Haar matrix decomposition in Proposition~\ref{pro:dft} requires $C_\ell=2^{L-\ell}$ and $S_\ell=2^L-2^{L-\ell+1}$, indicating that the subnetworks implementing the matrices $\mathbf{D}_{\ell}$ contain $2^{L-\ell}$ hybrid couplers in parallel with $2^L-2^{L-\ell+1}$ interconnections.
Furthermore, $\theta_{\ell,c,11}=\theta_{\ell,c,12}=\theta_{\ell,c,21}=0$ give that the phase shifts in \eqref{eq:alpha}-\eqref{eq:gamma} become $\alpha_{\ell,c}=0$, $\beta_{\ell,c}=\pi/2$, and $\gamma_{\ell,c}=-\pi/2$, for $c=1,\ldots,2^{L-\ell}$ and $\ell=1,\ldots,L$.
We can therefore construct a network in the form of the one in Fig.~\ref{fig:sufficient} that computes the Haar transform, as represented in Fig.~\ref{fig:haar}.
To compute the Haar transform, the network of hybrid couplers and phase shifters has $2^{L-\ell}$ in the $\ell$th layer implementing $\mathbf{D}_\ell$, for $\ell=1,\ldots,L$.
Therefore, the required number of hybrid couplers is
\begin{equation}
C^{\text{Haar}}=\sum_{\ell=1}^L2^{L-\ell}=2^L-1=N-1,
\end{equation}
meaning that the Haar transform can be implemented with a simpler circuit than the \gls{dft} and the Hadamard transform, which both required $N\log_2(N)/2$ hybrid couplers.

\section{Practical Applications of Analog DFT, Hadamard Transform, and Haar Transform}
\label{sec:app}

We have shown how to systematically design networks of hybrid couplers and phase shifters to compute the \gls{dft}, Hadamard transform and Haar transform directly at \gls{rf}.
In this section, we highlight practical applications in which these analog-domain transforms are useful.
A key scenario is provided by wireless multi-antenna transceivers, where \gls{rf} signals across an antenna array are linearly combined after reception or linearly precoded before transmission.
The proposed analog computing networks enable these operations to be performed directly in the analog \gls{rf} domain, reducing processing latency and avoiding the need for digital signal processing.

A representative application of the analog \gls{dft} processor is \gls{doa} estimation.
When the signals received from an antenna array are transformed by a \gls{dft}, the received energy is mapped onto a set of spatial beams associated with different angular directions.
The dominant output therefore indicates the most likely \gls{doa}, allowing angular estimation to be performed efficiently without first digitizing every antenna signal independently.
A related \gls{rf} processing approach for this task has also been studied in \cite{an24}.
More generally, the analog \gls{dft} is well suited to applications where we desire single-lobed beams.

The Hadamard transform can be employed similarly whenever an alternative beamspace representation is advantageous, since the resulting beam patterns differ from those of the \gls{dft}.
In particular, as shown in \cite{kim09}, the Hadamard transform can generate two-lobed beams rather than the conventional single-lobed beams.
This makes the analog Hadamard transform attractive in scenarios where such beam shapes are beneficial for multi-target sensing or multi-user beamforming.

A different application emerges for the Haar transform in the receiver front-end of large antenna arrays.
Consider an array with $N$ receiving antennas whose outputs are first processed by the proposed Haar network before reaching the \gls{rf} chains.
Unlike the \gls{dft} and Hadamard transform, the Haar transform produces outputs that are organized hierarchically.
The first outputs of the transform capture coarse spatial information by combining large groups of adjacent antennas, while subsequent outputs progressively capture finer spatial details.
As a result, retaining only the first outputs of the Haar transform provides a meaningful low-resolution approximation of the received wavefront.
This hierarchical representation of the received signals with multiple levels of resolution enables a flexible receiver architecture.
Instead of equipping all $N$ outputs of the Haar transform with dedicated \gls{rf} chains, only a subset of the outputs may be considered.
For example, connecting only $N/2$, $N/4$, or $N/8$ \gls{rf} chains yields progressively coarser spatial resolution while significantly reducing the number of mixers, \glspl{adc}, and therefore power consumption.
Although discarding the higher-resolution outputs results in a loss of spatial information, the receiver retains a compressed representation of the received wavefront at a resolution determined by the number of selected outputs.

\section{Experimental Verification}
\label{sec:experiment}

Based on the above theoretical analysis, the \gls{dft}, the Hadamard transform, and the Haar transform can all be implemented using microwave networks composed of hybrid couplers and fixed phase-shifting transmission lines.
To provide experimental validation of the proposed approach, we selected the $4\times4$ \gls{dft} as a representative case and fabricated and measured a corresponding multiport prototype, as shown in Fig.~\ref{fig:photo}.
Implementing the Hadamard or Haar transform would not introduce any additional practical difficulties compared with the \gls{dft}.
The \gls{dft} microwave circuit was implemented using microstrip transmission lines on a single-layer dielectric substrate with a relative permittivity of 3.0, a loss tangent of 0.0017, and a thickness of 0.508~mm.
Eight standard SMA connectors were mounted at the input and output ports to facilitate the S-parameter measurements.
The measurement setup was limited to a four-port \gls{vna}, while the circuit has eight ports.
Therefore, all unused ports were terminated with matched loads during the measurements.
This ensured accurate extraction of the complete multiport scattering matrix.
All the measured data from the \gls{vna} were combined to form multiple $8\times8$ scattering matrices, each at a different frequency, by using MATLAB.

\begin{figure}[t]
\centering
\includegraphics[width=0.48\textwidth]{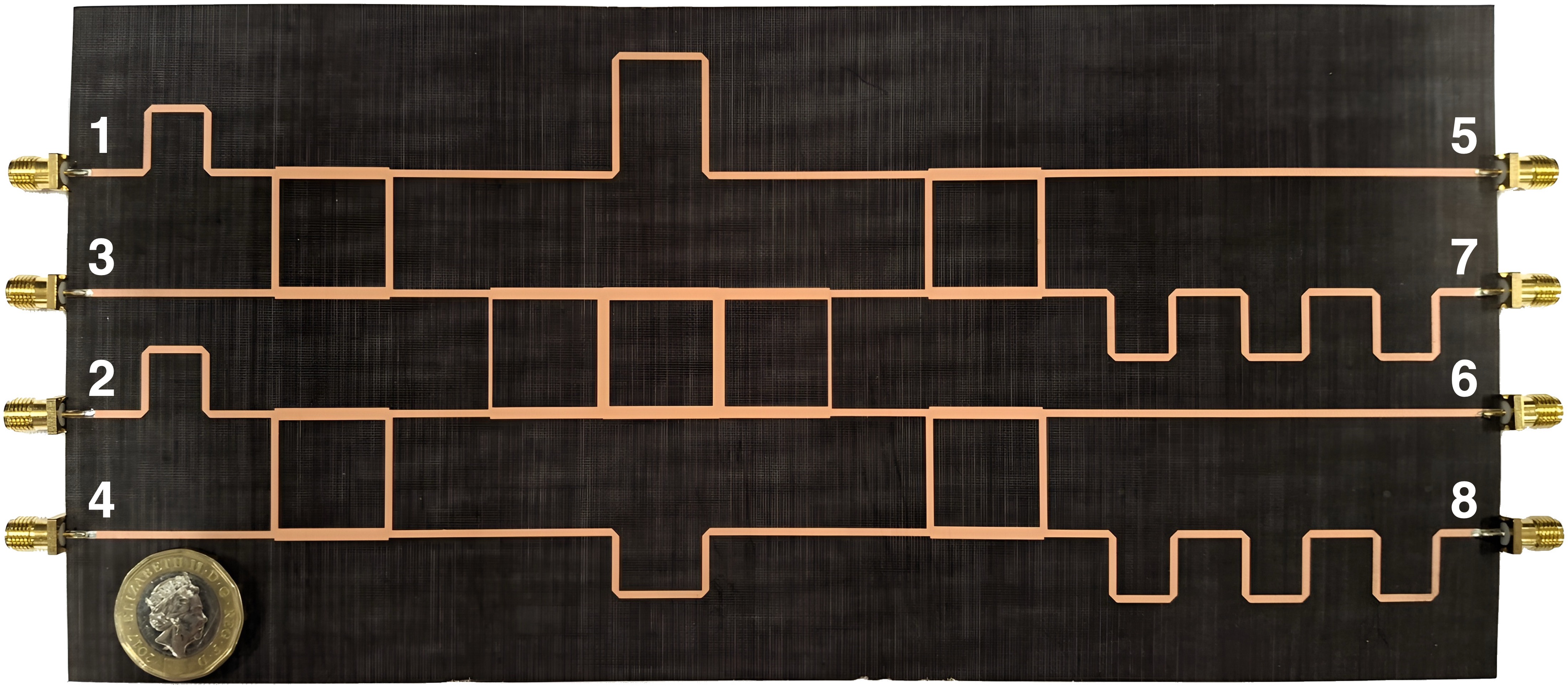}
\caption{Photo of the fabricated microwave circuit for DFT calculation (the board is $260\times120$~mm).}
\label{fig:photo}
\end{figure}

\begin{figure*}[t]
\centering
\includegraphics[width=0.245\textwidth]{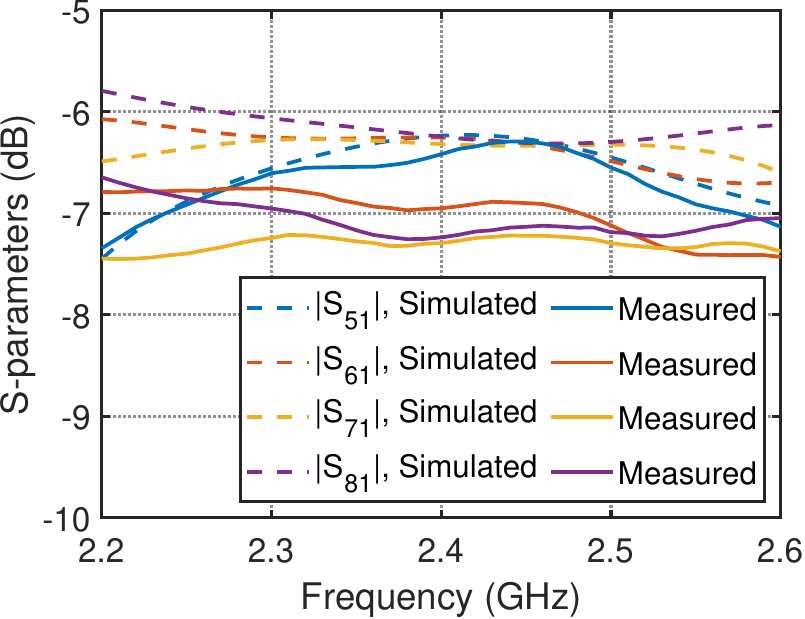}
\includegraphics[width=0.245\textwidth]{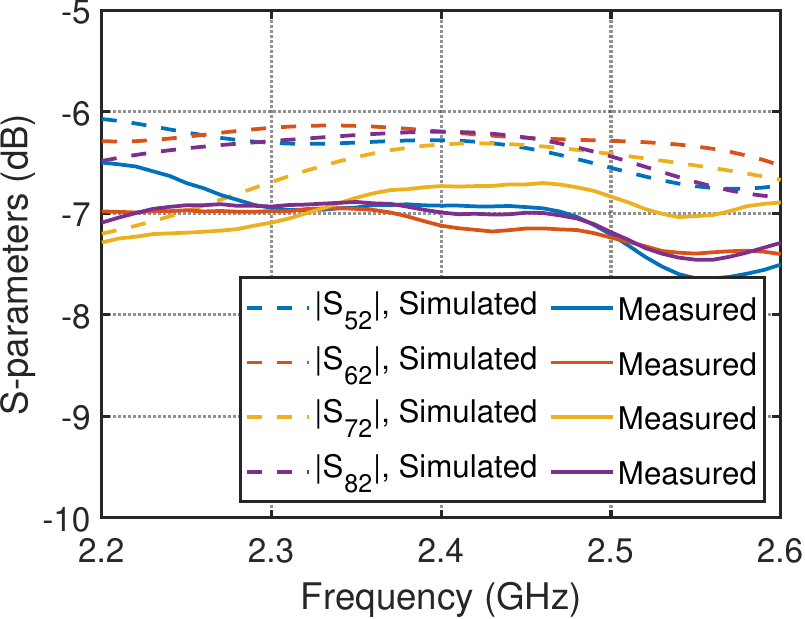}
\includegraphics[width=0.245\textwidth]{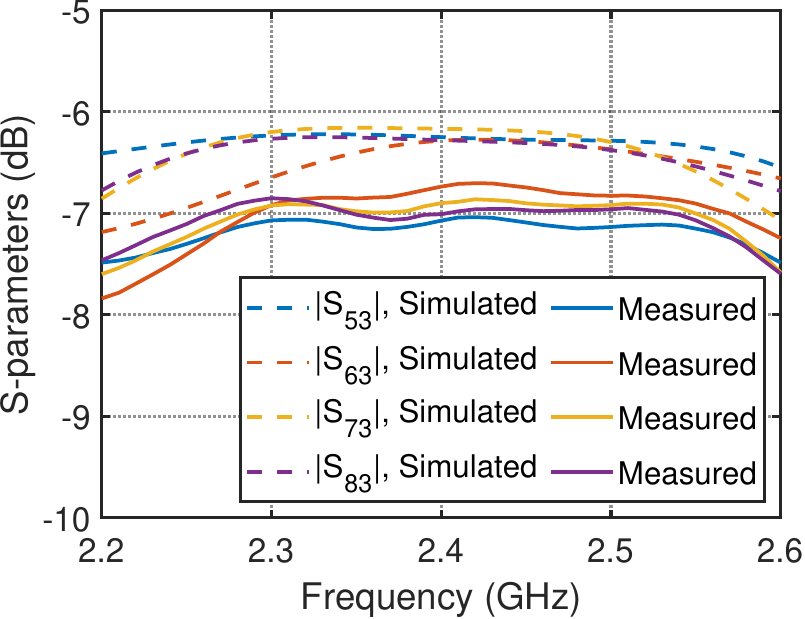}
\includegraphics[width=0.245\textwidth]{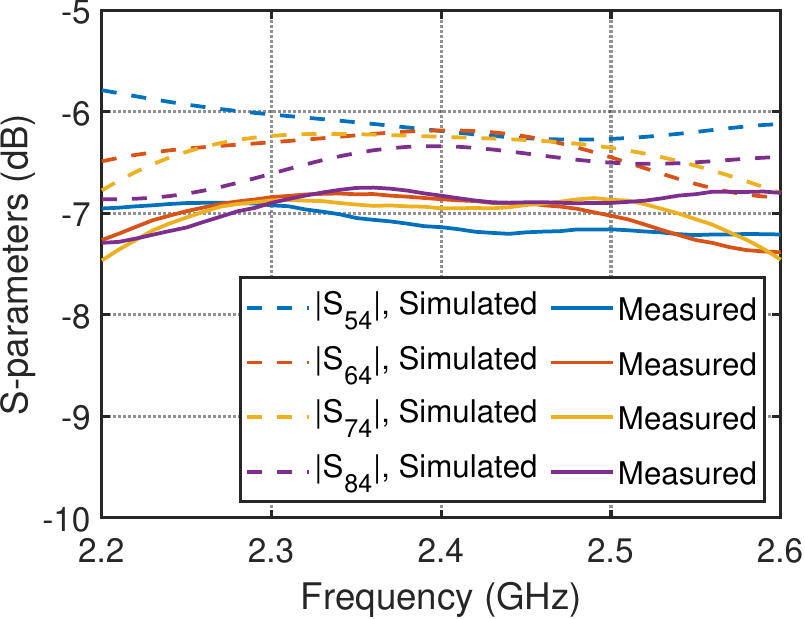}
\caption{Simulated and measured transmission coefficients (the ideal magnitudes of the entries of $\mathbf{F}_4$ are $1/2$, i.e., $-6$~dB).}
\label{fig:S}
\end{figure*}

The comparison between the simulated and measured transmission coefficients is given in Fig.~\ref{fig:S}.
Specifically, we report the magnitudes of the entries of the transmission scattering matrix $\mathbf{S}_{21}\in\mathbb{C}^{4\times4}$, defined as
\begin{equation}
\mathbf{S}_{21}=
\begin{bmatrix}
S_{51}&S_{52}&S_{53}&S_{54}\\
S_{61}&S_{62}&S_{63}&S_{64}\\
S_{71}&S_{72}&S_{73}&S_{74}\\
S_{81}&S_{82}&S_{83}&S_{84}
\end{bmatrix}.
\end{equation}
As can be seen from Fig.~\ref{fig:S}, the measured results show additional loss compared to the simulated results obtained using the full-wave electromagnetic solver Ansys High Frequency Structure Simulator (HFSS), in which the complete three-dimensional structure was modelled.
This is caused by the loss of the substrate and the fabrication error.
Compared to the ideal magnitude of the entries of the $4\times4$ \gls{dft} matrix $\mathbf{F}_4$, which is $1/2$, i.e., $-6$~dB, the insertion loss is approximately 1~dB and therefore very modest.
This insertion loss may affect the \gls{snr} of the output signals.
In the following, we evaluate the accuracy of the circuit in performing the desired \gls{dft} operation.

\subsection{Performance Evaluation Without Calibration}

To evaluate the performance of the fabricated microwave circuit, a comparison is conducted between the theoretical responses and the calculated results derived from the measured S-parameters.
Given an excitation vector $\mathbf{a}\in\mathbb{C}^{4\times1}$, the theoretical (ideal) response $\mathbf{b}\in\mathbb{C}^{4\times1}$ is obtained by multiplying the excitation vector by the unitary \gls{dft} matrix\footnote{With a slight abuse of notation, we denote as $\mathbf{a}$ and $\mathbf{b}$ in this section what was denoted as $\mathbf{a}_1$ and $\mathbf{b}_2$ in Section~\ref{sec:model}.}, namely
\begin{equation}
\mathbf{b}=\mathbf{F}_4\mathbf{a},\label{eq:ideal}
\end{equation}
where $\mathbf{F}_4$ is the $4\times4$ \gls{dft} matrix with elements defined as in \eqref{eq:dft}.
Besides, the calculated output based on the raw measured S-parameters $\mathbf{S}_{21}$ is
\begin{equation}
\mathbf{b}_{\text{raw}}=\mathbf{S}_{21}\mathbf{a}.\label{eq:raw}
\end{equation}

\begin{table}[t]
\centering
\caption{Comparison between theoretical and calculated output when the input is $\mathbf{a}=[e^{j\pi/8},e^{j2\pi/8},e^{j3\pi/8},e^{j4\pi/8}]$.}
\label{tab}
\begin{tabular}{@{}ccccc@{}}
\toprule
\multirow{2}{*}{} & \multicolumn{2}{c}{Theoretical output} & \multicolumn{2}{c}{Calculated output} \\
\cmidrule(l){2-5}
& Magnitude & Phase ($^\circ$) & \begin{tabular}[c]{@{}c@{}}Magnitude\\(Error)\end{tabular} & \begin{tabular}[c]{@{}c@{}}Phase ($^\circ$)\\(Error)\end{tabular}\\
\midrule
Entry 1 & $1.8123$ & $+56.25$  & \begin{tabular}[c]{@{}c@{}}$1.6256$\\($0.187$)\end{tabular} & \begin{tabular}[c]{@{}c@{}}$+100.03$\\($-43.27-0.51$)\end{tabular}\\
\\
Entry 2 & $0.6364$ & $-78.75$ & \begin{tabular}[c]{@{}c@{}}$0.5959$\\($0.041$)\end{tabular} & \begin{tabular}[c]{@{}c@{}}$-34.030$\\($-43.27-1.45$)\end{tabular}\\
\\
Entry 3 & $0.3605$ & $-33.75$ & \begin{tabular}[c]{@{}c@{}}$0.3377$\\($0.023$)\end{tabular} & \begin{tabular}[c]{@{}c@{}}$+8.9531$\\($-43.27+0.57$)\end{tabular}\\
\\
Entry 4 & $0.4252$ & $+11.25$ & \begin{tabular}[c]{@{}c@{}}$0.3478$\\($0.077$)\end{tabular} & \begin{tabular}[c]{@{}c@{}}$+53.132$\\($-43.27+1.39$)\end{tabular}\\
\bottomrule
\end{tabular}
\end{table}

One representative excitation vector is first selected to enable a direct comparison between the theoretical predictions and the results derived from the measured S-parameters.
Then, a Monte Carlo analysis is conducted to statistically characterize the performance of the network under various input vectors.
This approach provides a comprehensive assessment of amplitude and phase deviation due to insertion loss and fabrication errors.

The direct comparison results are summarized in Table~\ref{tab}, obtained with the excitation vector $\mathbf{a}=[e^{j\pi/8},e^{j2\pi/8},e^{j3\pi/8},e^{j4\pi/8}]$, assuming we are operating at 2.4~GHz.
As observed, the magnitude differences between the theoretical and calculated results are $0.187$, $0.041$, $0.023$, and $0.077$ at the four ports.
The corresponding phase differences are $-43.78^\circ$, $-44.72^\circ$, $-42.70^\circ$, and $-41.88^\circ$.
At first glance, the phase discrepancies appear relatively large.
However, these phase differences contain a common offset of approximately $-43.27^\circ$, which is obtained as their average value.
After removing this common phase term, the residual phase errors are $-0.51^\circ$, $-1.45^\circ$, $+0.57^\circ$, and $+1.39^\circ$, which are significantly smaller and due to fabrication tolerances.
This common phase offset is mainly introduced by the SMA connectors.
More importantly, it can be effectively eliminated through the calibration procedure introduced in the following section.
Note that the theoretical outputs in Table~\ref{tab} can also be exactly achieved by a lossless network of ideal hybrid couplers and phase shifters, i.e., characterized by their ideal scattering matrices as discussed in Section~\ref{sec:model}.

The Monte Carlo analysis is also conducted to check the overall performance under random inputs.
During the calculation, the ideal \gls{dft} output is obtained by using \eqref{eq:ideal} while the raw output is obtained through \eqref{eq:raw}, depending on $\mathbf{S}_{21}$, which is the measured S-parameter matrix of the microwave circuit.
The Monte Carlo set up is $\mathbf{a}=[A_1e^{j\phi_1},A_2e^{j\phi_2},A_3e^{j\phi_3},A_4e^{j\phi_4}]$, where $A_n\sim U[0,1]$ and $\phi_n\sim U[0,2\pi)$, for $n=1,2,3,4$.
The error between the ideal and the calculated output is evaluated using two metrics, namely the \gls{nmse}, defined as
\begin{equation}
E_{\text{nmse}}\left(\mathbf{a}\right)=\frac{\left\Vert\mathbf{b}_{\text{raw}}-\mathbf{b}\right\Vert^2}{\left\Vert\mathbf{b}\right\Vert^2},
\end{equation}
and the cosine distance, defined as
\begin{equation}
E_{\text{cos}}\left(\mathbf{a}\right)=1-\rho\left(\mathbf{b}_{\text{raw}},\mathbf{b}\right)=1-\frac{\left\vert\mathbf{b}_{\text{raw}}^H\mathbf{b}\right\vert}{\left\Vert\mathbf{b}_{\text{raw}}\right\Vert\left\Vert\mathbf{b}\right\Vert},
\end{equation}
where $\rho(\mathbf{b}_{\text{raw}},\mathbf{b})$ is the cosine similarity between the vectors $\mathbf{b}_{\text{raw}}$ and $\mathbf{b}$.
The average error is then evaluated over $K$ randomly generated excitation vectors using a Monte Carlo procedure.
It is computed as
\begin{equation}
\bar{E}=\frac{1}{K}\sum_{k=1}^KE\left(\mathbf{a}_k\right),
\end{equation}
where $\mathbf{a}_k$ denotes the k-th randomly generated input vector and $E$ is either $E_{\text{nmse}}$ or $E_{\text{cos}}$.
Similarly, the standard deviation of the error distribution is computed as
\begin{equation}
\sigma_{E}=\sqrt{\frac{1}{K}\sum_{k=1}^K\left(E\left(\mathbf{a}_k\right)-\bar{E}\right)^2},
\end{equation}
where the error $E$ is $E_{\text{nmse}}$ or $E_{\text{cos}}$.

Using $K=10^4$ randomly generated excitation vectors, the \gls{nmse} and cosine distance obtained by using the raw measured scattering matrix are $57.018\%\pm2.318\%$ and $0.066\%\pm0.028\%$, respectively.
The relatively large \gls{nmse} observed here is primarily due to the inclusion of the global phase shift introduced by the SMA connectors.
This effect can be readily mitigated by removing the phase contribution of the SMA connectors through calibration.
In contrast, the cosine distance is inherently insensitive to such global phase rotations and therefore provides a more accurate measure of the system’s functional fidelity.
The low cosine distance confirms that the proposed design preserves the intended transformation structure with high accuracy, highlighting its robustness and effectiveness.

\subsection{Performance Evaluation With Phase-Based Calibration}

The results above indicate that the observed discrepancy is largely dominated by a systematic global phase offset introduced by the SMA connectors.
This systematic phase deviation can be compensated by multiplying the measured output vector with a constant phase factor $e^{-j\phi_0}$, where $\phi_0$ corresponds to the extracted global phase offset (approximately $-43.27^\circ$ at 2.4~GHz).
After removing this common phase shift, the corrected output becomes
\begin{equation}
\mathbf{b}_{\text{phase}}=e^{-j\phi_0}\mathbf{b}_{\text{raw}},
\end{equation}
where $\phi_0=-43.27^\circ$.
Using the same Monte Carlo procedure, the \gls{nmse} is reduced to $1.456\%\pm0.179\%$, reflecting the very low intrinsic error of the fabricated microwave circuit after compensating for the systematic phase offset.
Meanwhile, the cosine distance remains unchanged, as expected, since it is inherently insensitive to a global phase rotation.
While this phase-based calibration significantly improves the agreement with the ideal \gls{dft} transformation, it corrects only a uniform phase rotation shared by all output ports.
The remaining error originates from port-dependent amplitude and phase imbalances introduced by fabrication tolerances, transmission line mismatches, and connector imperfections.

\subsection{Performance Evaluation With Vector-Based Calibration}

\begin{algorithm}[t]
\begin{algorithmic}[1]
\REQUIRE Measured transmission matrix $\mathbf{M}_{\text{meas}}$, DFT matrix $\mathbf{F}_4$, regularization weight $\lambda$, number of starts $S$, and phase perturbation range $\Delta$.
\ENSURE Calibrated diagonal matrices $\mathbf{D}_{\text{in}}$, $\mathbf{D}_{\text{out}}$, and calibrated transmission matrix $\mathbf{M}_{\text{calib}}$.\\
\textbf{Step 1: Initial guess in closed-form}
\STATE Compute initial vectors $\mathbf{d}_{\text{in}}^{(0)}$ and $\mathbf{d}_{\text{out}}^{(0)}$ using column-wise and row-wise least-squares estimation.\\
\textbf{Step 2: Multi-start Quasi-Newton optimization}
\FOR{$s=1$ \TO $S$}
    \STATE Generate phase-perturbed starting vectors:\\
    $\mathbf{d}_{\text{in}}^{(s)}\leftarrow e^{j\delta}\mathbf{d}_{\text{in}}^{(0)}$, $\mathbf{d}_{\text{out}}^{(s)}\leftarrow e^{j\delta}\mathbf{d}_{\text{out}}^{(0)}$,\\
    where $\delta\sim U[-\Delta/2,+\Delta/2]$.
    \STATE Solve $\min J(\mathbf{d}_{\text{in}},\mathbf{d}_{\text{out}})$ using the Quasi-Newton method starting from $(\mathbf{d}_{\text{in}}^{(s)},\mathbf{d}_{\text{out}}^{(s)})$.
\ENDFOR
\STATE Select the pair $(\mathbf{d}_{\text{in}},\mathbf{d}_{\text{out}})$ that yields the minimum objective value among all $S$ starts.
\STATE Return $\mathbf{D}_{\text{in}}=\text{diag}(\mathbf{d}_{\text{in}})$,
$\mathbf{D}_{\text{out}}=\text{diag}(\mathbf{d}_{\text{out}})$,\\
$\mathbf{M}_{\text{calib}}=\mathbf{D}_{\text{out}}\mathbf{M}_{\text{meas}}\mathbf{D}_{\text{in}}$.
\end{algorithmic}
\caption{Vector-based calibration.}
\label{alg}
\end{algorithm}

To further reduce the residual discrepancies after global phase compensation, a diagonal input-output calibration is applied to the measured transmission scattering matrix $\mathbf{S}_{21}$.
At the target frequency $f_0$, the transmission matrix used for calibration is directly taken as the measured submatrix $\mathbf{M}_{\text{meas}}=\mathbf{S}_{21}(f_0)$.
The calibrated matrix is modeled as
\begin{equation}
\mathbf{M}_{\text{calib}}=\mathbf{D}_{\text{out}}\mathbf{M}_{\text{meas}}\mathbf{D}_{\text{in}},
\end{equation}
with diagonal calibration matrices
\begin{gather}
\mathbf{D}_{\text{in}}=\text{diag}\left(d_{\text{in},1},d_{\text{in},2},d_{\text{in},3},d_{\text{in},4}\right),\\
\mathbf{D}_{\text{out}}=\text{diag}\left(d_{\text{out},1},d_{\text{out},2},d_{\text{out},3},d_{\text{out},4}\right),
\end{gather}
where each diagonal element is a complex coefficient providing per-port amplitude and phase calibration.

The calibration coefficients are obtained by minimizing the regularized non-linear least-squares objective
\begin{multline}
J=\left\Vert\mathbf{M}_{\text{calib}}-\mathbf{F}_4\right\Vert_F^2\\
+\lambda\left(\sum_{n=1}^4\left(\left\vert d_{\text{in},n}\right\vert-1\right)^2+\sum_{n=1}^4\left(\left\vert d_{\text{out},n}\right\vert-1\right)^2\right),\label{eq:J}
\end{multline}
where $\lambda$ controls the regularization strength used to keep $\vert d_{\text{in},n}\vert$ and $\vert d_{\text{out},n}\vert$ close to unity.
The optimization procedure is summarized in Alg.~\ref{alg} and consists of two stages.

First, a closed-form least-squares initialization is constructed.
The input-side coefficients are estimated column-wise by projecting each column of $\mathbf{M}_{\text{meas}}$ onto the corresponding column of the \gls{dft} matrix $\mathbf{F}_4$.
Using this preliminary input correction, the output-side coefficients are then estimated row-wise.
Specifically, the initial value of $d_{\mathrm{in},n}$ is obtained by solving
\begin{equation}
d_{\mathrm{in},n}^{(0)}=\arg\min_d\left\Vert\mathbf{m}_n d-\mathbf{f}_n\right\Vert_2^2,
\end{equation}
for $n=1,2,3,4$, where $\mathbf{m}_n$ and $\mathbf{f}_n$ are the $n$th columns of $\mathbf{M}_{\mathrm{meas}}$ and the target DFT matrix $\mathbf{F}_4$, respectively.
This yields the closed-form solution
\begin{equation}
d_{\mathrm{in},n}^{(0)}=\frac{\mathbf{m}_n^{H}\mathbf{f}_n}{\mathbf{m}_n^{H}\mathbf{m}_n}.
\end{equation}
After constructing $\mathbf{D}_{\text{in}}^{(0)}=\text{diag}(d_{\text{in},1}^{(0)},d_{\text{in},2}^{(0)},d_{\text{in},3}^{(0)},d_{\text{in},4}^{(0)})$, we calculate the intermediate matrix $\mathbf{X}_0=\mathbf{M}_{\mathrm{meas}}\mathbf{D}_{\mathrm{in}}^{(0)}$.
Then, the initial value of $d_{\mathrm{out},n}$ is set by solving
\begin{equation}
d_{\mathrm{out},n}^{(0)}=\arg\min_d\left\Vert\mathbf{x}_n d-\mathbf{f}_n\right\Vert_2^2,
\end{equation}
for $n=1,2,3,4$, where $\mathbf{x}_n$ is the $n$th column of $\mathbf{X}_0$, giving
\begin{equation}
d_{\mathrm{out},n}^{(0)}=\frac{\mathbf{x}_n^{H}\mathbf{f}_n}{\mathbf{x}_n^{H}\mathbf{x}_n}.
\end{equation}

\begin{figure}[t]
\centering
\includegraphics[width=0.38\textwidth]{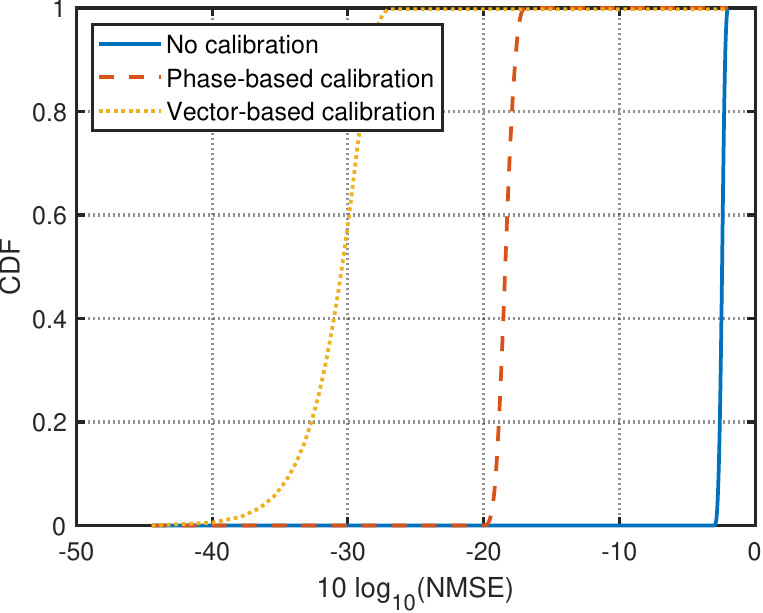}
\caption{Distribution of the NMSE.}
\label{fig:nmse}
\end{figure}

Second, the coefficients are refined through a multi-start optimization.
For each start, a phase-perturbed initialization is generated around the closed-form solution, and the objective function in \eqref{eq:J} is minimized using the Quasi-Newton method (the MATLAB function fminunc).
Among all starting points, the solution yielding the minimum objective value is retained.
This approach effectively performs a vector-based multiport calibration, simultaneously compensating amplitude and phase mismatches across all ports, rather than applying a uniform correction factor.
As a result, both systematic and port-dependent distortions can be significantly reduced, bringing the measured microwave network response closer to the ideal \gls{dft} transformation.
In our setting, we consider a regularization weight $\lambda=10^{-3}$, $S=8$ different starts for the Quasi-Newton method, and a maximum phase perturbation of $\Delta=0.2\pi$.

Figs.~\ref{fig:nmse} and \ref{fig:cos} present the Monte Carlo error distributions corresponding to the raw S-parameter measurements, the SMA phase-compensated results, and the fully calibrated transmission matrix.
For the raw measurements (``No calibration''), the \gls{nmse} and cosine distance are centered at $57.018\%$ and $0.066\%$, respectively.
After compensating for the common phase offset introduced by the SMA connectors, the \gls{nmse} distribution shifts significantly toward lower values and becomes tightly concentrated around $1.456\%$.
This substantial reduction indicates that the majority of the observed deviation originates from a uniform phase bias rather than structural inaccuracies in the circuit itself.
In contrast, the cosine distance remains unchanged, since a global phase rotation does not affect $\vert\mathbf{b}_{\text{phase}}^H\mathbf{b}\vert$ and therefore leaves the cosine distance unchanged.
This behavior further confirms that the underlying transformation structure is accurately preserved and that the proposed microwave computing circuit achieves high functional fidelity.
The proposed diagonal input-output calibration further pushes the error distribution toward zero, achieving an \gls{nmse} and cosine distance of $0.091\%\pm0.039\%$ and $0.037\%\pm0.017\%$, respectively.
Unlike simple global phase compensation, this multiport calibration simultaneously mitigates residual phase discrepancies and corrects amplitude imbalances across all ports.
As a result, both the mean error and its variance are significantly reduced, demonstrating the effectiveness of the proposed calibration in enhancing overall computational accuracy.

\begin{figure}[t]
\centering
\includegraphics[width=0.38\textwidth]{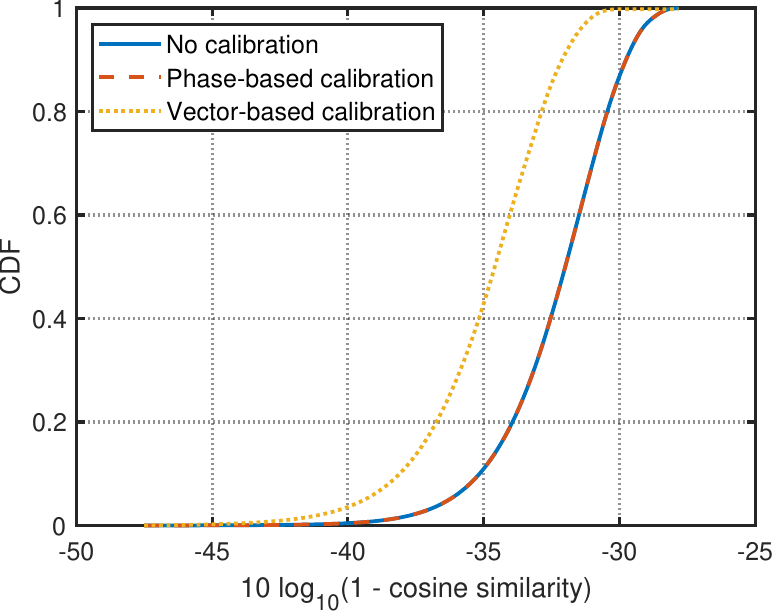}
\caption{Distribution of the cosine distance (1 $-$ cosine similarity).}
\label{fig:cos}
\end{figure}

Importantly, the proposed vector-based diagonal calibration is not merely a numerical post-processing technique but can be physically implemented in analog circuitry.
Since each diagonal element corresponds to an independent complex scaling coefficient, the calibration can be realized using per-port amplitude and phase control components.
Specifically, amplitude corrections can be implemented through microwave amplifiers or attenuators, while phase corrections can be achieved using phase shifters.
In practical implementations, tunable amplifiers, variable attenuators, and electronically controlled phase shifters can be employed to provide flexible and precise adjustment of the complex correction coefficients.
Such an analog realization enables real-time compensation of fabrication tolerances and measurement-induced distortions without requiring digital re-computation.
The proposed diagonal multiport calibration framework is therefore fully compatible with conventional microwave architecture and can be seamlessly integrated into practical systems to achieve adaptive and physically realizable \gls{dft} functionality.

\section{Related Work}
\label{sec:work}

In recent years, there has been a growing interest in performing signal processing operations, both linear and non-linear, directly with microwave signals.
In particular, several research directions have explored the computational capabilities of the linear analog computer illustrated in Fig.~\ref{fig:milac-intro}.
One representative example arises when the microwave network of the \gls{milac} in Fig.~\ref{fig:milac-intro} corresponds to the wireless channel between $N$ transmitting antennas and $M$ receiving antennas.
In this case, the input $\mathbf{u}$ is the signal at the transmitting antennas, the output $\mathbf{v}$ is the signal at the receiving antennas, and the matrix $\mathbf{W}$ is the \gls{mimo} wireless channel.
The operation of computing the product $\mathbf{W}\mathbf{u}$ over a wireless channel by directly observing the received signal $\mathbf{v}$ is known as over-the-air computing \cite{gol13}, and has interesting applications to federated learning \cite{yan20}.
Over-the-air computing allows performing the product $\mathbf{W}\mathbf{u}$ only for a given matrix $\mathbf{W}$ depending on the wireless channel.
To control what transformation $\mathbf{W}$ is computed, tunable metasurfaces, also known as \glspl{ris}, can be included in the propagation environment \cite{del18}.
\Glspl{ris} can also be used to engineer the frequency response of the environment such that linear operations can be computed on the envelope of modulated signals \cite{sol22}.

The microwave network of a \gls{milac} can also be implemented as a \gls{ris} having $N=M$ reflective (or transmissive) elements \cite{wu19,joy25}.
Considering the input signal $\mathbf{u}$ to be the incident signal on the metasurface and the output signal $\mathbf{v}$ the reflected (or refracted) signal, $\mathbf{W}$ depends in this case on the reflective (or transmissive) properties of the metasurface.
Therefore, the metasurface can be designed, or dynamically reconfigured, to compute the product $\mathbf{W}\mathbf{u}$ for some desired matrix $\mathbf{W}$.
Multiple transmissive metasurfaces can also be stacked to realize a more flexible microwave network, under the name of \gls{sim} \cite{an23}.
In a \gls{sim}, $\mathbf{W}$ can be chosen with greater flexibility since it is given by $\mathbf{W}=\mathbf{W}_L\mathbf{H}_L\cdots\mathbf{W}_2\mathbf{H}_2\mathbf{W}_1$, where $\mathbf{W}_\ell$ is the matrix of the $\ell$th metasurface and $\mathbf{H}_\ell$ is the channel between the $(\ell-1)$th and the $\ell$th metasurface, assuming $L$ layers.
Thanks to this flexibility, the refracting properties of a \gls{sim} can be configured to perform the desired linear transform, such as the \gls{dft} \cite{an24}.

Linear microwave networks can also be realized by physically interconnecting the input and output ports through \gls{rf} components.
The ports can be interconnected through power dividers/combiners/couplers and phase shifters, which is the common implementation to realize analog and hybrid beamforming in wireless communications \cite{soh16}.
These components have also been considered in \cite{zhu24,gu24,gao24} to realize reconfigurable matrix-vector products within \glspl{nn}, and in \cite{kes25} to realize universal unitary matrix transformations.
Note that \cite{kes25} proposed a network of hybrid couplers and phase shifters that approximately implements any unitary transformation.
This is different from the goal of this work to characterize all unitary transformations that are exactly implementable with hybrid couplers and phase shifters.
In addition, in our previous work \cite{ner25-1,ner25-2,ner25-3,ner25-4}, we have proposed to implement a \gls{milac} by interconnecting its ports through tunable impedance components.
With this implementation, the \gls{milac} can be used to compute not only the product $\mathbf{W}\mathbf{u}$, but also the \gls{lmmse} estimator (involving matrix inversions) in the analog domain.
Such a \gls{milac} can compute non-linear functions (like matrix inversions), since its matrix $\mathbf{W}$ is a non-linear function of the impedance values in the microwave network.
Furthermore, direct-complex-matrix (DCM) architectures with active phase and gain control have also enabled matrix inversion and iterative algorithms via feedback-configured \gls{rf} networks \cite{tza25}.

An interesting application of analog computing with \gls{em} signals is to implement wave-domain \glspl{nn}.
Early work on diffractive \glspl{nn} demonstrated that cascaded layers of diffractive optical components can perform machine-learning tasks, such as image classification \cite{lin18}.
Building on this idea, a programmable implementation of diffractive \gls{nn} based on multi-layer digital-coding metasurface arrays has been proposed to enable reconfigurable wave-domain neural computation \cite{liu22}.
However, these architectures still operate as linear systems because they lack non-linear activation functions.
To address this limitation, a programmable wave-domain \gls{nn} incorporating non-linearities has been introduced in \cite{gao23}.
Practical training strategies have also been explored, such as the backpropagation-free training approach for wave-domain \glspl{nn} experimentally demonstrated in \cite{mom23}.
More recently, wave-domain neural computing has been further advanced by exploiting intrinsic non-linearities in coupled microwave oscillations \cite{gov25}.

\section{Conclusion and Future Work}
\label{sec:conclusion}

Analog computing with \gls{em} signals is an interesting paradigm that can perform some computations, such as matrix-vector products (also known as linear transformations), ultra-fast and without any digital operation.
In this paper, we investigate what class of linear transformations can be computed in the analog domain with \gls{em} signals propagating through a network of hybrid couplers and phase shifters.
We focus on networks built with these two components because of their fundamental importance and practicality, since they both can be implemented in microstrip.
Among the possible transformations computable by networks of hybrid couplers and phase shifters, we identify three of practical relevance: the \gls{dft}, the Hadamard transform, and the Haar transform.
For each of them, we show how to systematically construct a network that computes the desired transformation of the input signals for any size power of two.
We fabricate a microwave hardware prototype to experimentally verify the proposed theory, implementing the $4\times4$ \gls{dft} using hybrid couplers and phase shifters.
A calibration methodology is proposed to compensate for non-idealities arising from connectors and fabrication tolerances.
Measurements obtained from the prototype confirm correct analog-domain \gls{dft} operation and show strong consistency with theoretical predictions.
Our results establish a framework for implementing low-power signal processing at light speed, without resorting to digital computations.

The three linear transforms considered in this paper (\gls{dft}, Hadamard, and Haar) share a remarkable property: they are unitary transforms that can be recursively defined invoking Kronecker products.
Interestingly, this property has been extensively exploited to develop low-complexity algorithms for the digital computation of such transforms, such as the \gls{fft} algorithm for the \gls{dft} \cite{fin77,reg89,van00}.
Thanks to these algorithms, the \gls{dft}, Hadamard, and Haar transforms can be digitally computed with a number of operations (time complexity) scaling with $\mathcal{O}(N\log N)$, $\mathcal{O}(N\log N)$, and $\mathcal{O}(N)$, respectively.
We have demonstrated in this paper how this Kronecker-based structure enables the computation of those transformations at light speed, i.e., in $\mathcal{O}(1)$, by replacing digital operations with combining operations executed by hybrid couplers.
Since digital operations are replaced by hybrid couplers, the number of hybrid couplers (space complexity) needed to compute the \gls{dft}, Hadamard, and Haar transforms in the analog domain grows with $\mathcal{O}(N\log N)$, $\mathcal{O}(N\log N)$, and $\mathcal{O}(N)$, respectively.
A promising direction for future research is thereby to explore whether the operations of hybrid couplers, used with other non-linear \gls{rf} components, can perform analog-domain computations beyond linear transformations.
Another interesting direction is to investigate how the impact of hardware imperfections scales with the network size.

\section*{Appendix}

\subsection{Proof of Proposition~\ref{pro:ser}}

Denote as $\mathbf{a}_1\in\mathbb{C}^{N\times 1}$ and $\mathbf{b}_1\in\mathbb{C}^{N\times 1}$ the incident and reflected waves on the first $N$ ports of the first network, respectively;
       as $\mathbf{a}_2\in\mathbb{C}^{N\times 1}$ and $\mathbf{b}_2\in\mathbb{C}^{N\times 1}$ the incident and reflected waves on the last $N$ ports of the first network, respectively;
   and as $\mathbf{a}_3\in\mathbb{C}^{N\times 1}$ and $\mathbf{b}_3\in\mathbb{C}^{N\times 1}$ the reflected and incident waves on the last $N$ ports of the second network, respectively.
Given the interconnections between the two networks, $\mathbf{a}_2$ and $\mathbf{b}_2$ are also the reflected and incident waves on the first $N$ ports of the second network, respectively.
Thus, by the definition of scattering matrix \cite[Chapter 4]{poz12}, we have
\begin{equation}
\begin{bmatrix}
\mathbf{b}_{1}\\
\mathbf{b}_{2}
\end{bmatrix}=
\begin{bmatrix}
\mathbf{Q}_{11} & \mathbf{Q}_{12}\\
\mathbf{Q}_{21} & \mathbf{Q}_{22}
\end{bmatrix}
\begin{bmatrix}
\mathbf{a}_{1}\\
\mathbf{a}_{2}
\end{bmatrix},\;
\begin{bmatrix}
\mathbf{a}_{2}\\
\mathbf{a}_{3}
\end{bmatrix}=
\begin{bmatrix}
\mathbf{R}_{11} & \mathbf{R}_{12}\\
\mathbf{R}_{21} & \mathbf{R}_{22}
\end{bmatrix}
\begin{bmatrix}
\mathbf{b}_{2}\\
\mathbf{b}_{3}
\end{bmatrix},\label{eq:QR-ser}
\end{equation}
and our goal is to characterize $\mathbf{S}$ such that
\begin{equation}
\begin{bmatrix}
\mathbf{b}_{1}\\
\mathbf{a}_{3}
\end{bmatrix}=
\mathbf{S}
\begin{bmatrix}
\mathbf{a}_{1}\\
\mathbf{b}_{3}
\end{bmatrix}.
\end{equation}
To this end, we first use \eqref{eq:QR-ser} to express $\mathbf{a}_{2}$ and $\mathbf{b}_{2}$ as functions of $\mathbf{a}_{1}$ and $\mathbf{b}_{3}$.
In detail, from \eqref{eq:QR-ser} we obtain
\begin{align}
\mathbf{a}_{2}
&=\mathbf{R}_{11}\mathbf{b}_{2}+\mathbf{R}_{12}\mathbf{b}_{3}\\
&=\mathbf{R}_{11}\mathbf{Q}_{21}\mathbf{a}_{1}+\mathbf{R}_{11}\mathbf{Q}_{22}\mathbf{a}_{2}+\mathbf{R}_{12}\mathbf{b}_{3},
\end{align}
giving
\begin{equation}
\mathbf{a}_{2}
=\left(\mathbf{I}_N-\mathbf{R}_{11}\mathbf{Q}_{22}\right)^{-1}\left(\mathbf{R}_{11}\mathbf{Q}_{21}\mathbf{a}_{1}+\mathbf{R}_{12}\mathbf{b}_{3}\right),\label{eq:a2}
\end{equation}
and we obtain
\begin{align}
\mathbf{b}_{2}
&=\mathbf{Q}_{21}\mathbf{a}_{1}+\mathbf{Q}_{22}\mathbf{a}_{2}\\
&=\mathbf{Q}_{21}\mathbf{a}_{1}+\mathbf{Q}_{22}\mathbf{R}_{11}\mathbf{b}_{2}+\mathbf{Q}_{22}\mathbf{R}_{12}\mathbf{b}_{3},
\end{align}
giving
\begin{equation}
\mathbf{b}_{2}
=\left(\mathbf{I}_N-\mathbf{Q}_{22}\mathbf{R}_{11}\right)^{-1}\left(\mathbf{Q}_{21}\mathbf{a}_{1}+\mathbf{Q}_{22}\mathbf{R}_{12}\mathbf{b}_{3}\right).\label{eq:b2}
\end{equation}
Then, by substituting \eqref{eq:a2} into $\mathbf{b}_{1}=\mathbf{Q}_{11}\mathbf{a}_{1}+\mathbf{Q}_{12}\mathbf{a}_{2}$, which is derived from \eqref{eq:QR-ser}, we obtain
\begin{multline}
\mathbf{b}_{1}=\mathbf{Q}_{11}\mathbf{a}_{1}
+\mathbf{Q}_{12}\left(\mathbf{I}_N-\mathbf{R}_{11}\mathbf{Q}_{22}\right)^{-1}\mathbf{R}_{11}\mathbf{Q}_{21}\mathbf{a}_{1}\\
+\mathbf{Q}_{12}\left(\mathbf{I}_N-\mathbf{R}_{11}\mathbf{Q}_{22}\right)^{-1}\mathbf{R}_{12}\mathbf{b}_{3},
\end{multline}
and by substituting \eqref{eq:b2} into $\mathbf{a}_{3}=\mathbf{R}_{21}\mathbf{b}_{2}+\mathbf{R}_{22}\mathbf{b}_{3}$, we obtain
\begin{multline}
\mathbf{a}_{3}=
 \mathbf{R}_{21}\left(\mathbf{I}_N-\mathbf{Q}_{22}\mathbf{R}_{11}\right)^{-1}\mathbf{Q}_{21}\mathbf{a}_{1}\\
+\mathbf{R}_{21}\left(\mathbf{I}_N-\mathbf{Q}_{22}\mathbf{R}_{11}\right)^{-1}\mathbf{Q}_{22}\mathbf{R}_{12}\mathbf{b}_{3}
+\mathbf{R}_{22}\mathbf{b}_{3},
\end{multline}
proving the proposition.

\subsection{Proof of Proposition~\ref{pro:par}}

Denote as $\mathbf{a}_1\in\mathbb{C}^{N\times 1}$ and $\mathbf{b}_1\in\mathbb{C}^{N\times 1}$ the incident and reflected waves on the first $N$ ports of the first network, respectively;
       as $\mathbf{a}_2\in\mathbb{C}^{N\times 1}$ and $\mathbf{b}_2\in\mathbb{C}^{N\times 1}$ the incident and reflected waves on the last $N$ ports of the first network, respectively;
       as $\mathbf{a}_3\in\mathbb{C}^{M\times 1}$ and $\mathbf{b}_3\in\mathbb{C}^{M\times 1}$ the incident and reflected waves on the first $M$ ports of the second network, respectively;
   and as $\mathbf{a}_4\in\mathbb{C}^{M\times 1}$ and $\mathbf{b}_4\in\mathbb{C}^{M\times 1}$ the incident and reflected waves on the last $M$ ports of the second network, respectively.
Thus, by the definition of scattering matrix \cite[Chapter 4]{poz12}, we have
\begin{equation}
\begin{bmatrix}
\mathbf{b}_{1}\\
\mathbf{b}_{2}
\end{bmatrix}=
\begin{bmatrix}
\mathbf{Q}_{11} & \mathbf{Q}_{12}\\
\mathbf{Q}_{21} & \mathbf{Q}_{22}
\end{bmatrix}
\begin{bmatrix}
\mathbf{a}_{1}\\
\mathbf{a}_{2}
\end{bmatrix},\;
\begin{bmatrix}
\mathbf{b}_{3}\\
\mathbf{b}_{4}
\end{bmatrix}=
\begin{bmatrix}
\mathbf{R}_{11} & \mathbf{R}_{12}\\
\mathbf{R}_{21} & \mathbf{R}_{22}
\end{bmatrix}
\begin{bmatrix}
\mathbf{a}_{3}\\
\mathbf{a}_{4}
\end{bmatrix},\label{eq:QR-par}
\end{equation}
and our goal is to characterize $\mathbf{S}$ such that
\begin{equation}
\begin{bmatrix}
\mathbf{b}_{1}\\
\mathbf{b}_{3}\\
\mathbf{b}_{2}\\
\mathbf{b}_{4}
\end{bmatrix}=
\mathbf{S}
\begin{bmatrix}
\mathbf{a}_{1}\\
\mathbf{a}_{3}\\
\mathbf{a}_{2}\\
\mathbf{a}_{4}
\end{bmatrix}.
\end{equation}
It is straightforward to recognize that \eqref{eq:QR-par} is fulfilled when
\begin{equation}
\mathbf{S}=
\begin{bmatrix}
\mathbf{Q}_{11} & \mathbf{0}_{N\times M} & \mathbf{Q}_{12} & \mathbf{0}_{N\times M}\\
\mathbf{0}_{M\times N} & \mathbf{R}_{11} & \mathbf{0}_{M\times N} & \mathbf{R}_{12}\\
\mathbf{Q}_{21} & \mathbf{0}_{N\times M} & \mathbf{Q}_{22} & \mathbf{0}_{N\times M}\\
\mathbf{0}_{M\times N} & \mathbf{R}_{21} & \mathbf{0}_{M\times N} & \mathbf{R}_{22}
\end{bmatrix},
\end{equation}
which proves the proposition.

\subsection{Proof of Proposition~\ref{pro:dft}}

This proof is divided into two parts.
First, we show that $\mathbf{F}_{N}$ can be decomposed as $\mathbf{F}_{N}=\mathbf{P}_{L}\mathbf{D}_{L}\cdots\mathbf{P}_{2}\mathbf{D}_{2}\mathbf{P}_{1}\mathbf{D}_{1}\mathbf{P}_{0}$, where $\mathbf{P}_{\ell}$ are permutation matrices and $\mathbf{D}_{\ell}$ are block diagonal matrices with blocks having size $2\times2$.
Second, we show that the matrices $\mathbf{P}_{\ell}$ and $\mathbf{D}_{\ell}$ are expressed as claimed by Proposition~\ref{pro:dft}.

For the first part, we apply the decomposition of $\mathbf{F}_{N}$ in Lemma~\ref{lem:fft}.
We begin by noticing the equivalence
\begin{equation}
\frac{1}{\sqrt{2}}
\begin{bmatrix}
\mathbf{I}_{N/2} & \mathbf{\Omega}_{N/2}\\
\mathbf{I}_{N/2} & -\mathbf{\Omega}_{N/2}
\end{bmatrix}
=\tilde{\mathbf{P}}_N\tilde{\mathbf{D}}_N\tilde{\mathbf{P}}_N^T,\label{eq:equiv1}
\end{equation}
where $\tilde{\mathbf{D}}_N\in\mathbb{C}^{N\times N}$ is a block diagonal matrix
\begin{equation}
\tilde{\mathbf{D}}_{N}=\text{diag}\left(\tilde{\mathbf{D}}_{N,1},\tilde{\mathbf{D}}_{N,2},\ldots,\tilde{\mathbf{D}}_{N,N/2}\right),
\end{equation}
with $\tilde{\mathbf{D}}_{N,c}\in\mathbb{C}^{2\times 2}$ given by
\begin{equation}
\tilde{\mathbf{D}}_{N,c}=\frac{1}{\sqrt{2}}
\begin{bmatrix}
1 & e^{-j\frac{2\pi(c-1)}{N}}\\
1 & -e^{-j\frac{2\pi(c-1)}{N}}
\end{bmatrix},\label{eq:equiv2}
\end{equation}
for $c=1,\ldots,N/2$.
As three illustrative examples of the block diagonal matrix $\tilde{\mathbf{D}}_{N}$, the matrices $\tilde{\mathbf{D}}_2$, $\tilde{\mathbf{D}}_4$, and $\tilde{\mathbf{D}}_8$ are
\begin{equation}
\tilde{\mathbf{D}}_2=\frac{1}{\sqrt{2}}
\begin{bmatrix}
1 &  1\\
1 & -1
\end{bmatrix},\;
\tilde{\mathbf{D}}_4=\frac{1}{\sqrt{2}}
\begin{bmatrix}
1 &  1 &   &   \\
1 & -1 &   &   \\
  &    & 1 & -j\\
  &    & 1 &  j
\end{bmatrix},
\end{equation}
\begin{equation}
\tilde{\mathbf{D}}_8=\frac{1}{\sqrt{2}}
\begin{bmatrix}
1 &  1 &   &                     &   &    &   &                     \\
1 & -1 &   &                     &   &    &   &                     \\
  &    & 1 & e^{-j\frac{\pi}{4}} &   &    &   &                     \\
  &    & 1 & e^{j\frac{3\pi}{4}} &   &    &   &                     \\
  &    &   &                     & 1 & -j &   &                     \\
  &    &   &                     & 1 &  j &   &                     \\
  &    &   &                     &   &    & 1 & e^{-j\frac{3\pi}{4}}\\
  &    &   &                     &   &    & 1 & e^{j\frac{\pi}{4}}
\end{bmatrix}.
\end{equation}
The equivalence in \eqref{eq:equiv1} can be easily verified by recalling the definitions of the odd-even permutation matrix $\tilde{\mathbf{P}}_N$ and the diagonal matrix $\mathbf{\Omega}_{N/2}$, used in Lemma~\ref{lem:fft}.
By substituting \eqref{eq:equiv1} into the decomposition of $\mathbf{F}_N$ in \eqref{eq:fft}, we obtain $\mathbf{F}_N$ as
\begin{equation}
\mathbf{F}_N=
\tilde{\mathbf{P}}_N\tilde{\mathbf{D}}_N\tilde{\mathbf{P}}_N^T
\begin{bmatrix}
\mathbf{F}_{N/2} & \\
 & \mathbf{F}_{N/2}
\end{bmatrix}
\tilde{\mathbf{P}}_N,
\end{equation}
which can be more conveniently rewritten as
\begin{equation}
\mathbf{F}_N=
\tilde{\mathbf{P}}_N\tilde{\mathbf{D}}_N\tilde{\mathbf{P}}_N^T
\left(\mathbf{I}_{2}\otimes\mathbf{F}_{N/2}\right)
\tilde{\mathbf{P}}_N,\label{eq:FN-proof0}
\end{equation}
by introducing a Kronecker product.

By using \eqref{eq:FN-proof0}, we can prove the first part of this proof by induction on $L=\log_2(N)$.
For the base case $L=1$, i.e., for the $2\times2$ \gls{dft} matrix, it is trivially proved since
\begin{equation}
\mathbf{F}_2=\underbrace{\mathbf{I}_2}_{\mathbf{P}_{1}^{(1)}}\underbrace{\mathbf{F}_2}_{\mathbf{D}_{1}^{(1)}}\underbrace{\mathbf{I}_2}_{\mathbf{P}_{0}^{(1)}},\label{eq:F2-proof}
\end{equation}
where $\mathbf{P}_{0}^{(1)}$ and $\mathbf{P}_{1}^{(1)}$ are trivially identity matrices (and hence permutation matrices) and $\mathbf{D}_{1}^{(1)}$ is block diagonal with blocks having size $2\times2$.
The superscript $^{(1)}$ indicates that these matrices are valid for $L=1$.
As the induction step, we prove that if the first part of this proof holds for the case $L-1$, i.e., for the $2^{L-1}\times2^{L-1}$ \gls{dft} matrix, then it also holds for the case $L$, i.e., for the $2^L\times 2^L$ \gls{dft} matrix.
By the induction hypothesis, we have
\begin{multline}
\mathbf{F}_{2^{L-1}}=\mathbf{P}_{L-1}^{(L-1)}\mathbf{D}_{L-1}^{(L-1)}\times\cdots\\
\times\mathbf{P}_{2}^{(L-1)}\mathbf{D}_{2}^{(L-1)}\mathbf{P}_{1}^{(L-1)}\mathbf{D}_{1}^{(L-1)}\mathbf{P}_{0}^{(L-1)},\label{eq:FL-1-proof}
\end{multline}
where $\mathbf{P}_{\ell}^{(L-1)}$ are permutation matrices and $\mathbf{D}_{\ell}^{(L-1)}$ are block diagonal with $2\times2$ blocks.
The superscript $^{(L-1)}$ indicates that these matrices are valid for $L-1$, and is needed to avoid ambiguity between, for instance, $\mathbf{P}_{0}^{(1)}$ and $\mathbf{P}_{0}^{(L-1)}$, which are different.
The induction step can be proved by using \eqref{eq:FN-proof0} as $\mathbf{F}_{2^L}=
\tilde{\mathbf{P}}_{2^L}\tilde{\mathbf{D}}_{2^L}\tilde{\mathbf{P}}_{2^L}^T(\mathbf{I}_{2}\otimes\mathbf{F}_{2^{L-1}})\tilde{\mathbf{P}}_{2^L}$ and expressing $\mathbf{F}_{2^{L-1}}$ therein with \eqref{eq:FL-1-proof}, which gives
\begin{multline}
\mathbf{F}_{2^L}=
\underbrace{\tilde{\mathbf{P}}_{2^L}}_{\mathbf{P}_{L}^{(L)}}\underbrace{\tilde{\mathbf{D}}_{2^L}}_{\mathbf{D}_{L}^{(L)}}\underbrace{\tilde{\mathbf{P}}_{2^L}^T
\left(\mathbf{I}_{2}\otimes\mathbf{P}_{L-1}^{(L-1)}\right)}_{\mathbf{P}_{L-1}^{(L)}}\underbrace{\left(\mathbf{I}_{2}\otimes\mathbf{D}_{L-1}^{(L-1)}\right)}_{\mathbf{D}_{L-1}^{(L)}}
\times\cdots\\
\times
\underbrace{\left(\mathbf{I}_{2}\otimes\mathbf{P}_{2}^{(L-1)}\right)}_{\mathbf{P}_{2}^{(L)}}
\underbrace{\left(\mathbf{I}_{2}\otimes\mathbf{D}_{2}^{(L-1)}\right)}_{\mathbf{D}_{2}^{(L)}}
\underbrace{\left(\mathbf{I}_{2}\otimes\mathbf{P}_{1}^{(L-1)}\right)}_{\mathbf{P}_{1}^{(L)}}\\
\times\underbrace{\left(\mathbf{I}_{2}\otimes\mathbf{D}_{1}^{(L-1)}\right)}_{\mathbf{D}_{1}^{(L)}}
\underbrace{\left(\mathbf{I}_{2}\otimes\mathbf{P}_{0}^{(L-1)}\right)
\tilde{\mathbf{P}}_{2^L}}_{\mathbf{P}_{0}^{(L)}},\label{eq:FL-proof}
\end{multline}
where we have exploited the mixed-product property of the Kronecker product.
In \eqref{eq:FL-proof}, we have highlighted the decomposition into blocks $\mathbf{P}_{\ell}^{(L)}$ and $\mathbf{D}_{\ell}^{(L)}$, which are permutation matrices and block diagonal matrices with $2\times2$ blocks, respectively, because of the induction hypothesis, concluding the first part of this proof.

For the second part of this proof, we need to show that the matrices $\mathbf{P}_{\ell}^{(L)}$, for $\ell=0,\ldots,L$, and $\mathbf{D}_{\ell}^{(L)}$, for $\ell=1,\ldots,L$, are expressed as stated in Proposition~\ref{pro:dft}.
We first derive the expressions of the matrices $\mathbf{P}_\ell^{(L)}$, for $\ell=0,\ldots,L$, starting from $\mathbf{P}_{0}^{(L)}$.
From \eqref{eq:FL-proof}, we notice that $\mathbf{P}_{0}^{(L)}$ is given by $\mathbf{P}_{0}^{(L)}=(\mathbf{I}_{2}\otimes\mathbf{P}_{0}^{(L-1)})
\tilde{\mathbf{P}}_{2^{L}}$, depending on $\mathbf{P}_{0}^{(L-1)}$, which depends in turn on $\mathbf{P}_{0}^{(L-2)}$, and so forth.
To obtain a closed-form of $\mathbf{P}_{0}^{(L)}$, we can rewrite it as
\begin{align}
\mathbf{P}_{0}^{(L)}
&=\left(\mathbf{I}_{2}\otimes\mathbf{P}_{0}^{(L-1)}\right)
\tilde{\mathbf{P}}_{2^{L}}\\
&=\left(\mathbf{I}_{2}\otimes\left(\left(\mathbf{I}_{2}\otimes\mathbf{P}_{0}^{(L-2)}\right)
\tilde{\mathbf{P}}_{2^{L-1}}\right)\right)
\tilde{\mathbf{P}}_{2^{L}}\label{eq:P0-1}\\
&=\left(\mathbf{I}_{2}\otimes\left(\mathbf{I}_{2}\otimes\mathbf{P}_{0}^{(L-2)}\right)\right)\left(\mathbf{I}_{2}\otimes\tilde{\mathbf{P}}_{2^{L-1}}\right)\tilde{\mathbf{P}}_{2^{L}}\label{eq:P0-2}\\
&=\left(\mathbf{I}_{2^2}\otimes\mathbf{P}_{0}^{(L-2)}\right)\left(\mathbf{I}_{2}\otimes\tilde{\mathbf{P}}_{2^{L-1}}\right)\tilde{\mathbf{P}}_{2^{L}},\label{eq:P0-3}
\end{align}
where in \eqref{eq:P0-1} we used that $\mathbf{P}_{0}^{(L-1)}=(\mathbf{I}_{2}\otimes\mathbf{P}_{0}^{(L-2)})
\tilde{\mathbf{P}}_{2^{L-1}}$, in \eqref{eq:P0-2} we have exploited the mixed-product property of the Kronecker product, and in \eqref{eq:P0-3} its associative property and noticed that $\mathbf{I}_{2}\otimes\mathbf{I}_{2}=\mathbf{I}_{2^{2}}$.
After these steps, the term $\mathbf{I}_{2}\otimes\mathbf{P}_{0}^{(L-1)}$ has been expressed as a function of $\mathbf{P}_{0}^{(L-2)}$, and the term $\mathbf{I}_{2^2}\otimes\mathbf{P}_{0}^{(L-2)}$ appeared.
Therefore, the steps in \eqref{eq:P0-1}-\eqref{eq:P0-3} can be repeated to express $\mathbf{I}_{2^2}\otimes\mathbf{P}_{0}^{(L-2)}$ as a function of $\mathbf{P}_{0}^{(L-3)}$, and so forth.
This recursion repeats until the term $\mathbf{I}_{2^{L-1}}\otimes\mathbf{P}_{0}^{(1)}$ appears, in which $\mathbf{P}_{0}^{(1)}=\mathbf{I}_2$ or equivalently $\mathbf{P}_{0}^{(1)}=\tilde{\mathbf{P}}_{2}$.
At the end of the recursion, we obtain
\begin{multline}
\mathbf{P}_{0}^{(L)}=
\left(\mathbf{I}_{2^{L-1}}\otimes\tilde{\mathbf{P}}_{2}\right)
\left(\mathbf{I}_{2^{L-2}}\otimes\tilde{\mathbf{P}}_{2^2}\right)\times\cdots\\
\times
\left(\mathbf{I}_{2^2}\otimes\tilde{\mathbf{P}}_{2^{L-2}}\right)
\left(\mathbf{I}_{2}\otimes\tilde{\mathbf{P}}_{2^{L-1}}\right)\tilde{\mathbf{P}}_{2^{L}},
\end{multline}
which can be more compactly written as \eqref{eq:P0-dft} in Proposition~\ref{pro:dft}.

For the matrices $\mathbf{P}_\ell^{(L)}$, for $\ell=1,\ldots,L$, we notice from \eqref{eq:FL-proof} that $\mathbf{P}_{L}^{(L)}=\tilde{\mathbf{P}}_{2^L}$, $\mathbf{P}_{L-1}^{(L)}=\tilde{\mathbf{P}}_{2^L}^T(\mathbf{I}_{2}\otimes\tilde{\mathbf{P}}_{2^{L-1}})$, and that $\mathbf{P}_{\ell}^{(L)}=\mathbf{I}_{2}\otimes\mathbf{P}_{\ell}^{(L-1)}$, for $\ell=1,\ldots,L-2$.
Hence, $\mathbf{P}_{\ell}^{(L)}$ can be obtained recursively and it is given by $L-\ell-1$ nested Kronecker products, until the matrix $\mathbf{P}_{\ell}^{(\ell+1)}=\tilde{\mathbf{P}}_{2^{\ell+1}}^T(\mathbf{I}_{2}\otimes\tilde{\mathbf{P}}_{2^{\ell}})$ is reached.
Specifically, we have
\begin{align}
\mathbf{P}_{\ell}^{(L)}
&=\left(\mathbf{I}_{2}\otimes\left(\mathbf{I}_{2}\otimes\cdots\left(\mathbf{I}_{2}\otimes\mathbf{P}_{\ell}^{(\ell+1)}\right)\cdots\right)\right)\\
&=\left(\mathbf{I}_{2}\otimes\cdots\otimes\mathbf{I}_{2}\right)\otimes\mathbf{P}_{\ell}^{(\ell+1)}\label{eq:Pl-1}\\
&=\mathbf{I}_{2^{L-\ell-1}}\otimes\mathbf{P}_{\ell}^{(\ell+1)}\label{eq:Pl-2}\\
&=\mathbf{I}_{2^{L-\ell-1}}\otimes\left(\tilde{\mathbf{P}}_{2^{\ell+1}}^T\left(\mathbf{I}_{2}\otimes\tilde{\mathbf{P}}_{2^{\ell}}\right)\right)\label{eq:Pl-3}\\
&=\left(\mathbf{I}_{2^{L-\ell-1}}\otimes\tilde{\mathbf{P}}_{2^{\ell+1}}^T\right)\left(\mathbf{I}_{2^{L-\ell-1}}\otimes\left(\mathbf{I}_{2}\otimes\tilde{\mathbf{P}}_{2^{\ell}}\right)\right)\label{eq:Pl-4}\\
&=\left(\mathbf{I}_{2^{L-\ell-1}}\otimes\tilde{\mathbf{P}}_{2^{\ell+1}}^T\right)\left(\mathbf{I}_{2^{L-\ell}}\otimes\tilde{\mathbf{P}}_{2^{\ell}}\right),\label{eq:Pl-5}
\end{align}
for $\ell=1,\ldots,L-2$, where in \eqref{eq:Pl-1} we have used the associative property of the Kronecker product to reorganize the parentheses, in \eqref{eq:Pl-2} we have noticed that the Kronecker product of $L-\ell-1$ identity matrices of size $2\times2$ is an identity matrix of size $2^{L-\ell-1}\times2^{L-\ell-1}$, in \eqref{eq:Pl-3} we have used $\mathbf{P}_{\ell}^{(\ell+1)}=\tilde{\mathbf{P}}_{2^{\ell+1}}^T(\mathbf{I}_{2}\otimes\tilde{\mathbf{P}}_{2^{\ell}})$, in \eqref{eq:Pl-4} we have applied the mixed-product property of the Kronecker product, and \eqref{eq:Pl-5} its associative property and noticed that $\mathbf{I}_{2^{L-\ell-1}}\otimes\mathbf{I}_{2}=\mathbf{I}_{2^{L-\ell}}$.
We have therefore proved that the matrices $\mathbf{P}_{\ell}^{(L)}$, for $\ell=0,\ldots,L$, are expressed as claimed in Proposition~\ref{pro:dft}.

We now derive the expressions of the matrices $\mathbf{D}_\ell^{(L)}$, for $\ell=1,\ldots,L$.
From \eqref{eq:FL-proof}, we notice that $\mathbf{D}_{L}^{(L)}=\tilde{\mathbf{D}}_{2^L}$, and that $\mathbf{D}_{\ell}^{(L)}=\mathbf{I}_{2}\otimes\mathbf{D}_{\ell}^{(L-1)}$, for $\ell=1,\ldots,L-1$.
Hence, $\mathbf{D}_{\ell}^{(L)}$ can be obtained recursively and it is given by $L-\ell$ nested Kronecker products, until the matrix $\mathbf{D}_{\ell}^{(\ell)}=\tilde{\mathbf{D}}_{2^\ell}$ is reached.
Specifically, we have
\begin{align}
\mathbf{D}_{\ell}^{(L)}
&=\left(\mathbf{I}_{2}\otimes\left(\mathbf{I}_{2}\otimes\cdots\left(\mathbf{I}_{2}\otimes\tilde{\mathbf{D}}_{2^\ell}\right)\cdots\right)\right)\\
&=\left(\mathbf{I}_{2}\otimes\cdots\otimes\mathbf{I}_{2}\right)\otimes\tilde{\mathbf{D}}_{2^\ell}\label{eq:Dl-1}\\
&=\mathbf{I}_{2^{L-\ell}}\otimes\tilde{\mathbf{D}}_{2^\ell},\label{eq:Dl-2}
\end{align}
for $\ell=1,\ldots,L$, where in \eqref{eq:Dl-1} we have used the associative property of the Kronecker product, and in \eqref{eq:Dl-2} we have noticed that the Kronecker product of $L-\ell$ identity matrices of size $2\times2$ is an identity matrix of size $2^{L-\ell}\times2^{L-\ell}$.
We have therefore proved that also the matrices $\mathbf{D}_{\ell}^{(L)}$, for $\ell=1,\ldots,L$, are expressed as claimed in Proposition~\ref{pro:dft}.

\subsection{Comparison Between DFT and Butler Matrix}

\begin{figure}[t]
\centering
\includegraphics[width=0.241\textwidth]{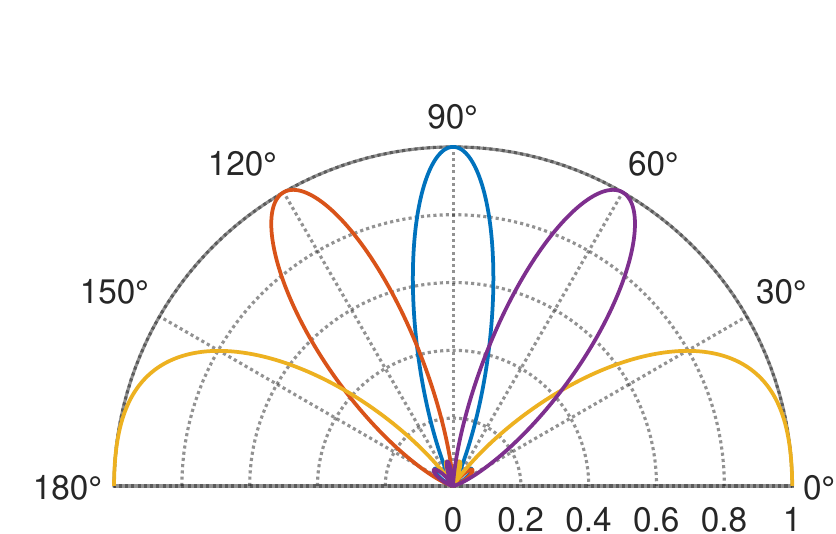}
\includegraphics[width=0.241\textwidth]{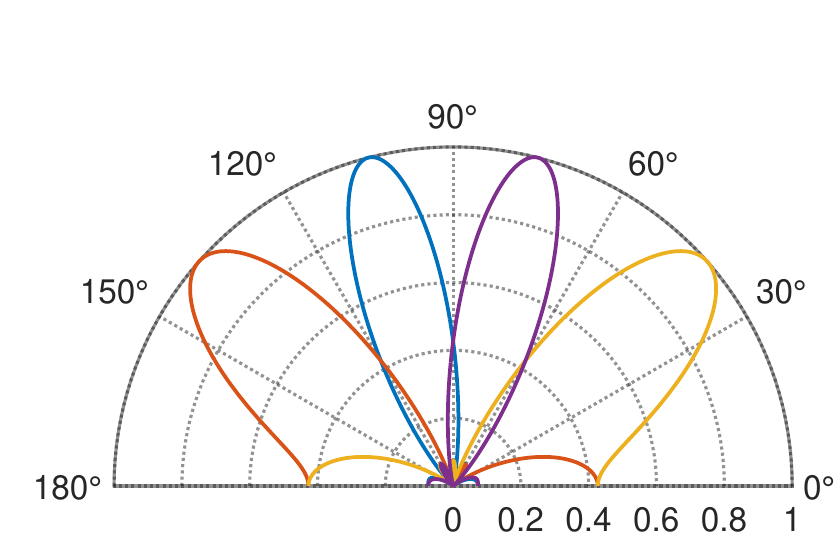}
\caption{Radiation pattern of the columns of (a) the DFT matrix $\mathbf{F}_4$ and (b) the odd-time odd-frequency DFT matrix $\mathbf{B}_4$.}
\label{fig:F-vs-B}
\end{figure}

Following \cite{moo64} and references therein, the Butler matrix is unitary and its entries have magnitude one and phase fulfilling a specific property.
The phases of the entries in each column of the $N\times N$ Butler matrix are spaced by $\psi_n=(2n-1)\pi/N$, for $n=1,\ldots,N$.
In this way, the $n$th input signal is steered toward the angle $\alpha_n$ such that $\psi_n=2\pi d\cos(\alpha_n)/\lambda$, where $d$ is the antenna spacing and $\lambda$ is the wavelength.
It is clear that the \gls{dft} matrix $\mathbf{F}_N$ in \eqref{eq:dft} does not fulfill this property.
Instead, a matrix that fulfills this property (up to a column permutation) is $\mathbf{B}_N\in\mathbb{C}^{N\times N}$ defined as
\begin{equation}
\left[\mathbf{B}_N\right]_{i,k}=\frac{1}{\sqrt{N}}\omega^{(i-1/2)(k-1/2)},\label{eq:Butler}
\end{equation}
for $i,k=1,\ldots,N$, where $\omega=e^{-j2\pi/N}$, which is the so-called odd-time odd-frequency \gls{dft} \cite{bon76a}.
Therefore, the columns of $\mathbf{B}$ have the same radiation pattern as those of the Butler matrix.
For comparison, Fig.~\ref{fig:F-vs-B} shows the radiation patterns of the columns of $\mathbf{F}_4$ and $\mathbf{B}_4$, when the four antennas are spaced $\lambda/2$.
Notably, while the transforms associated with $\mathbf{F}_N$ and $\mathbf{B}_N$ are different, they are both part of a family referred to as generalized \gls{dft} \cite{bon76b}.

\subsection{Proof of Proposition~\ref{pro:hada}}

As done to prove Proposition~\ref{pro:dft}, this proof is also divided into two parts.
First, we show that $\mathbf{H}_{N}$ can be decomposed as $\mathbf{H}_{N}=\mathbf{P}_{L}\mathbf{D}_{L}\cdots\mathbf{P}_{2}\mathbf{D}_{2}\mathbf{P}_{1}\mathbf{D}_{1}\mathbf{P}_{0}$, where $\mathbf{P}_{\ell}$ are permutation matrices and $\mathbf{D}_{\ell}$ are block diagonal matrices with $2\times2$ blocks.
Second, we show that the matrices $\mathbf{P}_{\ell}$ and $\mathbf{D}_{\ell}$ are expressed as claimed by Proposition~\ref{pro:hada}.

For the first part, we apply the definition of Hadamard matrix and some properties of the Kronecker product to rewrite it as
\begin{align}
\mathbf{H}_{N}
&=\left(\mathbf{I}_{2}\mathbf{H}_{2}\right)\otimes\left(\mathbf{H}_{N/2}\mathbf{I}_{N/2}\right)\\
&=\left(\mathbf{H}_{2}\otimes\mathbf{I}_{N/2}\right)\left(\mathbf{I}_{2}\otimes\mathbf{H}_{N/2}\right)\\
&=\tilde{\mathbf{P}}_{N}\left(\mathbf{I}_{N/2}\otimes\mathbf{H}_{2}\right)\tilde{\mathbf{P}}_{N}^T\left(\mathbf{I}_{2}\otimes\mathbf{H}_{N/2}\right),\label{eq:HN-proof0}
\end{align}
and we prove the first part by induction on $L=\log_2(N)$.
The base case $L=1$ is trivially proved since
\begin{equation}
\mathbf{H}_2=\underbrace{\mathbf{I}_2}_{\mathbf{P}_{1}^{(1)}}\underbrace{\mathbf{H}_2}_{\mathbf{D}_{1}^{(1)}}\underbrace{\mathbf{I}_2}_{\mathbf{P}_{0}^{(1)}},\label{eq:H2-proof}
\end{equation}
where $\mathbf{P}_{0}^{(1)}$ and $\mathbf{P}_{1}^{(1)}$ are permutation matrices and $\mathbf{D}_{1}^{(1)}$ is block diagonal with blocks having size $2\times2$.
As the induction step, we prove that if the first part of this proof holds for the case $L-1$, then it also holds for the case $L$.
By the induction hypothesis, we have
\begin{multline}
\mathbf{H}_{2^{L-1}}=\mathbf{P}_{L-1}^{(L-1)}\mathbf{D}_{L-1}^{(L-1)}\times\cdots\\
\times\mathbf{P}_{2}^{(L-1)}\mathbf{D}_{2}^{(L-1)}\mathbf{P}_{1}^{(L-1)}\mathbf{D}_{1}^{(L-1)}\mathbf{P}_{0}^{(L-1)},\label{eq:HL-1-proof}
\end{multline}
where $\mathbf{P}_{\ell}^{(L-1)}$ are permutation matrices and $\mathbf{D}_{\ell}^{(L-1)}$ are block diagonal with $2\times2$ blocks.
The induction step can be proved by using \eqref{eq:HN-proof0} as $\mathbf{H}_{2^L}=\tilde{\mathbf{P}}_{2^L}(\mathbf{I}_{2^{L-1}}\otimes\mathbf{H}_{2})\tilde{\mathbf{P}}_{2^L}^T(\mathbf{I}_{2}\otimes\mathbf{H}_{2^{L-1}})$ and expressing $\mathbf{H}_{2^{L-1}}$ therein with \eqref{eq:HL-1-proof}, which gives
\begin{multline}
\mathbf{H}_{2^L}=
\underbrace{\tilde{\mathbf{P}}_{2^L}}_{\mathbf{P}_{L}^{(L)}}
\underbrace{\left(\mathbf{I}_{2^{L-1}}\otimes\mathbf{H}_{2}\right)}_{\mathbf{D}_{L}^{(L)}}
\underbrace{\tilde{\mathbf{P}}_{2^L}^T
\left(\mathbf{I}_{2}\otimes\mathbf{P}_{L-1}^{(L-1)}\right)}_{\mathbf{P}_{L-1}^{(L)}}\times\cdots\\
\times
\underbrace{\left(\mathbf{I}_{2}\otimes\mathbf{D}_{L-1}^{(L-1)}\right)}_{\mathbf{D}_{L-1}^{(L)}}
\underbrace{\left(\mathbf{I}_{2}\otimes\mathbf{P}_{2}^{(L-1)}\right)}_{\mathbf{P}_{2}^{(L)}}
\underbrace{\left(\mathbf{I}_{2}\otimes\mathbf{D}_{2}^{(L-1)}\right)}_{\mathbf{D}_{2}^{(L)}}\\
\times
\underbrace{\left(\mathbf{I}_{2}\otimes\mathbf{P}_{1}^{(L-1)}\right)}_{\mathbf{P}_{1}^{(L)}}
\underbrace{\left(\mathbf{I}_{2}\otimes\mathbf{D}_{1}^{(L-1)}\right)}_{\mathbf{D}_{1}^{(L)}}
\underbrace{\left(\mathbf{I}_{2}\otimes\mathbf{P}_{0}^{(L-1)}\right)}_{\mathbf{P}_{0}^{(L)}},\label{eq:HL-proof}
\end{multline}
where we have exploited the mixed-product property of the Kronecker product.
The blocks $\mathbf{P}_{\ell}^{(L)}$ and $\mathbf{D}_{\ell}^{(L)}$ in \eqref{eq:HL-proof} are permutation matrices and block diagonal matrices with $2\times2$ blocks, respectively, because of the induction hypothesis, concluding the first part of this proof.

For the second part of this proof, we need to show that the matrices $\mathbf{P}_{\ell}^{(L)}$, for $\ell=0,\ldots,L$, and $\mathbf{D}_{\ell}^{(L)}$, for $\ell=1,\ldots,L$, are expressed as stated in Proposition~\ref{pro:hada}.
For the matrices $\mathbf{P}_\ell^{(L)}$, for $\ell=0,\ldots,L$, we notice from \eqref{eq:HL-proof} that $\mathbf{P}_{L}^{(L)}=\tilde{\mathbf{P}}_{2^L}$, $\mathbf{P}_{L-1}^{(L)}=\tilde{\mathbf{P}}_{2^L}^T(\mathbf{I}_{2}\otimes\tilde{\mathbf{P}}_{2^{L-1}})$, and that $\mathbf{P}_{\ell}^{(L)}=\mathbf{I}_{2}\otimes\mathbf{P}_{\ell}^{(L-1)}$, for $\ell=0,\ldots,L-2$.
Hence, $\mathbf{P}_{\ell}^{(L)}$ can be obtained recursively and it is given by $L-\ell-1$ nested Kronecker products, until the matrix $\mathbf{P}_{\ell}^{(\ell+1)}=\tilde{\mathbf{P}}_{2^{\ell+1}}^T(\mathbf{I}_{2}\otimes\tilde{\mathbf{P}}_{2^{\ell}})$ is reached.
Specifically, we have
\begin{align}
\mathbf{P}_{\ell}^{(L)}
&=\left(\mathbf{I}_{2}\otimes\left(\mathbf{I}_{2}\otimes\cdots\left(\mathbf{I}_{2}\otimes\mathbf{P}_{\ell}^{(\ell+1)}\right)\cdots\right)\right)\\
&=\left(\mathbf{I}_{2}\otimes\cdots\otimes\mathbf{I}_{2}\right)\otimes\mathbf{P}_{\ell}^{(\ell+1)}\label{eq:Pl-1-H}\\
&=\mathbf{I}_{2^{L-\ell-1}}\otimes\mathbf{P}_{\ell}^{(\ell+1)}\label{eq:Pl-2-H}\\
&=\mathbf{I}_{2^{L-\ell-1}}\otimes\left(\tilde{\mathbf{P}}_{2^{\ell+1}}^T\left(\mathbf{I}_{2}\otimes\tilde{\mathbf{P}}_{2^{\ell}}\right)\right)\label{eq:Pl-3-H}\\
&=\left(\mathbf{I}_{2^{L-\ell-1}}\otimes\tilde{\mathbf{P}}_{2^{\ell+1}}^T\right)\left(\mathbf{I}_{2^{L-\ell-1}}\otimes\left(\mathbf{I}_{2}\otimes\tilde{\mathbf{P}}_{2^{\ell}}\right)\right)\label{eq:Pl-4-H}\\
&=\left(\mathbf{I}_{2^{L-\ell-1}}\otimes\tilde{\mathbf{P}}_{2^{\ell+1}}^T\right)\left(\mathbf{I}_{2^{L-\ell}}\otimes\tilde{\mathbf{P}}_{2^{\ell}}\right),\label{eq:Pl-5-H}
\end{align}
for $\ell=0,\ldots,L-2$, which is as claimed in Proposition~\ref{pro:hada}.

We now derive the expressions of the matrices $\mathbf{D}_\ell^{(L)}$, for $\ell=1,\ldots,L$.
From \eqref{eq:HL-proof}, we notice that $\mathbf{D}_{L}^{(L)}=\mathbf{I}_{2^{L-1}}\otimes\mathbf{H}_{2}$, and that $\mathbf{D}_{\ell}^{(L)}=\mathbf{I}_{2}\otimes\mathbf{D}_{\ell}^{(L-1)}$, for $\ell=1,\ldots,L-1$.
Hence, $\mathbf{D}_{\ell}^{(L)}$ can be obtained recursively and it is given by $L-\ell$ nested Kronecker products, until the matrix $\mathbf{D}_{\ell}^{(\ell)}=\mathbf{I}_{2^{\ell-1}}\otimes\mathbf{H}_{2}$ is reached.
Specifically, we have
\begin{align}
\mathbf{D}_{\ell}^{(L)}
&=\left(\mathbf{I}_{2}\otimes\left(\mathbf{I}_{2}\otimes\cdots\left(\mathbf{I}_{2}\otimes\left(\mathbf{I}_{2^{\ell-1}}\otimes\mathbf{H}_{2}\right)\right)\cdots\right)\right)\\
&=\left(\mathbf{I}_{2}\otimes\cdots\otimes\mathbf{I}_{2}\otimes\mathbf{I}_{2^{\ell-1}}\right)\otimes\mathbf{H}_{2}\label{eq:Dl-1-H}\\
&=\mathbf{I}_{2^{L-1}}\otimes\mathbf{H}_{2},\label{eq:Dl-2-H}
\end{align}
for $\ell=1,\ldots,L$, completing the proof of Proposition~\ref{pro:hada}.

\subsection{Proof of Proposition~\ref{pro:haar}}

As done for Propositions~\ref{pro:dft} and \ref{pro:hada}, the proof of this proposition is also divided into two parts.
First, we show that $\mathbf{W}_{N}$ can be decomposed as $\mathbf{W}_{N}=\mathbf{P}_{L}\mathbf{D}_{L}\cdots\mathbf{P}_{2}\mathbf{D}_{2}\mathbf{P}_{1}\mathbf{D}_{1}\mathbf{P}_{0}$, where $\mathbf{P}_{\ell}$ are permutation matrices and $\mathbf{D}_{\ell}$ are block diagonal matrices with $2\times2$ blocks.
Second, we show that the matrices $\mathbf{P}_{\ell}$ and $\mathbf{D}_{\ell}$ are expressed as claimed by Proposition~\ref{pro:haar}.

For the first part, we depart from the recursive definition of Haar matrix, and equivalently rewrite it as
\begin{align}
\mathbf{W}_N
&=\frac{1}{\sqrt{2}}
\begin{bmatrix}\mathbf{W}_{N/2} & \mathbf{0}\\\mathbf{0} & \mathbf{I}_{N/2}\end{bmatrix}
\begin{bmatrix}\mathbf{I}_{N/2} & \mathbf{I}_{N/2}\\\mathbf{I}_{N/2} & -\mathbf{I}_{N/2}\end{bmatrix}\\
&=\text{diag}\left(\mathbf{W}_{N/2},\mathbf{I}_{N/2}\right)\left(\mathbf{W}_{2}\otimes\mathbf{I}_{N/2}\right)\\
&=\text{diag}\left(\mathbf{W}_{N/2},\mathbf{I}_{N/2}\right)\tilde{\mathbf{P}}_{N}\left(\mathbf{I}_{N/2}\otimes\mathbf{W}_{2}\right)\tilde{\mathbf{P}}_{N}^T,\label{eq:WN-proof0}
\end{align}
and we prove the first part by induction on $L=\log_2(N)$.
The base case $L=1$ is trivially proved since
\begin{equation}
\mathbf{W}_2=
\underbrace{\mathbf{I}_2}_{\mathbf{P}_{1}^{(1)}}
\underbrace{\mathbf{W}_2}_{\mathbf{D}_{1}^{(1)}}
\underbrace{\mathbf{I}_2}_{\mathbf{P}_{0}^{(1)}},\label{eq:W2-proof}
\end{equation}
where $\mathbf{P}_{0}^{(1)}$ and $\mathbf{P}_{1}^{(1)}$ are permutation matrices and $\mathbf{D}_{1}^{(1)}$ is block diagonal with blocks having size $2\times2$.
As the induction step, we prove that if the first part of this proof holds for the case $L-1$, then it also holds for the case $L$.
By the induction hypothesis, we have
\begin{multline}
\mathbf{W}_{2^{L-1}}=\mathbf{P}_{L-1}^{(L-1)}\mathbf{D}_{L-1}^{(L-1)}\times\cdots\\
\times\mathbf{P}_{2}^{(L-1)}\mathbf{D}_{2}^{(L-1)}\mathbf{P}_{1}^{(L-1)}\mathbf{D}_{1}^{(L-1)}\mathbf{P}_{0}^{(L-1)},\label{eq:WL-1-proof}
\end{multline}
where $\mathbf{P}_{\ell}^{(L-1)}$ are permutation matrices and $\mathbf{D}_{\ell}^{(L-1)}$ are block diagonal with $2\times2$ blocks.
The induction step can be proved by using \eqref{eq:WN-proof0} as $\mathbf{W}_{2^L}=\text{diag}(\mathbf{W}_{2^{L-1}},\mathbf{I}_{2^{L-1}})\tilde{\mathbf{P}}_{2^L}(\mathbf{I}_{2^{L-1}}\otimes\mathbf{W}_{2})\tilde{\mathbf{P}}_{2^L}^T$ and expressing $\mathbf{H}_{2^{L-1}}$ therein with \eqref{eq:WL-1-proof}, which gives
\begin{multline}
\mathbf{W}_{2^L}=
\underbrace{\text{diag}\left(\mathbf{P}_{L-1}^{(L-1)},\mathbf{I}_{2^{L-1}}\right)}_{\mathbf{P}_{L}^{(L)}}
\underbrace{\text{diag}\left(\mathbf{D}_{L-1}^{(L-1)},\mathbf{I}_{2^{L-1}}\right)}_{\mathbf{D}_{L}^{(L)}}\\
\times\underbrace{\text{diag}\left(\mathbf{P}_{L-2}^{(L-1)},\mathbf{I}_{2^{L-1}}\right)}_{\mathbf{P}_{L-1}^{(L)}}
\underbrace{\text{diag}\left(\mathbf{D}_{L-2}^{(L-1)},\mathbf{I}_{2^{L-1}}\right)}_{\mathbf{D}_{L-1}^{(L)}}\times\cdots\\
\times
\underbrace{\text{diag}\left(\mathbf{P}_{1}^{(L-1)},\mathbf{I}_{2^{L-1}}\right)}_{\mathbf{P}_{2}^{(L)}}
\underbrace{\text{diag}\left(\mathbf{D}_{1}^{(L-1)},\mathbf{I}_{2^{L-1}}\right)}_{\mathbf{D}_{2}^{(L)}}\\
\times
\underbrace{\text{diag}\left(\mathbf{P}_{0}^{(L-1)},\mathbf{I}_{2^{L-1}}\right)\tilde{\mathbf{P}}_{2^L}}_{\mathbf{P}_{1}^{(L)}}
\underbrace{\left(\mathbf{I}_{2^{L-1}}\otimes\mathbf{W}_{2}\right)}_{\mathbf{D}_{1}^{(L)}}
\underbrace{\tilde{\mathbf{P}}_{2^L}^T}_{\mathbf{P}_{0}^{(L)}},\label{eq:WL-proof}
\end{multline}
where we have exploited the mixed-product property of the Kronecker product.
The blocks $\mathbf{P}_{\ell}^{(L)}$ and $\mathbf{D}_{\ell}^{(L)}$ in \eqref{eq:WL-proof} are permutation matrices and block diagonal matrices with blocks being $2\times2$, respectively, because of the induction hypothesis, concluding the first part of this proof.

For the second part of this proof, we need to show that the matrices $\mathbf{P}_{\ell}^{(L)}$, for $\ell=0,\ldots,L$, and $\mathbf{D}_{\ell}^{(L)}$, for $\ell=1,\ldots,L$, are expressed as stated in Proposition~\ref{pro:haar}.
For the matrices $\mathbf{P}_\ell^{(L)}$, for $\ell=0,\ldots,L$, we notice from \eqref{eq:WL-proof} that $\mathbf{P}_{0}^{(L)}=\tilde{\mathbf{P}}_{2^L}^T$,
\begin{equation}
\mathbf{P}_{1}^{(L)}=
\begin{cases}
\text{diag}\left(\tilde{\mathbf{P}}_{2^{L-1}}^T,\mathbf{I}_{2^{L-1}}\right)\tilde{\mathbf{P}}_{2^L}&\text{ if }L>1\\
\mathbf{I}_{2} &\text{ if }L=1
\end{cases},\label{eq:P1-W}
\end{equation}
and that $\mathbf{P}_{\ell}^{(L)}=\text{diag}(\mathbf{P}_{\ell-1}^{(L-1)},\mathbf{I}_{2^{L-1}})$, for $\ell=2,\ldots,L$.
Hence, $\mathbf{P}_{\ell}^{(L)}$ can be obtained recursively and it is given by $\ell-1$ nested $\text{diag}(\cdot)$ operations, until the matrix $\mathbf{P}_{1}^{(L-\ell+1)}$ is reached.
Specifically, we have
\begin{align}
\mathbf{P}_{\ell}^{(L)}
&=\text{diag}\left(\mathbf{P}_{\ell-1}^{(L-1)},\mathbf{I}_{2^{L-1}}\right)\label{eq:Pl-1-W}\\
&=\text{diag}\left(\text{diag}\left(\mathbf{P}_{\ell-2}^{(L-2)},\mathbf{I}_{2^{L-2}}\right),\mathbf{I}_{2^{L-1}}\right)\label{eq:Pl-2-W}\\
&=\text{diag}\left(\text{diag}\left(\cdots\mathbf{P}_{1}^{(L-\ell+1)}\cdots,\mathbf{I}_{2^{L-2}}\right),\mathbf{I}_{2^{L-1}}\right)\label{eq:Pl-3-W}\\
&=\text{diag}\left(\mathbf{P}_{1}^{(L-\ell+1)},\mathbf{I}_{2^{L}-2^{L-\ell+1}}\right),\label{eq:Pl-4-W}
\end{align}
for $\ell=2,\ldots,L$.
By expressing $\mathbf{P}_{1}^{(L-\ell+1)}$ by \eqref{eq:P1-W}, we obtain
\begin{equation}
\mathbf{P}_{\ell}^{(L)}=\text{diag}\left(\text{diag}\left(\tilde{\mathbf{P}}_{2^{L-\ell}}^T,\mathbf{I}_{2^{L-\ell}}\right)\tilde{\mathbf{P}}_{2^{L-\ell+1}},\mathbf{I}_{2^{L}-2^{L-\ell+1}}\right),
\end{equation}
for $\ell=2,\ldots,L-1$, and $\mathbf{P}_{L}^{(L)}=\mathbf{I}_{2^{L}}$, as given by Proposition~\ref{pro:haar}.

We now derive the expressions of the matrices $\mathbf{D}_\ell^{(L)}$, for $\ell=1,\ldots,L$.
From \eqref{eq:WL-proof}, we notice that $\mathbf{D}_{1}^{(L)}=\mathbf{I}_{2^{L-1}}\otimes\mathbf{W}_{2}$, and $\mathbf{D}_{\ell}^{(L)}=\text{diag}(\mathbf{D}_{\ell-1}^{(L-1)},\mathbf{I}_{2^{L-1}})$, for $\ell=2,\ldots,L$.
Hence, $\mathbf{D}_{\ell}^{(L)}$ can be obtained recursively by $\ell-1$ nested $\text{diag}(\cdot)$ operations, until the matrix $\mathbf{D}_{1}^{(L-\ell+1)}=\mathbf{I}_{2^{L-\ell}}\otimes\mathbf{W}_{2}$ is reached.
In detail, we have
\begin{align}
\mathbf{D}_{\ell}^{(L)}
&=\text{diag}\left(\mathbf{D}_{\ell-1}^{(L-1)},\mathbf{I}_{2^{L-1}}\right)\\
&=\text{diag}\left(\text{diag}\left(\mathbf{D}_{\ell-2}^{(L-2)},\mathbf{I}_{2^{L-2}}\right),\mathbf{I}_{2^{L-1}}\right)\\
&=\text{diag}\left(\text{diag}\left(\cdots\mathbf{D}_{1}^{(L-\ell+1)}\cdots,\mathbf{I}_{2^{L-2}}\right),\mathbf{I}_{2^{L-1}}\right)\\
&=\text{diag}\left(\mathbf{D}_{1}^{(L-\ell+1)},\mathbf{I}_{2^{L}-2^{L-\ell+1}}\right)\\
&=\text{diag}\left(\mathbf{I}_{2^{L-\ell}}\otimes\mathbf{W}_{2},\mathbf{I}_{2^{L}-2^{L-\ell+1}}\right),
\end{align}
for $\ell=2,\ldots,L$, precisely as given by Proposition~\ref{pro:haar}.

\bibliographystyle{IEEEtran}
\bibliography{IEEEabrv,main}

\begin{thebibliography}{10}
\providecommand{\url}[1]{#1}
\csname url@samestyle\endcsname
\providecommand{\newblock}{\relax}
\providecommand{\bibinfo}[2]{#2}
\providecommand{\BIBentrySTDinterwordspacing}{\spaceskip=0pt\relax}
\providecommand{\BIBentryALTinterwordstretchfactor}{4}
\providecommand{\BIBentryALTinterwordspacing}{\spaceskip=\fontdimen2\font plus
\BIBentryALTinterwordstretchfactor\fontdimen3\font minus \fontdimen4\font\relax}
\providecommand{\BIBforeignlanguage}[2]{{%
\expandafter\ifx\csname l@#1\endcsname\relax
\typeout{** WARNING: IEEEtran.bst: No hyphenation pattern has been}%
\typeout{** loaded for the language `#1'. Using the pattern for}%
\typeout{** the default language instead.}%
\else
\language=\csname l@#1\endcsname
\fi
#2}}
\providecommand{\BIBdecl}{\relax}
\BIBdecl

\bibitem{rab02}
J.~M. Rabaey, A.~Chandrakasan, and B.~Nikolic, \emph{Digital integrated circuits}.\hskip 1em plus 0.5em minus 0.4em\relax Prentice hall Englewood Cliffs, 2002, vol.~2.

\bibitem{cal13}
C.~Caloz, S.~Gupta, Q.~Zhang, and B.~Nikfal, ``Analog signal processing: A possible alternative or complement to dominantly digital radio schemes,'' \emph{IEEE Microw. Mag.}, vol.~14, no.~6, pp. 87--103, 2013.

\bibitem{sil14}
A.~Silva, F.~Monticone, G.~Castaldi, V.~Galdi, A.~Al{\`u}, and N.~Engheta, ``Performing mathematical operations with metamaterials,'' \emph{Science}, vol. 343, no. 6167, pp. 160--163, 2014.

\bibitem{zan21}
F.~Zangeneh-Nejad, D.~L. Sounas, A.~Al{\`u}, and R.~Fleury, ``Analogue computing with metamaterials,'' \emph{Nature Reviews Materials}, vol.~6, no.~3, pp. 207--225, 2021.

\bibitem{iel18}
D.~Ielmini and H.-S.~P. Wong, ``In-memory computing with resistive switching devices,'' \emph{Nature electronics}, vol.~1, no.~6, pp. 333--343, 2018.

\bibitem{ner25-1}
M.~Nerini and B.~Clerckx, ``Analog computing for signal processing and communications – {Part I}: Computing with microwave networks,'' \emph{IEEE Trans. Signal Process.}, vol.~73, pp. 5183--5197, 2025.

\bibitem{ner25-2}
M.~Nerini and B.~Clerckx, ``Analog computing for signal processing and communications – {Part II}: Toward gigantic {MIMO} beamforming,'' \emph{IEEE Trans. Signal Process.}, vol.~73, pp. 5198--5212, 2025.

\bibitem{moo64}
H.~Moody, ``The systematic design of the {Butler} matrix,'' \emph{IEEE Trans. Antennas Propag.}, vol.~12, no.~6, pp. 786--788, 1964.

\bibitem{yan25}
X.~Yang, O.~C. Vicente, and C.~Caloz, ``Analog {OFDM} based on real-time fourier transformation,'' \emph{arXiv preprint arXiv:2506.20287}, 2025.

\bibitem{poz12}
D.~M. Pozar, \emph{Microwave engineering}, 4th~ed.\hskip 1em plus 0.5em minus 0.4em\relax John wiley \& sons, 2012.

\bibitem{sha41}
C.~E. Shannon, ``Mathematical theory of the differential analyzer,'' \emph{Journal of Mathematics and Physics}, vol.~20, no. 1-4, pp. 337--354, 1941.

\bibitem{str09}
G.~Strang, \emph{Introduction to linear algebra}, 4th~ed.\hskip 1em plus 0.5em minus 0.4em\relax Wellesley-Cambridge Press, 2009.

\bibitem{coo65}
J.~W. Cooley and J.~W. Tukey, ``An algorithm for the machine calculation of complex {Fourier} series,'' \emph{Mathematics of computation}, vol.~19, no.~90, pp. 297--301, 1965.

\bibitem{nes68}
W.~Nester, ``The fast {Fourier} transform and the {Butler} matrix,'' \emph{IEEE Trans. Antennas Propag.}, vol.~16, no.~3, pp. 360--360, 1968.

\bibitem{she68}
J.~Shelton, ``Fast {Fourier} transforms and {Butler} matrices,'' \emph{Proc. IEEE}, vol.~56, no.~3, pp. 350--350, 1968.

\bibitem{uen81}
M.~Ueno, ``A systematic design formulation for {Butler} matrix applied {FFT} algorithm,'' \emph{IEEE Trans. Antennas Propag.}, vol.~29, no.~3, pp. 496--501, 1981.

\bibitem{an24}
J.~An, C.~Yuen, Y.~L. Guan, M.~D. Renzo, M.~Debbah, H.~V. Poor, and L.~Hanzo, ``Two-dimensional direction-of-arrival estimation using stacked intelligent metasurfaces,'' \emph{IEEE J. Sel. Areas Commun.}, vol.~42, no.~10, pp. 2786--2802, 2024.

\bibitem{fin77}
B.~J. Fino and V.~R. Algazi, ``A unified treatment of discrete fast unitary transforms,'' \emph{SIAM Journal on Computing}, vol.~6, no.~4, pp. 700--717, 1977.

\bibitem{kim09}
J.~H. Kim and W.~S. Park, ``A {Hadamard} matrix feed network for a dual-beam forming array antenna,'' \emph{IEEE Trans. Antennas Propag.}, vol.~57, no.~1, pp. 283--286, 2009.

\bibitem{gol13}
M.~Goldenbaum, H.~Boche, and S.~Stańczak, ``Harnessing interference for analog function computation in wireless sensor networks,'' \emph{IEEE Trans. Signal Process.}, vol.~61, no.~20, pp. 4893--4906, 2013.

\bibitem{yan20}
K.~Yang, T.~Jiang, Y.~Shi, and Z.~Ding, ``Federated learning via over-the-air computation,'' \emph{IEEE Trans. Wireless Commun.}, vol.~19, no.~3, pp. 2022--2035, 2020.

\bibitem{del18}
P.~del Hougne and G.~Lerosey, ``Leveraging chaos for wave-based analog computation: Demonstration with indoor wireless communication signals,'' \emph{Physical Review X}, vol.~8, no.~4, p. 041037, 2018.

\bibitem{sol22}
J.~Sol, D.~R. Smith, and P.~Del~Hougne, ``Meta-programmable analog differentiator,'' \emph{Nature Communications}, vol.~13, no.~1, p. 1713, 2022.

\bibitem{wu19}
Q.~Wu and R.~Zhang, ``Intelligent reflecting surface enhanced wireless network via joint active and passive beamforming,'' \emph{IEEE Trans. Wireless Commun.}, vol.~18, no.~11, pp. 5394--5409, 2019.

\bibitem{joy25}
A.~T. Joy, A.~Tishchenko, H.~Taghvaee, P.~Mursia, V.~Sciancalepore, and M.~Khalily, ``{RIS}-enabled {ISAC} in {6G}: Exploring the role of wave domain computing,'' \emph{IEEE Commun. Standards Mag.}, 2025.

\bibitem{an23}
J.~An, C.~Xu, D.~W.~K. Ng, G.~C. Alexandropoulos, C.~Huang, C.~Yuen, and L.~Hanzo, ``Stacked intelligent metasurfaces for efficient holographic {MIMO} communications in {6G},'' \emph{IEEE J. Sel. Areas Commun.}, vol.~41, no.~8, pp. 2380--2396, 2023.

\bibitem{soh16}
F.~Sohrabi and W.~Yu, ``Hybrid digital and analog beamforming design for large-scale antenna arrays,'' \emph{IEEE J. Sel. Top. Signal Process.}, vol.~10, no.~3, pp. 501--513, 2016.

\bibitem{zhu24}
M.~Zhu, T.-W. Kuo, and C.-T.~M. Wu, ``A reconfigurable linear {RF} analog processor for realizing microwave artificial neural network,'' \emph{IEEE Trans. Microw. Theory Tech.}, vol.~72, no.~2, pp. 1290--1301, 2024.

\bibitem{gu24}
Z.~Gu, Q.~Ma, X.~Gao, J.~W. You, and T.~J. Cui, ``Direct electromagnetic information processing with planar diffractive neural network,'' \emph{Science Advances}, vol.~10, no.~29, p. eado3937, 2024.

\bibitem{gao24}
X.~Gao, Z.~Gu, Q.~Ma, B.~J. Chen, K.-M. Shum, W.~Y. Cui, J.~W. You, T.~J. Cui, and C.~H. Chan, ``Terahertz spoof plasmonic neural network for diffractive information recognition and processing,'' \emph{Nature Communications}, vol.~15, no.~1, p. 6686, 2024.

\bibitem{kes25}
R.~Keshavarz, K.~Zelaya, N.~Shariati, and M.-A. Miri, ``Programmable circuits for analog matrix computations,'' \emph{Nature Communications}, vol.~16, no.~1, p. 8514, 2025.

\bibitem{ner25-3}
M.~Nerini and B.~Clerckx, ``Capacity of {MIMO} systems aided by microwave linear analog computers (milacs),'' \emph{arXiv preprint arXiv:2506.05983}, 2025.

\bibitem{ner25-4}
M.~Nerini and B.~Clerckx, ``{MIMO} systems aided by microwave linear analog computers: Capacity-achieving architectures with reduced circuit complexity,'' \emph{arXiv preprint arXiv:2506.15052}, 2025.

\bibitem{tza25}
D.~C. Tzarouchis, B.~Edwards, and N.~Engheta, ``Programmable wave-based analog computing machine: A metastructure that designs metastructures,'' \emph{Nature Communications}, vol.~16, no.~1, p. 908, 2025.

\bibitem{lin18}
X.~Lin, Y.~Rivenson, N.~T. Yardimci, M.~Veli, Y.~Luo, M.~Jarrahi, and A.~Ozcan, ``All-optical machine learning using diffractive deep neural networks,'' \emph{Science}, vol. 361, no. 6406, pp. 1004--1008, 2018.

\bibitem{liu22}
C.~Liu, Q.~Ma, Z.~J. Luo, Q.~R. Hong, Q.~Xiao, H.~C. Zhang, L.~Miao, W.~M. Yu, Q.~Cheng, L.~Li \emph{et~al.}, ``A programmable diffractive deep neural network based on a digital-coding metasurface array,'' \emph{Nature Electronics}, vol.~5, no.~2, pp. 113--122, 2022.

\bibitem{gao23}
X.~Gao, Q.~Ma, Z.~Gu, W.~Y. Cui, C.~Liu, J.~Zhang, and T.~J. Cui, ``Programmable surface plasmonic neural networks for microwave detection and processing,'' \emph{Nature Electronics}, vol.~6, no.~4, pp. 319--328, 2023.

\bibitem{mom23}
A.~Momeni, B.~Rahmani, M.~Mall{\'e}jac, P.~del Hougne, and R.~Fleury, ``Backpropagation-free training of deep physical neural networks,'' \emph{Science}, vol. 382, no. 6676, pp. 1297--1303, 2023.

\bibitem{gov25}
B.~Govind, M.~G. Anderson, F.~O. Wu, P.~L. McMahon, and A.~Apsel, ``An integrated microwave neural network for broadband computation and communication,'' \emph{Nature Electronics}, pp. 1--13, 2025.

\bibitem{reg89}
P.~A. Regalia and M.~K. Sanjit, ``Kronecker products, unitary matrices and signal processing applications,'' \emph{SIAM review}, vol.~31, no.~4, pp. 586--613, 1989.

\bibitem{van00}
C.~F. Van~Loan, ``The ubiquitous kronecker product,'' \emph{Journal of computational and applied mathematics}, vol. 123, no. 1-2, pp. 85--100, 2000.

\bibitem{bon76a}
G.~Bonnerot and M.~Bellanger, ``Odd-time odd-frequency discrete {Fourier} transform for symmetric real-valued series,'' \emph{Proc. IEEE}, vol.~64, no.~3, pp. 392--393, 1976.

\bibitem{bon76b}
G.~Bongiovanni, P.~Corsini, and G.~Frosini, ``One-dimensional and two-dimensional generalised discrete {Fourier} transforms,'' \emph{IEEE Trans. Acoust., Speech, Signal Process.}, vol.~24, no.~1, pp. 97--99, 1976.

\end{thebibliography}

\end{document}